%% file: main.tex
\pdfoutput=1
\documentclass[11pt]{article}
\usepackage[margin=1in]{geometry}
\usepackage{lmodern}
\usepackage{amsmath,amssymb,amsthm,mathtools}
\usepackage{booktabs}
\usepackage{array}
\usepackage{graphicx}
\usepackage{tikz}
\usepackage[round]{natbib}
\usepackage[colorlinks=true,linkcolor=blue,citecolor=blue,hypertexnames=false,bookmarksnumbered=true]{hyperref}
\usepackage{enumitem}
\hypersetup{
  pdftitle={The Anatomy and Boundary of Adaptation under Temporal Tabular Shift},
  pdfauthor={Tianyu Wang, Xi Vincent Wang, Lihui Wang, Mian Li, Zhihao Liu},
  pdfkeywords={test-time adaptation, distribution shift, tabular foundation models, partial identification, minimax rates}
}

\input{appendix/app_macros}

\newcommand{\SuppFirstMentionNote}{\footnote{Online Appendix~1 is the
  supplementary document that accompanies this arXiv submission as the
  ancillary file \texttt{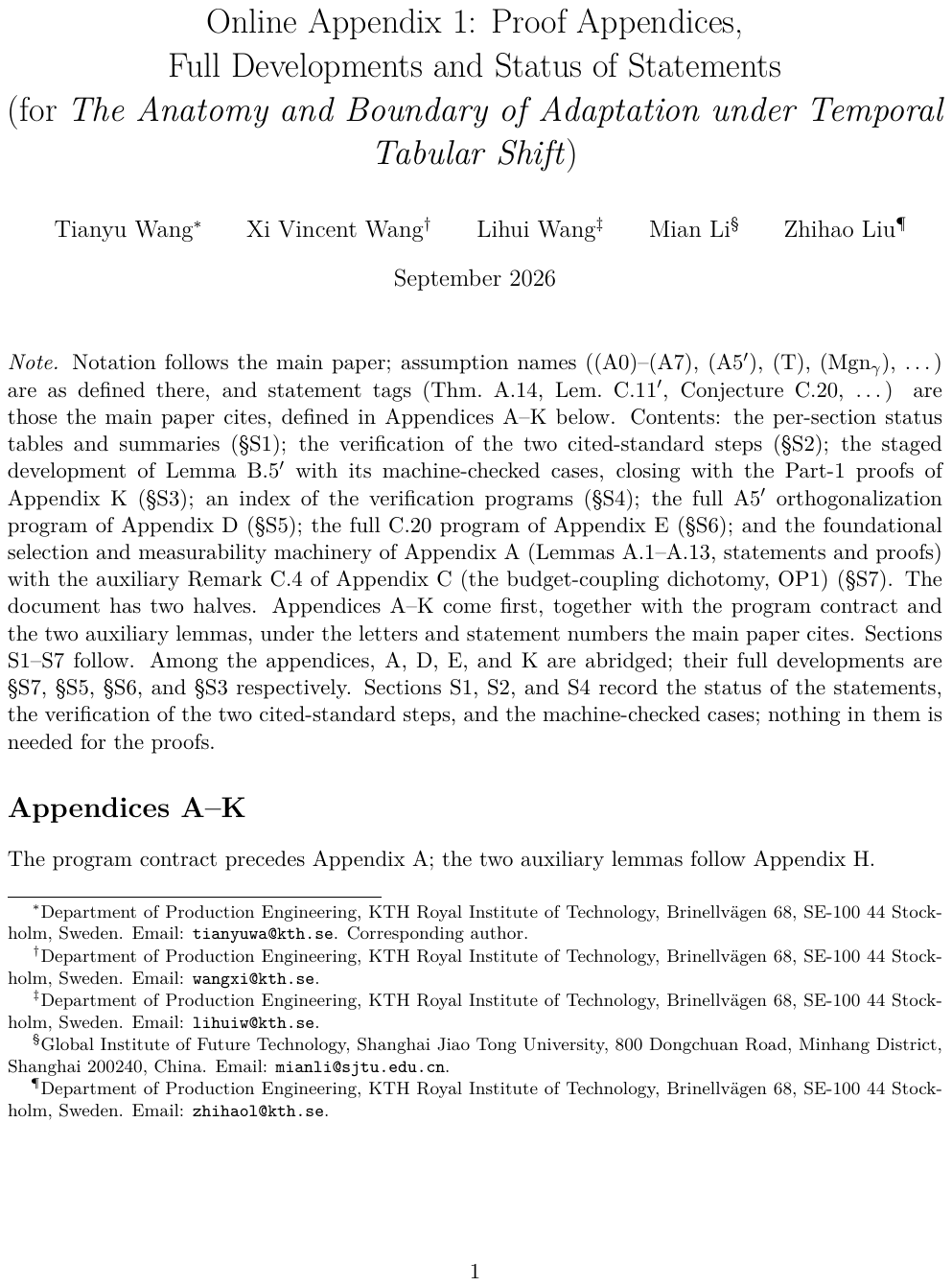}; it carries the proof
  appendices A--K under their original numbering, their full developments,
  and the status of every statement.}}

\newcommand{\ArtifactDepositNote}{are collected in a reproduction deposit
that will be released in a public archive; its citable DOI will be recorded
in a revised version of this manuscript}

\theoremstyle{plain}
\newtheorem{theorem}{Theorem}[section]
\newtheorem{proposition}[theorem]{Proposition}
\newtheorem{lemma}[theorem]{Lemma}
\newtheorem{corollary}[theorem]{Corollary}

\theoremstyle{definition}

\theoremstyle{remark}
\newtheorem{remark}[theorem]{Remark}

\title{The Anatomy and Boundary of Adaptation\\
under Temporal Tabular Shift}
\author{Tianyu Wang\thanks{Department of Production Engineering, KTH Royal
  Institute of Technology, Stockholm, Sweden.
  Email: \texttt{tianyuwa@kth.se}. Corresponding author.}
  \and Xi Vincent Wang\thanks{Department of Production Engineering, KTH Royal
  Institute of Technology, Stockholm, Sweden.
  Email: \texttt{wangxi@kth.se}.}
  \and Lihui Wang\thanks{Department of Production Engineering, KTH Royal
  Institute of Technology, Stockholm, Sweden.
  Email: \texttt{lihuiw@kth.se}.}
  \and Mian Li\thanks{Global Institute of Future Technology, Shanghai Jiao
  Tong University, Shanghai, China. Email: \texttt{mianli@sjtu.edu.cn}.}
  \and Zhihao Liu\thanks{Department of Production Engineering, KTH Royal
  Institute of Technology, Stockholm, Sweden.
  Email: \texttt{zhihaol@kth.se}.}}
\date{September 2026}

\begin{document}
\maketitle

\input{sec_abstract}

\input{sec_intro}

\input{sec_related}

\input{sec_setup}

\input{sec_anatomy}

\input{sec_boundary}

\input{sec_beyond}

\section*{Author Contributions}
Contributor roles follow the CRediT taxonomy.
\textbf{Tianyu Wang:} Conceptualization; Methodology; Formal analysis;
Software; Data curation; Investigation; Visualization; Writing -- original
draft. \textbf{Xi Vincent Wang:} Conceptualization; Funding acquisition;
Supervision; Writing -- review \& editing. \textbf{Lihui Wang:} Funding
acquisition; Supervision; Writing -- review \& editing. \textbf{Mian Li:}
Conceptualization; Methodology; Funding acquisition; Writing -- review \&
editing. \textbf{Zhihao Liu:} Software; Writing -- review \& editing.

\section*{Acknowledgments and Disclosure of Funding}
The work presented here was supported by the Swedish research centre of
eXcellence in PRoduction RESearch (XPRES) and the National Natural
Science Foundation of China (Grant No.\ 52275263). The authors declare
no competing interests.

\bibliographystyle{plainnat}
\bibliography{refs}

\clearpage
\appendix
\renewcommand{\theHsection}{appendix.\thesection}
\input{appendix/app_guide}
\setcounter{section}{1}               
\input{appendix/app_B}
\setcounter{section}{7}               
\input{appendix/app_H}
\input{appendix/app_discharges}
\setcounter{section}{9}               
\input{appendix/app_J}
\end{document}

%% file: appendix/app_macros.tex
\newcommand{\E}{\mathbb{E}}

\newcommand{\J}{\mathcal{J}}
\newcommand{\barmu}{\bar\mu}

\newcommand{\Gcov}{\Gamma_{\mathrm{cov}}}
\newcommand{\piout}{\pi_{\mathrm{out}}}
\newcommand{\epsr}{\varepsilon_r}
\newcommand{\Dcov}{D_{\mathrm{cov}}}

\newcommand{\Pmon}{\mathcal{P}_{\mathrm{mon}}}
\newcommand{\PmonL}{\mathcal{P}^L_{\mathrm{mon}}}
\newcommand{\as}{\alpha_s}

\newcommand{\Pbar}{\bar P}

\newcommand{\sgn}{\operatorname{sign}}
\DeclareMathOperator{\Var}{Var}

\DeclareMathOperator{\supp}{supp}

\newenvironment{astmt}[1]{\par\medskip\noindent\textbf{#1}\hspace{0.5em}\itshape}{\par\medskip}
\providecommand{\OAname}{Online Appendix~1}
\makeatletter
\newcommand{\taglabel}[2]{\phantomsection\def\@currentlabel{#1}\label{#2}}
\makeatother
\newenvironment{astmtrm}[1]{\par\medskip\noindent\textbf{#1}\hspace{0.5em}}{\par\medskip}

%% file: sec_abstract.tex
\begin{abstract}%
Prequential adaptation of frozen tabular foundation models under temporal
drift, with each label revealed only after prediction, helps some deployments
and harms others, yet current practice does not predict which. We study the
sources and limits of these gains. A diagnostic anatomy attributes gains to
four recurring mechanisms under a streaming protocol that removes three
optimistic biases and quantifies a fourth. Within an agnostic total-variation
drift class, the target conditional is only partially identified: its
identified-set diameter, the \emph{wall}, is irreducible from unlabeled data
uniformly in sample size. A second, orthogonal $L^2$ projection wall
quantifies what the frozen representation cannot express. Two canonical
mechanism priors collapse the first wall. Under stated nuisance-rate
conditions, the wall can be estimated from labeled historical windows at a
$\sqrt N$ rate above the margin threshold
$\gamma^\star=d_0/(2\alpha_s)$. At $\gamma=0$, the conditional lower-bound
program depends on an open affinity estimate; the positive-margin lower
branch also remains open. Semi-synthetic data illustrate the finite-sample
mechanism with calibrated exponents. Stream-level proxies on eight
industrial streams fall on the difficult side under a stated roughness
bound, while the equality case $\gamma=\gamma^\star$ remains unresolved.
\end{abstract}

%% file: sec_intro.tex
\section{Introduction}
\label{sec:intro}

A frozen tabular foundation model \citep{hollmann2025accurate,qu2025tabicl}
deployed on an industrial data stream can degrade as the world drifts. Its
operator must repeatedly decide whether to adapt (correct the model online from whatever the stream
reveals) or to freeze and wait for retraining. The empirical record for
adaptation is mixed. In our case study on the eight industrial streams of
the TabReD benchmark \citep{rubachev2024tabred}, streaming residual correction improves a frozen
model on seven (by $+0.59\%$ to $+3.94\%$ RMSE, or up to $+0.86$ AUC
points), while on the eighth it \emph{costs} $2.28$ AUC points. These outcomes
motivate a diagnostic question that method rankings alone do not answer.

Three information states recur throughout, and each headline result is
stated relative to one of them. Before a window's labels arrive, the operator
holds only the unlabeled stream and the frozen model; after each
prediction, the revealed label joins a strictly-past prefix that supports
prequential correction; and occasional fully labeled historical windows
support learning where the boundary itself lies. Every measured gain in
this paper belongs to the second state, the identification impossibility
(the identified-set diameter, or Wall~A, Section~\ref{sec:idset}) to the first, and the margin-boundary theory of Part~II to the
third; the representation deficit (Wall~B) binds in all three, with or
without labels.

Method rankings presuppose that the shift is correctable and tend to treat
failure as method immaturity. We ask a more basic question: \emph{of the
degradation a frozen model suffers under temporal tabular shift, how much is
correctable at all, from what information, and at what price?} We answer this
question in three parts.

\textbf{The anatomy} (Part~I) attributes the evaluated system's gains
to four recurring gain mechanisms (global debias, freshness, local
re-estimation in representation space, ranking recalibration) plus two
canonical identification-collapse classes, covariate and label-marginal
response. We do not assert these to be a
unique additive decomposition of every adaptation algorithm; they are a
mechanism-oriented audit of this deployment protocol. At matched sample
size on one stream, the \emph{oldest} labels degrade the corrector
($-1.1\%$) while the newest improve it ($+4.9\%$): a staleness
penalty in that case. On another stream, $93\%$ of the corrector's measured gain is
reproduced by a running mean. Such examples show why we report gains
with channel attribution and time-respecting, leak-free oracles.
Part~I therefore also quantifies four optimistic evaluation biases; two of
them (the retrieval leak and the deduplication) reverse conclusions that
the uncorrected protocol supports.

\textbf{The boundary} (Part~II) uses partial identification to explain why
the correctable residual signal measured in Part~I can be small \citep{manski2003partial}. Within one unlabeled window
the drifted conditional is identified only up to a set whose diameter
(the \emph{wall}) reduces exactly to an averaged pointwise budget. We
prove the wall irreducible: no unlabeled statistic shrinks it, uniformly
in sample size. The wall is one of \emph{two} obstructions, orthogonal in
the projection geometry rather than statistically independent. The
second is the information the frozen representation discarded, and the two
prescribe different remedies (labels versus a new embedding). For a fixed
number $K$ of labeled historical windows holding $N$ anchor labels in
total, we establish a $\sqrt N$-regular estimator for learning the wall
from that history when the margin exponent $\gamma$ exceeds
$\gamma^\star=d_0/(2\alpha_s)$ (intrinsic dimension $d_0$, drift
smoothness $\alpha_s$), under stated nuisance-rate conditions. On the other side, the lower-bound program is
conditional at $\gamma=0$ on one open affinity estimate and remains
open for $0<\gamma<\gamma^\star$. The equality case is not covered by
the current log-rate argument. On eight TabReD streams, noisy stream-level
proxies lie below the threshold under the rough-field bound
$\alpha_s\le1$. For fifteen additional benchmarks without drift labels,
we report only a sensitivity analysis under
$(\alpha_s,\gamma)=(1,1)$; these points are not data-set-specific phase
measurements.

\textbf{Beyond the boundary} (Part~III) develops a local-correction
screening rule that needs no labels beyond those already consumed by the
gated corrector. Retrospectively, the rule detects our harm case
(\textsc{homecredit-default}, Section~\ref{sec:anatomy-numbers}). This part
also gives an attribution procedure that splits residual error into Wall-A,
Wall-B, and correctable shares, and it derives the rate at which anchors
reduce the wall. The latter turns an unspecified need for more labels into a
labeling requirement.

Our contributions follow the three parts of that question.
\begin{enumerate}
\item \textbf{What is correctable.} A diagnostic anatomy of prequential
tabular adaptation: four measured gain mechanisms plus two
identification-collapse classes, with per-stream attribution on eight
industrial streams and a mechanistic harm case
(Section~\ref{sec:anatomy}). Its evidence layer is a measurement methodology that quantifies four
optimistic biases in streaming evaluation. The protocol removes
future-neighbor oracle leakage (up to $4.5\times$, with a sign
flip on one stream), temporal-twin noise floors, and duplicate-deflated intrinsic
dimension. The fourth bias arises when a safety gate monitors one loss while
harm accrues in another. We show that this mismatch is structural rather
than tunable and measure its cost
(Sections~\ref{sec:methodology} and~\ref{sec:certificate}).
\item \textbf{What limits the remainder.} A partial-identification
boundary: the exact reduction identity, the irreducibility theorem,
the collapse classes, and the two-wall attribution (geometrically
orthogonal components) with a deployable label-free wall estimand
(Section~\ref{sec:wallA}); and a margin-indexed price for learning the
wall from labeled history: a proved $\sqrt N$ sufficient branch above
$\gamma^\star$, a conditional $\gamma=0$ lower-bound program below it,
and an open positive-margin lower branch. Semi-synthetic data illustrate
the sign-recovery mechanism with calibrated exponents; a
proxy-and-sensitivity study covers twenty-three data sets
(Sections~\ref{sec:dichotomy}--\ref{sec:phasesim}).
\item \textbf{What to do about it.} Deployment diagnostics: the
local-correction screening rule, evaluated on all eight streams with a
one-sided error profile; the degradation-attribution procedure; and the
price of labels (Section~\ref{sec:beyond}).
\end{enumerate}

Section~\ref{sec:related} situates this agenda in the five literatures it
touches (empirical test-time adaptation, tabular foundation models,
distribution-shift theory, representation-level limits, and drift
monitoring) and in the two it
borrows its formal tools from, partial identification and semiparametric
efficiency theory. Section~\ref{sec:setup} then fixes the setting. The
three parts occupy Sections~\ref{sec:anatomy}--\ref{sec:methodology}
(Part~I), Section~\ref{sec:wallA} (Part~II), and Section~\ref{sec:beyond}
(Part~III).

Empirically, adaptation can help, harm, or reduce to a simple correction.
The theory identifies which limitations follow from the stated drift and
representation assumptions. For each claim, we state both the supporting
evidence and the assumptions it requires.

%% file: sec_related.tex
\providecommand{\SuppFirstMentionNote}{}

\section{Related Work}
\label{sec:related}

\noindent\emph{Test-Time Adaptation Practice.} The modern test-time-adaptation (TTA)
literature adapts a deployed model from the unlabeled test stream itself.
That label-free constraint delimits the wall theory of
Sections~\ref{sec:idset}--\ref{sec:wallB} rather than our measured system,
which is prequential and consumes each label after its prediction
(Section~\ref{sec:setup}).
\citet{sun2020ttt} update the model on each test input through an auxiliary
self-supervised task before predicting; \citet{wang2021tent} minimize
prediction entropy, updating only normalization statistics and channel-wise
affine parameters; \citet{wang2022cotta} extend the setting to continually
changing targets, countering the error accumulation and forgetting that
self-training induces with averaged pseudo-labels and stochastic weight
restoration. These studies are mostly empirical and vision-centered, using
corruption and style-shift benchmarks. On a large controlled benchmark,
\citet{zhao2023pitfalls} show that reported gains are sensitive to
hyperparameter and model selection, which is difficult without labels, and
that well-configured methods still fail on some shift types.

Direct tabular TTA methods address the feature and shift structure that
vision methods miss. TabLog learns adaptable logical
rules \citep{ren2024tablog}, while AdapTable combines uncertainty calibration
with target-label-distribution adjustment \citep{kim2024adaptable}. With a
small labeled target sample, \citet{zeng2024llmtta} study representation
choice under tabular $Y\mid X$ shift. These methods provide algorithmic
baselines but do not ask which parts of conditional drift are determined by the time-respecting information
available to a frozen deployment. Section~\ref{sec:channels} supplies the
diagnostic taxonomy. Section~\ref{sec:wallA} separates the limits of
identification from those of the frozen representation. The attribution
procedure in Section~\ref{sec:attribution-procedure} then distinguishes
method immaturity, channel exhaustion, and unidentifiability, which require
different remedies.

\noindent\emph{Tabular Foundation Models.} Frozen in-context tabular learners
are the setting our theory takes as given. TabPFN \citep{hollmann2023tabpfn}
is a prior-fitted transformer that classifies small tables in a single
forward pass, approximating Bayesian inference under a synthetic prior. Its
foundation-model successor \citep{hollmann2025accurate} outperforms tuned
baselines on small-to-medium tables, and TabICL \citep{qu2025tabicl} scales
tabular in-context learning to hundreds of thousands of rows. These models are distributed as frozen artifacts, and their context is the
only adaptation interface. This is the fixed-embedding regime studied here
(Section~\ref{sec:setup}). Our measurements use the frozen backbone of a companion system built on
this model family \citep{wang2026beyond} (Section~\ref{sec:setup}). For evaluation, TabReD \citep{rubachev2024tabred} provides industrial
streams with temporal splits on which method rankings reorder.
\citet{cai2025limits} document the gap between i.i.d.\ and temporal
evaluation for deep tabular methods. Drift-Resilient
TabPFN \citep{helli2024drift} incorporates a drift prior into pretraining.
This is complementary to our account, which takes whatever frozen model is
deployed as given and asks what any post-hoc procedure can still recover
from it.

\noindent\emph{Distribution-Shift Theory: Bounds versus Identification.}
Classical domain-adaptation theory controls target risk through computable
divergences: \citet{bendavid2010theory} bound the target error by the
source error plus a divergence between the two covariate laws that finite
unlabeled samples estimate, plus the error of the best joint hypothesis;
\citet{mansour2009domain} generalize the divergence to arbitrary loss
classes via the discrepancy distance. Covariate-shift correction operates
under the complementary assumption that the conditional is invariant, so
\citet{shimodaira2000covariate} reweights the likelihood by the density
ratio of the covariate laws, and \citet{sugiyama2007covariate} carry the
same weighting into model selection. The taxonomy of data set shift
\citep{quinonero2009dataset,morenotorres2012unifying} names the case both
lines exclude: concept shift, where $p(y\mid x)$ itself moves. This machinery does not answer our question. Its label-free quantities are
functionals of the covariate laws, whereas concept drift appears only in
terms that unlabeled data do not determine: the joint-error term in the
bounds or the invariance assumption behind reweighting. A bound conditions
on that term. A deployment decision instead requires knowing whether the
observable data determine the drifted conditional. \citet{bendavid2010impossibility} prove impossibility theorems for
unsupervised domain adaptation, and \citet{bendavid2012hardness} quantify
the hardness without target labels. Those results do not measure the
remaining ambiguity. Section~\ref{sec:wallA} does: within the agnostic
drift class, no unlabeled statistic shrinks the set of conditionals
consistent with the data, uniformly in sample size
(Theorem~\ref{thm:irred}). We therefore treat identification as the primary
question. The \emph{wall} is the diameter of what the label-free
observables leave undetermined, rather than another divergence to estimate.

\noindent\emph{Partial Identification and Label Shift.} Partial
identification in econometrics provides the relevant framework.
\citet{manski2003partial} develops the program of reporting parameter
values consistent with observable data and stated assumptions without
forcing point identification; \citet{molinari2020partial} surveys its
development. Our wall adapts a Manski-style identified set to frozen
representations under temporal shift. The estimand is the set's diameter,
a function on covariate space (Section~\ref{sec:wallA}), and the observations
arrive in streaming windows rather than in a single survey. Where
\citet{kong2022partial} impose a latent causal model whose changing
components are themselves only partially identifiable, and use that
partial identification to determine the target joint distribution, we retain the
set of drifted conditionals and study its diameter. The label-shift
literature is the collapse case.
\citet{lipton2018bbse} recover the shifted label marginal from a black-box
predictor's confusion matrix, and \citet{garg2020unified} unify the moment
and likelihood variants of that estimation. The two drift classes studied
here recover these identification mechanisms: covariate shift determines the
conditional by assumption, while label shift identifies its
finite-dimensional mixture weights. We use them as canonical collapse
examples, not as an exhaustive characterization of every structural prior
that could identify drift.

\noindent\emph{Representation-Level Limits.} Our second obstruction is the
information the frozen representation itself discarded, formalized as
Wall~B in Section~\ref{sec:wallB}. It instantiates the
information-theoretic lower bounds of \citet{zhao2019learning} (who show
that enforcing invariant representations under shifted label marginals
forces a floor on joint error) and the support/invertibility analysis of
\citet{johansson2019support}, as a per-window, label-free-deployable
sufficiency deficit. Previous analyses of adaptation failure can conflate information lost by
the representation with information that no representation could determine.
Section~\ref{sec:wallB} separates these quantities exactly
(Theorem~\ref{thm:wallB}). The attribution procedure in
Section~\ref{sec:attribution-procedure} uses that separation to distinguish
when labels are needed from when the embedding must change.

\noindent\emph{Semiparametric Efficiency and Non-Smooth Functionals.} The
margin-boundary analysis in Section~\ref{sec:dichotomy} draws on two
literatures. The proved
regular branch ($\sqrt N$ rates for
the wall as a functional of nuisance conditionals) builds on
efficient-influence-function and debiased-machine-learning machinery
\citep{chernozhukov2018dml,kennedy2022review}, with the classical
semiparametric backdrop of \citet{bickel1993efficient} and \citet{vandervaart1998}.
The conditional lower-bound program draws on the literature on non-smooth
functionals. \citet{lepski1999} study $L_r$ norms of a regression function,
and \citet{cailow2011} give composite-hypothesis lower bounds. The
white-noise equivalences of \citet{nussbaum1996} and \citet{brownlow1996} connect
those results to our below-threshold analysis. The margin
parameter that indexes our threshold is the Mammen--Tsybakov margin
condition \citep{tsybakov2004aggregation,audibert2007fast} in a new role.
There it governs classification rates; here it governs whether the
\emph{wall} (an absolute-value functional of drift) is smooth or
non-smooth at the relevant scale. The margin-indexed threshold at
$\gamma^\star=d_0/(2\alpha_s)$ (the threshold of
Section~\ref{sec:dichotomy}, combining the margin exponent with the
drift smoothness $\alpha_s$ and the intrinsic dimension $d_0$ of the
embedded data) and the conditional coherent-adversary construction for
multi-window pooling appear to be new. The latter remains conditional on
the affinity estimate identified in Appendix~K of Online
Appendix~1 (gap (2c), Section~\ref{sec:dichotomy}).\SuppFirstMentionNote{}

\noindent\emph{Drift Monitoring and Evaluation Practice.} Our diagnostics build on
the monitoring literature surveyed by \citet{gama2014survey}, whose
detectors and adaptation strategies largely presume that labels arrive
promptly enough to score the stream. Our local-correction screening rule
(Section~\ref{sec:certificate}) is instead near-label-free by construction:
it asks for no labels beyond those the corrector it gates already
consumes, and the irreducibility result (Theorem~\ref{thm:irred}) tells the
operator which of its readings are measurements and which restate the
assumed drift budget. Label-free
risk monitors for test-time adaptation \citep{schirmer2025monitoring} are
loss-denominated, as are the game-theoretic sequential tests they can be
built on \citep{shafer2021betting,ramdas2023gametheoretic}. Our currency-mismatch result in
Section~\ref{sec:gates} concerns safety gates scored in a loss the
deployment does not rank by. It identifies what such monitors can detect and documents a failure in
practice. Conformal prediction under covariate shift
\citep{tibshirani2019conformal} gives distribution-free coverage when only
the covariate law moves and the likelihood ratio is estimable. That guarantee depends on conditional invariance, which our setting does not
assume. Adaptive conformal inference \citep{gibbs2021adaptive} drops the
invariance assumption and tracks arbitrary drift online, but its target
is marginal coverage of prediction sets, not the attribution and
identification questions asked here. The local-correction screening rule instead accounts for
conditional drift through the drift budget. Our future-neighbor leak (Section~\ref{sec:leak}) is the
retrieval-oracle analogue of the look-ahead biases long managed in finance
by purged cross-validation \citep{lopezdeprado2018advances}. Across the eight measured streams,
$49.5$--$53.2\%$ of the batch oracle's retrieved neighbors come from the future, and removing
them can reverse the sign of a signal ceiling. The four-bias protocol of
Section~\ref{sec:methodology} applies the same safeguards to noise floors,
intrinsic dimension, and safety gating. The prequential evaluation
discipline itself is classical \citep{dawid1984prequential}. Online
learning with expert advice \citep{cesabianchi2006prediction} operates on
the same revealed-label stream and bounds regret against a comparator
class without distributional assumptions; a regret bound says how well a
corrector tracks its best comparator, not whether the information needed
to correct exists at all. The latter is the identification question of
Section~\ref{sec:wallA}.

%% file: sec_setup.tex
\providecommand{\SuppFirstMentionNote}{}

\section{Setting, Protocol, and Measurement Conventions}
\label{sec:setup}

\noindent\emph{Model and Representation.}
Covariates live in a Polish space $\mathcal X$ and labels in $\mathcal Y$,
binary $\{0,1\}$ for classification (the primary case for our identification
results) or real-valued for regression. A \emph{frozen} pretrained tabular
foundation model supplies a representation $\phi:\mathcal X\to\mathbb R^{D}$
(ambient width $D=192$; the intrinsic dimension $d$ of (A3), estimated as
$d_0$, is far smaller) and a readout $h_0$, jointly defining the model's source
conditional $p_0(\cdot\mid x)=h_0(\phi(x))$ with mean $\eta_0(x)$. Neither
$\phi$ nor $h_0$ is retrainable at deployment, so every deployed prediction is
$\mathcal F_\phi$-measurable, $\mathcal F_\phi=\sigma(\phi)$, a constraint with consequences
(Theorem~\ref{thm:borrowed}).

\noindent\emph{Stream and Prequential Protocol.}
Deployment data arrive as a temporally ordered stream
$x_1,x_2,\dots$. The model predicts each $x_t$, the label $y_t$ is revealed
afterwards, and any corrector may use only the strictly-past labeled prefix.
All adaptation gains in this paper are computed under this prequential
discipline \citep{dawid1984prequential}. A \emph{window} $W$ is a contiguous segment of the stream with
covariate law $p_W$ and (unobserved) conditional $p_W(\cdot\mid x)$.

\noindent\emph{Observables and the Drift Budget.}
Within a window, the label-free information is: (O1) the unlabeled draws
from $p_W$, hence any functional of the covariate law, including density
ratios against the source; (O2) the frozen pair $(\phi,h_0)$, hence
$p_0(\cdot\mid x)$ pointwise; and optionally (O3) cross-window structure
such as recurring regimes. The sole assumption linking the drifted
conditional to the model is a total-variation \emph{budget}
$\beta:\mathcal X\to[0,1]$,
$p_W(\cdot\mid x)\in B_\beta(x)
 =\{q:\ D_{\mathrm{TV}}(q,\,p_0(\cdot\mid x))\le\beta(x)\}$,
with no coupling across covariates and no assumed drift mechanism
(the \emph{agnostic class}; Part~II examines the consequences of this
choice).

\noindent\emph{Data Sets.}
The empirical analysis uses the eight industrial temporal streams of TabReD
\citep{rubachev2024tabred} ($n_{\mathrm{test}}=4.6$k--$60$k points each; five
regression, three classification, the latter with positive rates from
$2.2\%$ to $36\%$), extended in Section~\ref{sec:map} by fifteen shift
benchmarks across census, fraud, network-security, medical, credit, and bike-sharing
domains. Deployed gains are taken from a companion system study by four of the
present authors \citep{wang2026beyond}. That system is a frozen
in-context tabular foundation model whose predictions are adjusted by an
embedding-space $k$-nearest-neighbor residual corrector over a past-only
buffer, with recency weighting, global debias, and per-region calibration.
The present paper treats that system as a \emph{measurement instrument},
not as a contribution of its own. Part~I uses its recorded predictions for a
retrospective channel audit. The measurements below are computed on
eight per-stream slices of that deployment, one per stream, each
holding the frozen predictions, embeddings, and residuals in temporal
order; we call these the \emph{probe streams}, and a corrector applied to
them a \emph{probe corrector}.

\subsection{Standing Assumptions}
\label{sec:assumptions}

The boundary theorems of Part~II quote their hypotheses by name. We collect
the names here, one line each, so that the semi-formal statements of
Section~\ref{sec:wallA} can be read without leaving
the main text. In brief, the multi-window setting posits windows
arriving in sequence, occasional windows carrying small anchor sets
labeled missing completely at random (MCAR), a stationary drift response tying the budget to an
observable novelty index, and regularity for the rest. The formal
versions (with the measurability conventions,
the failure modes each assumption excludes, and the counterexamples that
fix their scope) are given in the program contract
(Online Appendix~1, ``the formal program'') and in the
standing-assumption blocks of Appendix~\ref{app:B} and
Appendices~A and~C of Online Appendix~1. The single-window results (Sections~\ref{sec:idset}%
--\ref{sec:irred}) use the regularity conditions (S1)--(S7) together with
the mechanism prior (B). The conditions require a Polish covariate space,
Borel kernels and budget, i.i.d.\ unlabeled draws, and randomized Borel
estimators; prior (B) states that the budget is the \emph{only} constraint
on the drifted conditional, with no cross-covariate coupling. The multi-window learning-the-wall results
(Section~\ref{sec:dichotomy}) additionally assume, writing $r_j$ for window
$j$'s scalar novelty (density-ratio) index, $\bar b_j=|\bar\eta_j-\eta_0|$,
with $\bar\eta_j=\mathbb E[\eta_j\mid\mathcal F_\phi]$ (Section~\ref{sec:wallB}),
for its $\phi$-reachable drift magnitude, and $\J$ for the covered ratio
interval (the subscript $J$ in $\bar\eta_J$ below indexes the new window
$K{+}1$, Section~\ref{sec:dichotomy}):

\begin{itemize}[nosep,leftmargin=3.2em]
\item[(A0)] \emph{Windows and anchors.} $K$ windows; window $j$ carries
  $m_j$ unlabeled i.i.d.\ draws from $p_j$ and $k_j\ge k_{\min}$ labeled
  anchors drawn i.i.d.\ from $p_j$ (MCAR anchor selection), $y\mid x\sim\mathrm{Bern}(\eta_j(x))$,
  independent across windows and anchors.
\item[(A1)] \emph{Stationarity of the response curve.} A single
  $\Psi\in\Pmon$ with $\E[\bar b_j(x)\mid r_j(x)=r]=\Psi(r)$ for every
  window.
\item[(A2)] \emph{Support/relevance (rank condition).} The anchor-weighted
  pooled novelty law $\barmu$ has Lebesgue density $\ge\mu_{\min}>0$ on
  $\J$.
\item[(A3)] \emph{Smoothness and intrinsic dimension.} Each $\eta_j$ is
  $(L,\as)$-H\"older in $\phi$-coordinates on $\supp(p_j)$, of doubling
  intrinsic dimension $\le d$ (the model parameter that the deduplicated
  empirical estimate $d_0$ of Section~\ref{sec:map} instantiates), with
  densities bounded above and below on the anchor region.
\item[(A4)] \emph{Shape.} $\Psi\in\PmonL$: nondecreasing,
  $L_\Psi$-Lipschitz, range $\le1$, and the estimator knows an upper bound
  on $L_\Psi$.
\item[(A5)] \emph{Ratio error.} Label-free ratio estimates with
  $\sup_R|\hat r_j-r_j|\le\epsr$ on an event of probability $\ge1-\delta_r$.
\item[(A5$'$)] \emph{Exogeneity of the ratio error.} The drift noise
  $\xi=\bar b-\Psi(r)$ satisfies $\E[\xi\mid r,\hat r]=0$. This is not implied by (A5);
  Remark~C.13 of Online Appendix~1 gives an errors-in-variables counterexample.
\item[(A6)] \emph{Novelty coverage for transfer.}
  $d\mu_{K+1}/d\barmu\le\Gcov$ on $\J$, where $\mu_{K+1}=\nu$ is the
  novelty law of the new window; the unprecedented-novelty mass
  $\piout=\mu_{K+1}(\J^c)$ is priced separately.
\item[(A7)] \emph{Density upper bounds.} $\mu_{K+1}$ has Lebesgue density
  $\le f_{\max}$ on an $\epsr$-neighborhood of $\partial\J$, plus
  comparability constants for the stationarity test.
\item[(T)] \emph{Transversality.} $\Pbar(\bar\eta_J(X)=\eta_0(X))=0$: the
  sign-ambiguous crossing set is null under the pooled anchor law.
\item[(Mgn$_\gamma$)] \emph{Margin.} $\Pbar(0<|\bar\eta_J-\eta_0|\le t)\lesssim t^\gamma$ for
  some $\gamma>0$, so drift mass thins polynomially near the crossing.
\end{itemize}
\noindent Of these, the MCAR selection in (A0) is the assumption most
easily violated in practice; (A1) is anchor-falsifiable, unlike (B); and
constant within-window novelty profiles reduce (A2) to identification at
atoms only.

\noindent\emph{Currencies and Measurement Conventions.}
Regression is scored in RMSE, classification deployments in AUC. Loss-%
denominated quantities (squared error, log-loss) and ranking-denominated
ones are \emph{not} interchangeable on these streams, and we track the
currency of every claim (Sections~\ref{sec:channels}, \ref{sec:gates}).
The boundary theory of Part~II is stated in the proper-loss,
conditional-mean currency throughout: identified sets, walls, and the
margin-boundary rates price conditional-mean decisions. Ranking (AUC) sits outside
that scope, and the paper treats it as a measured
counterexample rather than a covered case. Channel C4's gains (Section~\ref{sec:channels}) are
invisible to squared-error ceilings, \textsc{homecredit-default} shows
log-loss and AUC moving in opposite directions, and the
local-correction screening rule's ranking-currency false freeze on
\textsc{ecom-offers} is priced to this scope
(Section~\ref{sec:certificate}). A ranking-currency analogue of the theory
and of the screening rule is left open (Section~\ref{sec:conclusion}).
Unless stated otherwise, retrieval oracles and signal ceilings are
\emph{time-respecting}: they use either past-only data or strict
out-of-time splits (Section~\ref{sec:leak}). Noise floors are estimated
from temporally decoupled pairs (Section~\ref{sec:twins}), and intrinsic
dimensions are estimated after exact-duplicate removal
(Section~\ref{sec:dedup}). Every reported number was recomputed by a
second, independently written implementation; the two preliminary
readings that this recomputation overturned are noted where they arose
(Sections~\ref{sec:dedup} and~\ref{sec:map}).

\noindent\emph{Notation.}
The recurring symbols are collected below, with their defining locations.

\begin{center}\footnotesize
\setlength{\tabcolsep}{4pt}%
\begin{tabular}{@{}l>{\raggedright\arraybackslash}p{0.395\linewidth}l>{\raggedright\arraybackslash}p{0.315\linewidth}@{}}
\toprule
$\phi,\ h_0$ & frozen embedding and readout (Section~\ref{sec:setup}) &
  $\gamma$ & margin exponent, (Mgn$_\gamma$) \\
$\eta_0(x)$ & source conditional mean $p_0(1\mid x)$, $p_0=h_0\circ\phi$ &
  $\gamma^\star$ & threshold $d_0/(2\alpha_s)$ (Section~\ref{sec:dichotomy}) \\
$p_W,\ \eta_W$ & window covariate law (Section~\ref{sec:setup});
  $\mathbb E_{p_W}[y\mid x]$ (Section~\ref{sec:wallB}) &
  $d_0$ & deduplicated intrinsic dimension of $\phi$
  (Section~\ref{sec:dedup}) \\
$\beta(x)$ & drift budget, a prior (Section~\ref{sec:idset}) &
  $\alpha_s$ & drift-field H\"older smoothness, (A3) \\
$B_\beta(x)$ & TV-ball of radius $\beta(x)$ around $p_0(\cdot\mid x)$ &
  $a$ & rate exponent $\alpha_s/(2\alpha_s+d_0)$ \\
$I(\beta)$ & identified set: selections through $B_\beta$ &
  $K,\ k,\ N$ & windows, anchors per window, $N=Kk$
  (Section~\ref{sec:dichotomy}) \\
$D_\Delta(\beta;R)$ & its $\Delta$-diameter: the wall (Section~\ref{sec:idset}) &
  $R$ & region: a measurable subset of $\mathcal X$ \\
$w_\beta$ & post-clip width $\min(1,\eta_0{+}\beta){-}\max(0,\eta_0{-}\beta)$ &
  $\Psi(r)$ & drift response $\mathbb E[\,|\bar\eta_J-\eta_0|\mid R=r\,]$ \\
$c(\beta;R)$ & clip term, the subtractive label-free content &
  $w$ & novelty representer $d\nu/d\bar\mu$ \\
$\mathrm{WallA}$ & identification wall on $\phi$-measurable acts (Section~\ref{sec:wallB}) \\
$\mathrm{WallB}$, $\mathrm{WallB}_\rho$ & $\mathbb E[\operatorname{Var}(\eta_W\mid\mathcal F_\phi)]$; deployed at
  stated resolution $\rho$ (Section~\ref{sec:wallB}, Appendix~\ref{app:J}) &
  $\rho_\beta(x)$ & price of freeze,
  $2\max(0,\beta-|\eta_0-\tfrac12|)$ (Section~\ref{sec:idset}) \\
\bottomrule
\end{tabular}
\end{center}

Three entries require additional qualification. The bound
$\alpha_s\le1$ is an assumption about the roughness of the drift field, not
a derived property; Section~\ref{sec:map} and Appendix~\ref{app:J} discuss
what happens when it fails. Per-data-set values of $d_0$ are tabulated in
Section~\ref{sec:map}; the length of a test stream is written
$n_{\mathrm{test}}$, which is not the anchor count $N=Kk$. The letter $R$ is
overloaded: it denotes a region
throughout, but in Section~\ref{sec:dichotomy}, Appendix~\ref{app:H}, and
Appendices~G and~K of Online Appendix~1 it is also the scalar novelty index with law $\nu$, in
Section~\ref{sec:dichotomy} the symbols $R_2,R_{\mathrm{prod}},\dots$ are
remainder terms, and $R^2$ is always the coefficient of determination. The
letters $K$, $k$, $m$, and $D$ are also overloaded: $K$ is the window count
throughout (Appendix~\ref{app:B} writes the class count as $|\mathcal Y|$); $k$ is the anchor count
per window except as the $k$-NN neighbor count ($k=20$) of Parts~I
and~III; $m$ denotes unlabeled draws, the target sample size, and the
finest retrieval scale ($m=5$) of Section~\ref{sec:attribution-procedure};
and $D$ is the ambient width $192$, the wall $D_\Delta$, and the polynomial
degree of Section~\ref{sec:dichotomy}. Wall~A denotes the identification
wall; its raw form is $D_\Delta$ and its deployable form $\mathrm{WallA}$,
which coincide on fiber-unions (Theorem~\ref{thm:borrowed}). Each use is
disambiguated where it occurs.

\noindent\emph{Reproducibility.}
All data sets are publicly documented research benchmarks; access-controlled
sources are used under their applicable data-use terms. They comprise the TabReD streams
\citep{rubachev2024tabred} and
the fifteen shift benchmarks of Section~\ref{sec:map} (census, fraud,
network-security, medical, credit, and bike-sharing sources). Appendix~\ref{app:J}
specifies the full measurement protocol (probe and oracle
configurations, noise-floor estimation, deduplication, and the
threshold computation)
at a level intended to permit independent re-implementation. The
machine-verification programs behind the appendix proofs (exact rational
enumeration and an independent cross-check written without shared code)
are part of the reproduction deposit together with their outputs, and
Section~S4 of Online Appendix~1 indexes them.
\providecommand{\ArtifactDepositNote}{will be deposited in a public
archive; the DOI reference will be added in a revised version of this
manuscript}%
The measurement code and every archived input behind the primary empirical
results \ArtifactDepositNote. The deposit holds the eight
per-stream probe slices (frozen
predictions, embeddings, residuals), the stored embeddings behind the
deduplication checks, the verification programs, and a manifest with
hashes. The historical system's per-sample outputs (a tuned seven-prior
ensemble, Section~\ref{sec:weld}) belong to the companion study and are not redistributed; its values are retained only
as a secondary comparison. Section~\ref{sec:weld}'s single-prior
reconstruction, fully covered by the deposit, is the primary reproducible result.

%% file: sec_anatomy.tex
\section{The Anatomy of Adaptation Gains}
\label{sec:anatomy}

Prequential adaptation helps on most of the evaluated temporal tabular
streams and harms on one, but existing evidence offers little guidance on
when either outcome should occur. We ask which components of the observed
gains can be attributed to recurring mechanisms. The gains do not generally
amount to recovering the drifted conditional $p_t(y\mid x)$. Without labels,
such recovery is impossible within the agnostic drift class (Part~II); with
revealed labels, it is rate-limited by the embedding's intrinsic dimension
(measured $d_0\approx3$--$18$; channel C3 below). We therefore treat the
observed improvements as diagnostic channels, each with a theoretical status
and a measured value on the evaluated streams.

Throughout, gains are measured on the eight TabReD streams under the
prequential protocol of Section~\ref{sec:setup}; all oracle and ceiling
quantities follow the time-respecting measurement methodology of
Section~\ref{sec:methodology}.

\subsection{Four Gain Mechanisms and Two Collapse Classes}
\label{sec:channels}

We organize the corrector's observed improvement into four measured
mechanisms (C1--C4) and one pair of identification-collapse classes
(jointly labeled C5). For each mechanism, we specify why it is available, how its
time-respecting ceiling is measured, its value on the eight streams
(Table~\ref{tab:channels}), and the conditions under which it is absent or
harmful. The list is
designed for mechanism attribution in the evaluated corrector, not as a
mutually exclusive decomposition of every possible adaptation algorithm.
Its levels differ deliberately: C1--C4 are measured mechanisms of the
evaluated corrector, while C5 sits one level up, as a pair of identifiable
drift classes (mechanism priors that collapse the identified set of
Part~II) with no measured value of its own.
Theorem~\ref{thm:irred} (Part~II) separately shows that the conditional law
$p_t(y\mid x)$ is unidentified from unlabeled data beyond its ball
constraint, while Theorem~\ref{thm:covshift} gives two canonical mechanism
classes under which that identified set collapses.

\begin{table}[t]
\centering\small
\setlength{\tabcolsep}{4pt}%
\begin{tabular}{llrlr}
\toprule
stream & task & reconstructed gain & dominant channel(s) & signal ceiling \\
\midrule
\textsc{sberbank-housing} & reg & $+3.80\%$ RMSE & C2 freshness & $\le 0$ \\
\textsc{weather} & reg & $+3.94\%$ RMSE & C3 local & $+0.025$ \\
\textsc{cooking-time} & reg & $+0.83\%$ RMSE & C1 debias & $\le 0$ \\
\textsc{delivery-eta} & reg & $+0.66\%$ RMSE & C1 debias ($93\%$ on the probe) & $\le 0$ \\
\textsc{maps-routing} & reg & $+0.59\%$ RMSE & C1/partial-scale C3 & $\le 0$ \\
\textsc{ecom-offers} & cls & $+0.86$pp AUC & C4 ranking & $\le 0$ (MSE) \\
\textsc{homesite-insurance} & cls & $+0.74$pp AUC & C4 ranking & $+0.016$ (MSE) \\
\textsc{homecredit-default} & cls & $\mathbf{-2.28}$\textbf{pp AUC} & harm case (Section~\ref{sec:anatomy-numbers}) & $\le 0$ \\
\bottomrule
\end{tabular}
\caption{The anatomy on eight TabReD streams. The gain column reports the
reproducible single-prior reconstruction over the frozen model; the dominant
channel column records the counterfactual attribution of the historical
seven-prior deployed system (the \textsc{delivery-eta} percentage is the
probe-corrector replay of Section~\ref{sec:attribution}). The final column is the time-respecting
(leak-free) squared-error signal ceiling. Ranking
gains on classification streams flow through channel C4, which squared-error
ceilings cannot register (Section~\ref{sec:channels}). Ceilings are fractions of
residual variance under the strict out-of-time protocol of Section~\ref{sec:leak}.
Table~\ref{tab:bootstrap} gives bootstrap intervals and the
non-redistributed seven-prior system's figures as a secondary historical
comparison.}
\label{tab:channels}
\end{table}

\noindent\emph{C1: Global and Slow-Timescale Debias.}
The simplest channel tracks the slowly-moving mean of the residual
$r_t = \hat y_t - y_t$ and subtracts it. The mechanism asks nothing of the
representation. Whatever component of the drift is shared across the
covariate space (a common shift in the target's level) surfaces in
the prequential residual stream as a slowly-moving mean, and a running
average over strictly-past residuals estimates that mean with essentially no
variance cost. The channel is trivially identifiable from prequential
labels because it uses them only through their first moment: no retrieval,
no locality, no model of \emph{where} the drift acts.

Its contribution is measured by the counterfactual replay of
Section~\ref{sec:attribution}: repeat the identical prequential protocol
with the corrector reduced to its running-mean term alone, and compare the
debias-only gain with the full system's. On one stream, this simple channel accounts for most of the gain. On \textsc{delivery-eta}, the running-mean-only
corrector reproduces $93\%$ of the probe corrector's $+0.66\%$ RMSE
improvement, and the retrieval machinery contributes the small remainder.
Refining the offset yields little more (on \textsc{cooking-time},
region-resolved offsets add at most $+0.07$ percentage points over a single
global offset).

C1 is inexpensive and does not appear as a distinct architectural component,
so its contribution is easy to misattribute to the full method. Any evaluation of a sophisticated
adaptation method that does not report the debias-only baseline conflates
this channel with genuine local structure (Section~\ref{sec:attribution}).
Conversely, the channel is small where the drift is not a level shift.
On \textsc{sberbank-housing} and \textsc{weather} the dominant channels are
instead freshness and local structure (Table~\ref{tab:channels}).

\noindent\emph{C2: Freshness (Drift Tracking).}
Under genuine distribution drift, \emph{when} a labeled point was observed
matters independently of how many are available. The residual field the
corrector must estimate is itself moving, so an old label reports on a
field that no longer holds. The channel is isolated by a size-matched
design on \textsc{sberbank-housing}: the same corrector is run with buffer
size held fixed and only the buffer's \emph{composition} varied (the
most recent labels, a random past-only subset, or the oldest labels of the
stream). All three arms draw from the prequential past, so no future label
(Section~\ref{sec:leak}) can inflate any arm, and the spread across arms
measures recency and nothing else. The arms deliver $+4.9\%$, $+4.1\%$,
and $-1.1\%$ RMSE respectively. Stale supervision is actively
\emph{harmful}, not merely useless.

The matched-buffer comparison isolates freshness as a measurable channel:
the freshest and stalest arms differ by about six points. A corrector whose
buffer is never refreshed also inherits the stale-label penalty as that
buffer ages. In this stream,
tracking recency prevents that degradation; the experiment does not
establish that online adaptation is necessary for every frozen predictor
or deployment. Note also what C2 does \emph{not} deliver: it
consumes labels, the one currency the wall of Part~II respects.
Tracking recency from revealed labels is possible precisely because it
never requires certifying, from unlabeled data, how far the world has
moved (Theorem~\ref{thm:irred}).

\noindent\emph{C3: Local Re-Estimation in Representation Space.}
Prediction errors of the frozen model are \emph{spatially} organized in its
own embedding. On \textsc{weather}, the mean absolute residual difference
between embedding nearest neighbors is $78\%$ of that between random pairs
(neighbors share errors), while the \emph{temporal} autocorrelation of the
residual stream is $0.025$ at lag one. Errors live in space, not in
time. The mechanism is a property of the frozen pair $(\phi,h_0)$: points
the embedding places together err together, so a corrector that indexes
its memory by embedding position (a past-only $k$-nearest-neighbor
regression of residuals over past labeled points) can re-estimate the
frozen model's error locally. Deployed, this delivers $+3.9\%$ RMSE on
\textsc{weather}.

This channel's ceiling is the quantity most easily overstated, and its
measurement is where the time-respecting methodology is essential. The natural
oracle (batch leave-one-out $k$-NN regression of residuals on
embeddings) draws $49.5$--$53.2\%$ of its neighbors from the future of the
query point (Section~\ref{sec:leak}). Every ceiling we report is instead
time-respecting, computed with past-only growing buffers or strict out-of-time
splits. On \textsc{weather} the correction is first-order: the batch
ceiling of $0.0995$ of residual variance falls to $0.067$ with the causally
growing buffer and $0.0245$ under the strict out-of-time split (both
defined in Section~\ref{sec:leak}), and the deployed corrector (RMSE
$1.543$, a $+3.9\%$ reduction against its own frozen arm) already sits
near what even the leaking batch
oracle reaches on the probe stream ($1.513$ against a frozen $1.615$;
Section~\ref{sec:weld}). The channel is
identifiable but rate-limited. Its ceiling is governed by the intrinsic
dimension of the embedding (measured $d_0 \approx 3$--$18$ across streams,
Table~\ref{tab:J-map}, robust to the deduplication check of
Section~\ref{sec:dedup}), and Part~II shows the
resulting nonparametric rates are the binding constraint. The measured
time-respecting ceilings of Section~\ref{sec:weld} are small \emph{because} $d_0$
is not small.

The channel is absent in two distinct ways. On the classification streams
its squared-error signal is zero or negative on two of the three and
marginal on the third (Table~\ref{tab:channels};
\textsc{homesite-insurance}'s $+0.016$ falls to $0.006$ under the
temporal-gap stress test of Section~\ref{sec:leak}). The residual variance there is almost entirely
Bernoulli label noise. Prediction-based floors run from $0.85$ to
$0.88$ of residual variance across the three streams, and temporally
decoupled variogram readings from $0.73$ to $1.02$, leaving no
spatially organized residual for retrieval to recover, and what is
correctable on those streams lives in the ranking currency of C4. And on
\textsc{homecredit-default} the channel turns actively harmful: at a
$2.2\%$ positive rate, local residual smoothing degenerates into noisy
base-rate estimation (Section~\ref{sec:anatomy-numbers}).

\noindent\emph{C4: Ranking and Recalibration (Classification).}
On classification streams a fourth channel appears that squared-error
analysis cannot see. The mechanism is a currency phenomenon: the
classification residual variance is almost entirely Bernoulli label noise
(the noise-floor fractions already measured under C3), so in squared error there
is nothing to correct. The time-respecting signal ceilings of
Table~\ref{tab:channels} are zero, negative, or marginal on these streams. But the
frozen model's \emph{scores} can still be systematically miscalibrated
region by region, and offsets applied in the natural parameter move
ranking and calibration in ways mean-squared-error accounting cannot
register.

The measurement is guarded by an out-of-time design: region-resolved
calibration offsets are fit strictly before a temporal cut and evaluated
strictly after it, and the result is replicated across four temporal cut
points (the first $40$, $50$, $60$, and $70\%$ of the stream) and six
$k$-means codebook seeds, at the best of three codebook sizes
(Appendix~\ref{app:J}). Measured this way, the offsets improve AUC by
$+1.3$ to $+1.9$ points on \textsc{homesite-insurance} and $+1.3$ to
$+1.8$ points on \textsc{ecom-offers}, while the deployed full systems on
the same streams realize $+0.7$ and $+0.9$ points (Table~\ref{tab:bootstrap},
historical column; the reconstruction of Table~\ref{tab:channels} gives
$+0.74$ and $+0.86$). These gains coexist with
squared-error ceilings that are zero, negative, or marginal: a loss-denominated
audit would declare the streams uncorrectable while ranking metrics
improve.

The channel's failure mode is the harm case of
Section~\ref{sec:anatomy-numbers}: on \textsc{homecredit-default},
likelihood-optimal regional offsets \emph{anti-align} with ranking
in-sample, and no out-of-time offset configuration is positive at any
granularity (the numbers are given there). The currency of measurement
is therefore not an incidental detail. Section~\ref{sec:gates} shows the
same mismatch defeats loss-based safety monitors.

\noindent\emph{C5: Identifiable Drift Classes (Covariate and Label-Marginal
Response).}
Two canonical drift mechanisms are identifiable \emph{even without
target labels} (Theorem~\ref{thm:covshift}). Both are \emph{collapse
classes}: mechanism
priors under which the identified set of Part~II shrinks to a single
conditional, so the wall is zero and no labels are required.

The first collapse is \emph{prior-driven}. Under pure covariate shift
(the window's joint law is any covariate marginal paired with the unchanged
source conditional) the identified set is the singleton
$\{p_0(\cdot\mid\cdot)\}$, and the wall vanishes for every
diagonal-vanishing discrepancy, every region, and every budget, already at
$n=0$. No unlabeled sample is even needed, because the collapse is the
content of the mechanism assumption itself
(\hyperref[a:B13]{Proposition~B.13}). The frozen predictor remains correct
where it matters, and the channel consists of leaving the conditional
alone.

The second is \emph{data-pinned}. Under label-marginal shift with known
class-conditionals that are linearly independent as measures (and a
window marginal consistent with the class), the target prior $\pi_W$ is
uniquely identified from unlabeled covariates by a moment system,
$\pi_W = M^{-1}\mathbb E_{p_W}[T]$ for a bounded statistic $T$ whose
moment matrix $M$ against the class-conditionals is nonsingular
(\hyperref[a:B16]{Theorem~B.16}; this is the moment system of black-box
shift estimation, \citealp{lipton2018bbse,garg2020unified}). The
identified set collapses to the single reweighted conditional, and
$\pi_W$ is estimated from unlabeled draws at the parametric rate
(\hyperref[a:B17]{Proposition~B.17}). Unlike the covariate-shift collapse
this one has genuine data content, but its hypothesis is not without cost: the
known class-conditionals are label-derived side information obtained
outside the window (\hyperref[a:B21]{Remark~B.21}).

A deployed corrector inherits both mechanisms automatically, and they delimit the
``unlabeled-correctable'' portion of any real drift. The channel's failure
mode is itself a theorem: \emph{which} mechanism class a given window's
drift belongs to is unfalsifiable from unlabeled data. When the window
marginal is consistent with the label-shift class, worlds drawn from the
agnostic class, from covariate shift, and from label shift all induce the
same unlabeled law at every sample size, so every label-free test between
mechanism classes has power equal to size
(\hyperref[a:B20]{Proposition~B.20}). Invoking C5 is a modeling
commitment, not a measurement. The collapse is real when the mechanism
prior is right, and whether it is right cannot be checked from unlabeled
data.

\subsection{The Eight-Stream Decomposition and the Harm Case}
\label{sec:anatomy-numbers}

Table~\ref{tab:channels} summarizes the per-stream decomposition. The results show three recurring patterns.

\emph{Gain concentration.} Every positive
stream's improvement ($+0.6\%$ to $+4.6\%$ RMSE and $+0.7$ to $+0.9$ AUC
points on the historical deployed system, Table~\ref{tab:bootstrap};
$+0.59\%$ to $+3.94\%$ and $+0.74$ to $+0.86$ in the reconstruction of
Table~\ref{tab:channels}) is accounted for by C1--C4, with
the mixture varying by stream (\textsc{delivery-eta} is almost purely C1,
\textsc{sberbank-housing} is dominated by C2, \textsc{weather} by C3, and the
classification streams by C4).

\emph{Observed harm.} On \textsc{homecredit-default} (default
rate $2.2\%$), naive residual correction \emph{costs} $2.46$ AUC points
on the deployed system ($-2.28$ in the reproducible reconstruction of
Table~\ref{tab:bootstrap}; the probe corrector's never-lock evaluation
reads $-2.36$, Section~\ref{sec:gates}), and
the damage is mechanistic rather than accidental. Likelihood-optimal
regional offsets anti-align with ranking on this stream (in-sample,
offsets that improve log-loss by $2.1\%$ simultaneously reduce AUC by $1.47$
points), and no out-of-time offset configuration is positive at any
granularity. Rare-positive streams turn local residual smoothing into noisy
base-rate estimation. The harm case motivates the gate analysis in Section~\ref{sec:gates}.

\emph{Proximity to the descriptive ceiling.} The time-respecting
(leak-free) signal ceilings of Section~\ref{sec:weld} are small (e.g.\
$2.5\%$ of residual variance on \textsc{weather}, less elsewhere), and the
deployed systems capture the bulk of them. The measured room for local residual correction is small for this
representation and protocol, consistent with the labeled-history boundary
of Part~II, whose nonparametric rates depend on $d_0$
(Section~\ref{sec:dichotomy}). We therefore compare adaptation methods
against the time-respecting ceiling rather than only against each other.

\section{Measurement Methodology for Streaming Adaptation}
\label{sec:methodology}

The preceding measurements depend on the evaluation protocol. Several common
practices in streaming adaptation introduce optimistic bias; on our own
data, two of them reverse a conclusion. This section documents five corrections. Each is illustrated with the measured size of the artifact it
removes. Table~\ref{tab:biases} summarizes the four biases and their
corrections, and Figure~\ref{fig:artifacts} plots the per-stream
measurements behind three of them. The fifth item
(Section~\ref{sec:attribution}) is a reporting correction rather than a
numerical one.

\begin{table}[t]
\centering\small
\begin{tabular}{p{0.26\linewidth}p{0.36\linewidth}p{0.28\linewidth}}
\toprule
bias & measured artifact & corrected protocol \\
\midrule
future-neighbor leakage (Section~\ref{sec:leak}) &
$49.5$--$53.2\%$ future neighbors on $8/8$ streams; ceilings inflated up to
$4.5\times$; sign flip on \textsc{sberbank-housing} &
time-respecting oracles: past-only buffers or strict out-of-time splits \\
\addlinespace
temporal-twin noise floors (Section~\ref{sec:twins}) &
\textsc{weather} nugget $0.0009$ vs.\ decoupled floor $0.13$--$0.62$ of
residual variance &
variogram on pairs $\ge 1{,}000$ steps apart \\
\addlinespace
duplicate-deflated $d_0$ (Section~\ref{sec:dedup}) &
\textsc{unsw\_nb15} $0.87\to3.58$; \textsc{acs\_employment}
$1.62\to3.11$; both leave the regular branch &
deduplicate at four decimals before estimating $d_0$ \\
\addlinespace
currency-mismatched gates (Section~\ref{sec:gates}) &
$84\%$ of AUC harm before the log-loss lock; $-11.4$ AUC points in
the first $1\%$ &
monitor the deployment metric itself \\
\bottomrule
\end{tabular}
\caption{Four optimistic measurement biases for streaming adaptation, the
measured size of the artifact each introduces (each magnitude is
measured on the stream that exhibits it), and the corrected protocol.
Every number is discussed in the corresponding subsection. The fifth
correction, counterfactual channel attribution
(Section~\ref{sec:attribution}), guards reporting rather than a number and is
not tabulated.}
\label{tab:biases}
\end{table}

\begin{figure}[t]
\centering
\includegraphics{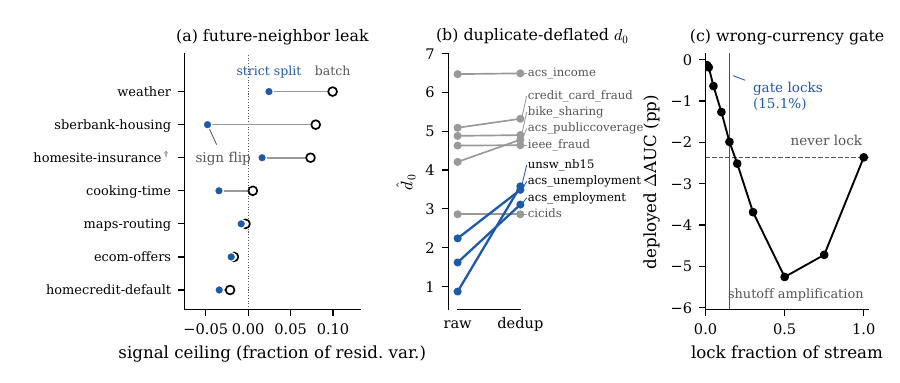}
\caption{The measured size of three artifacts (per-stream intrinsic
dimensions for panel~(b) in Appendix~\ref{app:J}; the remaining values are
part of the reproduction deposit). \emph{(a)}~Squared-error signal ceilings: the batch
leave-one-out oracle (open) against the strict out-of-time split
(filled), as fractions of residual variance. The leak inflates
\textsc{weather} by $4.1\times$ and \textsc{homesite-insurance} by $4.5\times$
and flips the sign on \textsc{sberbank-housing}
(the batch ceiling for \textsc{delivery-eta} was not recorded, so only
its strict-split value is shown). The dagger flags
\textsc{homesite-insurance}, whose strict-split value (filled) collapses
further, to $0.006$, under the temporal-gap stress test of
Section~\ref{sec:leak}.
\emph{(b)}~Intrinsic dimension before and after four-decimal
deduplication for the nine re-measured related benchmarks
(accented: the
three large movers, \textsc{unsw\_nb15}, \textsc{acs\_employment}, and
\textsc{acs\_unemployment}; the TabReD streams move by at most $0.08$ and are
not shown; Section~\ref{sec:dedup}). \emph{(c)}~\textsc{homecredit-default}: deployed AUC change
of the probe corrector as a function of where the log-loss gate locks.
The actual lock ($15.1\%$, vertical line) arrives after $84\%$ of the
never-lock harm (dashed), and a mid-stream shutoff can more than double
that harm (Section~\ref{sec:gates}).}
\label{fig:artifacts}
\end{figure}

\subsection{Future-Neighbor Leakage Inflates Signal Ceilings}
\label{sec:leak}

The natural oracle for ``how much correctable signal exists'' is a batch
leave-one-out $k$-NN regression of residuals on embeddings. The mechanism
of the bias is that batch retrieval is symmetric in time while deployment
is not. On all eight TabReD streams, $49.5$--$53.2\%$ of the oracle's
selected neighbors lie in the \emph{future} of the query point. On a drifting
stream this is not an incidental detail. The residual field itself moves,
so a future neighbor that is close in embedding space carries the locally
drifted residual the oracle is being asked to predict. The batch oracle is
scored with information no deployed corrector can possess, and its
``ceiling'' overstates the correctable signal.

The leak is not benign, and its measured size changes conclusions
(Figure~\ref{fig:artifacts}a). Recomputing ceilings under a time-respecting
protocol deflates the batch numbers by up to a factor of $4.5$ (on
\textsc{homesite-insurance}, $R^2$ falls from $0.0735$ to $0.0164$, and a
temporal-gap stress test, which additionally withholds the $1{,}000$ fitting
rows nearest the cut, lowers even the corrected value to $0.006$)
and \emph{reverses the sign} of the apparent signal on
\textsc{sberbank-housing} (batch $+0.080$, time-respecting $-0.048$), where
the entire apparent squared-error signal was the leak.

The corrected procedure comes in two grades. The deployment-faithful grade
restricts neighbors to the strictly-past prefix with a causally growing
buffer, exactly as a deployed corrector would retrieve; the conservative
grade is a strict out-of-time split: fit on the first half of the
stream, evaluate on the second. On \textsc{weather}, the batch value of
$0.0995$ becomes $0.067$ with the causally growing buffer and $0.0245$
under the strict split. Every ceiling
in this paper is time-respecting, and we recommend the field adopt the same convention
for any retrieval-based oracle on ordered data.

\subsection{Temporal Twins Deflate Noise Floors}
\label{sec:twins}

The irreducible-noise floor of a stream is naturally estimated from the
variogram of residuals at vanishing embedding distance. Pairs that the
embedding places arbitrarily close should disagree only through
irreducible noise, so the variogram's intercept (the \emph{nugget}) reads
off the floor. The reading is valid only if near pairs carry
\emph{independent} noise realizations. On temporally ordered tabular data
they do not: the closest embedding pairs are \emph{temporal twins},
near-duplicate records of the same physical event (on \textsc{weather},
median stream separation $12$ steps) that share their noise realization.
A shared realization contributes nothing to pair disagreement, so
the nugget reads duplication as correctability.

The size of the artifact is extreme. On \textsc{weather} the naive nugget
is $0.0009$ (suggesting essentially all residual variance is
correctable in principle), while the temporally decoupled estimate puts
the floor at $0.13$--$0.62$ of residual variance, over two orders of
magnitude higher.

The corrected procedure restricts the variogram to pairs at least
$1{,}000$ steps apart in the stream. On \textsc{weather} this leaves
\emph{no support} at small embedding distances (the twins were the only
close pairs), so the decoupled floor is an extrapolation, and we report it
as the range $0.13$--$0.62$ rather than a point. Two rules follow: noise
floors on ordered data must be computed from temporally decoupled pairs;
and a near-zero nugget should be treated as a duplication symptom, not a
promise of correctable signal. The artifact is the streaming sibling of
the duplicate-deflated intrinsic dimension of Section~\ref{sec:dedup},
because both are driven by (near-)duplicate records. The measurement is
scoped: on the three classification streams, where the residual field is
a bounded Bernoulli residual rather than a smooth field, decoupling does
not systematically raise the floor (the readings move from $0.98$ to
$0.73$, from $1.05$ to $1.02$, and from $0.97$ to $0.98$; these readings are reported
for contrast only, since the decision-rule evaluation of
Section~\ref{sec:certificate} uses the Bernoulli floor $\mathbb E[p(1-p)]$ on
the classification streams). We claim the
twin artifact for the regression streams we measured it on, not as a
general law.

\subsection{Duplicates Deflate Intrinsic Dimension}
\label{sec:dedup}

The channel-C3 ceiling and the boundary threshold of Part~II both depend
on the intrinsic dimension $d_0$ of the frozen embedding. The mechanism of
the bias is elementary. Duplicate rows are endemic in tabular data
(repeated transactions, network flows, census microdata), and they place
many points at first-neighbor distance $r_1\approx 0$, which breaks the
standard estimators (TwoNN, \citealp{facco2017twonn}; the maximum-likelihood estimator,
\citealp{levina2004mle}) and biases $d_0$ sharply \emph{downward}. The
direction of the bias is what makes it dangerous. A spuriously low $d_0$
manufactures exactly the ``low-dimensional, inexpensively adaptable'' reading
that the threshold $\gamma^\star=d_0/(2\alpha_s)$ of Part~II rewards.

The bias is large enough to change conclusions. \textsc{unsw\_nb15}, a
network-intrusion stream with $41\%$ duplicate flows, reads $d_0=0.87$ raw
and $3.58$ after removing rows identical to four decimals. The raw
value is not even reproducible across estimators (TwoNN gives $2.93$ on
the same raw data), an instability that is itself part of the diagnosis.
\textsc{acs\_employment}, a census stream with $16\%$ of rows duplicated to four decimals,
moves from $1.62$ to $3.11$ (Figure~\ref{fig:artifacts}b). On the raw values, both data sets appear
to fall into the regular ($\sqrt N$) branch of Part~II. After deduplication,
neither remains in that branch under the sensitivity inputs
$(\alpha_s,\gamma)=(1,1)$ (Appendix~\ref{app:J}); this correction does not
constitute a measured phase call for the related benchmarks.

The corrected procedure removes rows identical to four decimal places,
recomputes every $d_0$ estimate on the unique subset (for the essentially
duplicate-free TabReD streams the raw and deduplicated readings agree to
$0.08$, Appendix~\ref{app:J}), and cross-checks
estimators on borderline cases. The correction is inexpensive and, on clean
data, inert (the eight TabReD streams are $99.8$--$100\%$ duplicate-free
and the six high-$d_0$ related benchmarks $99.97$--$100\%$; the $d_0$
estimates move by at most $0.08$ and $0.10$ respectively). It matters exactly on
the data sets that look most favorable.

\subsection{Architectural Attribution Miscredits Gains}
\label{sec:attribution}

Which channel produced a system's gain cannot be read off its architecture
or its learned weights. The artifact here is interpretive rather than
numerical: an adaptation system's gain is naturally credited to its most
sophisticated component, because that component is what the paper is
about. The \textsc{delivery-eta} decomposition of
Section~\ref{sec:channels} is the cautionary measurement. On that stream, $93\%$ of the probe corrector's gain is reproduced by a running mean, and nothing in the
system's architecture gives any hint that its retrieval machinery is doing
almost nothing.

The corrected procedure is \emph{leave-one-channel-out counterfactual
replay} on the same stream: repeat the identical prequential protocol with
one channel disabled, and attribute to each channel the difference its
removal makes. Disabling a channel means freezing the debias term at
zero (C1 off), removing the buffer's recency weighting (C2 off), or
replacing the local corrector by the global one (C3 off). The
C1--C3 entries of the dominant-channel column of Table~\ref{tab:channels}
were produced this way (the C4 entries come from the out-of-time
calibration design of Section~\ref{sec:channels}), and the
\textsc{cooking-time} comparison of region-resolved against global offsets
($+0.07$~percentage points at most) is the same replay at finer granularity. We
regard channel-resolved
attribution tables as the minimum reporting standard for adaptation
papers, for the same reason ablations are standard elsewhere: without
them, trivial channels masquerade as method contributions.

\subsection{Currency-Mismatched Gates Miss Deployment Harm}
\label{sec:gates}

Because adaptation can harm (Section~\ref{sec:anatomy-numbers}), deployed
correctors need an online no-harm monitor. A safety gate is a measurement
instrument too, and it can be biased in the same optimistic direction as
an oracle. The natural gates are \emph{loss}-denominated, because that is
the currency in which anytime-valid sequential tests are available
\citep{ramdas2023gametheoretic}, while
the deployment currency on classification streams is ranking. A statistically valid anytime monitor in the \emph{wrong currency} can
therefore miss deployment harm.

On \textsc{homecredit-default} we ran an e-process (test-by-betting,
\citealp{shafer2021betting}) monitor on per-round \emph{log-loss} regret of the corrector against the
frozen model. The monitor is sound (its false-alarm rate under
permutation nulls matches its nominal level), and it does eventually
fire, locking the corrector out at $15.1\%$ of the stream
(Figure~\ref{fig:artifacts}c). But $84\%$ of the \emph{AUC} harm
is incurred before the lock: over the first $1\%$ of the stream AUC falls by
$11.4$ points while the monitored log-loss barely moves, and from $2$ to
$5\%$ of the stream the correction \emph{improves} log-loss outright
while the ranking harm keeps accruing. This is the anti-alignment of
Section~\ref{sec:anatomy-numbers} seen from the monitor's side. The gate
analysis uses the probe corrector (configuration in
Appendix~\ref{app:J}); its never-lock harm of $-2.36$ points reflects the same
mechanism as the deployed system's $-2.46$. It is measured by the original
probe-corrector code inside the gate experiment, whereas the $-2.28$ of
Table~\ref{tab:bootstrap} is the deposit's reimplementation of the same
operator scored on the $55{,}901$ points after the $100$-point buffer; the
two agree well within the bootstrap interval of Table~\ref{tab:bootstrap}. On this
stream, offsets that help the likelihood hurt the ranking, so a gate
watching the likelihood is reassured precisely while the deployment
metric is being damaged.

The failure is structural in two senses. First, the information needed for
an early stop does not exist in the monitored currency (log-loss is
flat, then \emph{improving}, while the worst of the harm accrues), so no
threshold tuning, no faster alarm, no better betting scheme can repair the
monitor. Second, the gate's own action is not free in the deployment
currency: mid-stream stopping introduces a score-scale discontinuity that
can add further ranking harm. Locking at half-stream more than doubles
the never-lock harm ($-5.26$ against $-2.36$ AUC points on the probe stream;
Figure~\ref{fig:artifacts}c). The evidence supports a simple operational rule: \emph{monitor the metric
you deploy on}. Constructing anytime-valid monitors for ranking currencies is,
to our knowledge, open: the game-theoretic testing framework
\citep{shafer2021betting,ramdas2023gametheoretic} supplies the loss-denominated
instruments used above, and we are not aware of a ranking-denominated
counterpart.

\subsection{Ceilings versus Realized Gains}
\label{sec:weld}

\begin{figure}[t]
\centering
\includegraphics{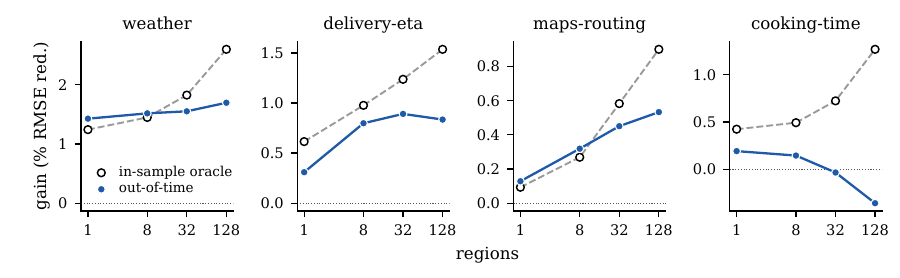}
\caption{In-sample versus out-of-time gain as adaptation capacity grows: the
regional-offset probe ($k$-means codebooks of $1$, $8$, $32$, and
$128$ regions; protocol in Appendix~\ref{app:J}) on the four
regression streams with adequate sample size for the split. Open, dashed: offsets
fit with hindsight on the full stream. Filled: fit on the first half and
evaluated on the second (a single out-of-time split, no error bars). The
in-sample gain grows monotonically with codebook size on every stream
while the
out-of-time gain flattens, saturates, or reverses sign: added capacity
yields apparent gain, not deployable gain. At small codebook sizes the arms can
cross, as
on \textsc{weather}, since they score different halves of a drifting
stream. \textsc{sberbank-housing} (the shortest stream) is excluded because its
evaluation half is too short to support the codebook sweep; the classification streams use the out-of-time
protocol of Section~\ref{sec:channels}.}
\label{fig:kf1}
\end{figure}

With the corrected protocol in hand, the regression-stream audit finds
small time-respecting signal ceilings (between $0$ and $2.5\%$ of residual
variance in squared error); classification has a separate ranking-currency residual signal in channel C4. These quantities are descriptive point estimates,
not confidence bounds or universal ceilings.

\noindent\emph{Reproducibility and Uncertainty of the Headline Gains.}
Table~\ref{tab:bootstrap} reports a single-prior local corrector evaluated
with a block bootstrap; this reconstruction is the primary empirical record.
The corrector is a compact reimplementation of the
streaming $k$-NN probe corrector used in Sections~\ref{sec:gates}
and~\ref{sec:certificate} (configuration in Appendix~\ref{app:J}), applied
to the probe streams; the reproduction
deposit (Section~\ref{sec:setup}) covers it in full. The secondary historical figures come from a tuned
seven-prior ensemble whose per-sample outputs we do not redistribute. The reconstruction
reproduces the deployed figure to its printed decimal on five of the
eight streams, and on two more it agrees in sign and magnitude without
matching the printed digit (\textsc{cooking-time} $+0.83$ against
$+0.9$; the harm case $-2.28$ against $-2.5$). Thus the qualitative
seven-gain/one-harm pattern and its approximate magnitudes are reproducible from the deposit;
the exact historical figures still depend on access to the tuned system. On
\textsc{sberbank-housing} the reconstruction reads $+3.8\%$ against
$+4.6\%$: the historical system's attribution assigns that stream's gain to
freshness (C2) rather than to the local channel this corrector
implements, so a local-only corrector is not expected to recover it.

Every interval excludes zero, the harm case included. The last column measures the resolution of the evaluation
itself. The half-width of the moving-block interval for the frozen
model's \emph{own} metric exceeds the reconstructed gain on four of the
eight streams, so a difference of a few tenths of a percent is resolvable
only because the bootstrap is paired on the same resampled blocks.
Unpaired comparisons at these effect sizes carry more evaluation noise
than signal. This applies to two systems scored on different splits, and to any
table of published numbers drawn from different evaluations.

On \textsc{weather}, the tuned corrector
achieves an RMSE of $1.543$, while the batch oracle on the single-prior
probe (scored with future neighbors it could not have had) reaches
$1.513$ against that probe's frozen $1.615$, little beyond the deployed
value. The deployed gain
itself is the $+3.9\%$ historical figure of Table~\ref{tab:bootstrap},
measured against the tuned system's own frozen arm; the oracle figure bounds what the leak
could have yielded on the probe, not what the deployment gained. Figure~\ref{fig:kf1} shows the same fact in capacity-resolved
form: growing the regional capacity of an offset corrector inflates the
hindsight gain monotonically while the out-of-time gain flattens
or reverses. Added fitting capacity converts drift into apparent
signal, not into deployable gain. Part~II explains why this ceiling is low within the stated agnostic drift
class. It also separates the limitations addressed by additional labels from
those addressed by a different representation.

\begin{table}[t]
\centering\small
\begin{tabular}{llrrl r}
\toprule
stream & task & reconstruction & historical & $95\%$ CI & metric $\pm$ \\
\midrule
\textsc{weather}            & reg & $+3.94\%$ & $+3.9\%$ & $[+3.60,+4.27]$ & $1.43$ \\
\textsc{sberbank-housing}   & reg & $+3.80\%$ & $+4.6\%$ & $[+2.54,+5.14]$ & $7.80$ \\
\textsc{cooking-time}       & reg & $+0.83\%$ & $+0.9\%$ & $[+0.57,+1.05]$ & $0.86$ \\
\textsc{delivery-eta}       & reg & $+0.66\%$ & $+0.7\%$ & $[+0.45,+0.91]$ & $0.97$ \\
\textsc{maps-routing}       & reg & $+0.59\%$ & $+0.6\%$ & $[+0.43,+0.77]$ & $0.96$ \\
\textsc{ecom-offers}        & cls & $+0.86$pp & $+0.9$pp & $[+0.29,+1.53]$ & $0.69$ \\
\textsc{homesite-insurance} & cls & $+0.74$pp & $+0.7$pp & $[+0.49,+0.98]$ & $0.49$ \\
\textsc{homecredit-default} & cls & $-2.28$pp & $-2.5$pp & $[-3.41,-1.16]$ & $1.38$ \\
\bottomrule
\end{tabular}
\caption{Reconstruction and historical comparison. The primary
reconstruction is the single-prior streaming $k$-NN corrector
($k=20$, unweighted neighbor mean, correction scale $0.7$, block $64$)
run causally over the probe stream; the historical column reports
the non-redistributed seven-prior deployment. Intervals are $95\%$ moving-block bootstrap
($B=1000$, block length $4\lceil n_{\mathrm{test}}^{1/3}\rceil$, paired: each replicate
scores both arms on the same resampled blocks); verdicts are unchanged at
block lengths $2\lceil n_{\mathrm{test}}^{1/3}\rceil$ and $8\lceil n_{\mathrm{test}}^{1/3}\rceil$. The
last column is the half-width of the same bootstrap applied to the frozen
model's own metric, on the gain scale. That half-width is the resolution
of the evaluation itself.}
\label{tab:bootstrap}
\end{table}

%% file: sec_boundary.tex
\section{The Boundary: What Unlabeled Data Can and Cannot Determine}
\label{sec:wallA}

Part~I organized the observed gains into four measured mechanisms (plus
two collapse classes) and
reported descriptive signal ceilings. We use partial identification to study
an informational limitation within the stated agnostic drift class \citep{manski2003partial}. The label-free
information state (labels not yet revealed) is nested inside the
prequential protocol of Part~I, not beside it: at
each prediction time the arriving label has not yet been revealed, so
every prequential step passes through this information state. The
wall of Sections~\ref{sec:idset}--\ref{sec:irred} is what a corrector
faces before its labels arrive, Section~\ref{sec:wallB} adds the wall that
labels cannot remove, and Section~\ref{sec:dichotomy} then shows what
labeled history can recover and at what rate. Every statement is in the proper-loss,
conditional-mean currency of Section~\ref{sec:setup}; ranking is outside
the theory's scope, and we treat it as a measured counterexample. Main-text statements refer to the
named assumptions of Section~\ref{sec:assumptions} and give the exact
conclusions needed for interpretation. Each is followed by its core proof
(routine steps cite appendix lemmas) and a pointer to the full counterpart
in Appendix~\ref{app:B}, Appendix~\ref{app:H}, or Online Appendix~1, all
of which
retain the original letters A--K and statement numbers (A.14, G.1.4,
etc.). There every claim is tagged \emph{proven}, \emph{partial} (a
flagged step remains), \emph{conditional} (on the named gap (2c)), \emph{cited-standard},
\emph{conjecture}, or \emph{verified measurement}; the Status paragraph of
Section~\ref{sec:dichotomy} repeats the flags for the multi-window
results, and the flags for Sections~\ref{sec:idset}--\ref{sec:wallB} are
quoted where the statements occur.

\subsection{The Identified Set and Its Diameter}
\label{sec:idset}

Fix a deployment window $W$ with observable covariate law $p_W$ and
unobserved conditional $p_W(\cdot\mid x)$. The frozen model supplies its
source conditional $p_0(\cdot\mid x)=h_0(\phi(x))$. Absent labels, the only
constraint linking the two is an assumed \emph{drift budget}
$\beta:\mathcal X\to[0,1]$: we posit
$p_W(\cdot\mid x)\in B_\beta(x):=\{q: D_{\mathrm{TV}}(q,p_0(\cdot\mid x))\le\beta(x)\}$.
The \emph{identified set} $I(\beta)$ collects every conditional law
consistent with all label-free observables, and the \emph{wall} is its
diameter in prediction space,
\[
D_\Delta(\beta;R)\;=\;\sup_{q,q'\in I(\beta)}\;
\mathbb E_{x\sim p_W(\cdot\mid R)}\big[\Delta\big(f_q(x),f_{q'}(x)\big)\big],
\]
the worst disagreement between two worlds that no unlabeled quantity can
distinguish. A natural first conjecture holds that $D$ decomposes
additively into an estimable covariate-geometry part and a budget part,
$D=f(R)+g(\beta)$. That conjecture is \emph{false}. Under the correct
normalization the covariate geometry contributes no additive floor at all.

\begin{theorem}[Exact reduction]\label{thm:reduction}
Assume the single-window regularity conditions (S1)--(S4) of Appendix~A
of Online Appendix~1 and a bounded discrepancy kernel $G$ of
Carath\'eodory class (C) or $x$-free finite-range class (U$'$) (the kernel
conditions (S5) of that appendix, not the mechanism prior (S5) of
Appendix~\ref{app:B}). The mechanism prior
(B) is encoded in the definition of $I(\beta)$ as all measurable pointwise
selections. Then the pointwise diameter
$\delta_\beta^\Delta(x)=\sup_{a,b\in B_\beta(x)}G(x,a,b)$ is measurable
with its supremum attained, and
\[
D_\Delta(\beta;R)
=\sup_{q,q'\in I(\beta)}\;\mathbb E_{x\sim p_W(\cdot\mid R)}
 \big[G\big(x,q(\cdot\mid x),q'(\cdot\mid x)\big)\big]
=\mathbb E_{x\sim p_W(\cdot\mid R)}\big[\delta_\beta^\Delta(x)\big],
\]
the outer supremum attained by a measurable (Kuratowski--Ryll-Nardzewski-selected) pair
$q^*,q'^*\in I(\beta)$, so covariate geometry enters only as the averaging
measure. In the binary case with the mean-gap discrepancy
(Corollary~A.16), with $\mu=p_W(\cdot\mid R)$,
\[
\begin{aligned}
D_{\mathrm{mean}}(\beta;R)
  &=\mathbb E_\mu[w_\beta]
    =2\,\mathbb E_\mu[\beta]-c(\beta;R),\\
c(\beta;R)
  &=\mathbb E_\mu\big[(\eta_0+\beta-1)^++(\beta-\eta_0)^+\big]
    \in[0,\mathbb E_\mu\beta],
\end{aligned}
\]
so $\mathbb E_\mu[\beta]\le D_{\mathrm{mean}}\le2\,\mathbb E_\mu[\beta]$,
attained by the explicit endpoint pair
$\Pi_{[0,1]}(\eta_0\pm\beta)$. The decision-flip diameter is
$D_{01}(\beta;R)=\mu(|\eta_0-\tfrac12|\le\beta)$ (tie-inclusive
convention; the tie-broken convention differs only on
$\{\eta_0-\beta=\tfrac12\}$).
\end{theorem}

The clip term $c\ge0$ is the only label-free-measurable content of the
identity, and it is purely subtractive: confident source predictions pull
the wall down, but nothing in the unlabeled data pushes it up.

\medskip\noindent\emph{Proof idea (Appendix~A of Online Appendix~1).} The budget ball
$B_\beta(x)$ is nonempty, compact, convex, with an exact ball-distance
formula that makes the correspondence $\Gamma_\beta:x\rightrightarrows B_\beta(x)$
weakly measurable. The Kuratowski--Ryll-Nardzewski selection theorem then
supplies measurable selections, and a Castaing family of selections dense
in every ball. The measurable maximum theorem (re-proved in full in Appendix~A of
Online Appendix~1, since the argument rests on it) delivers measurability of
$\delta_\beta^\Delta$ and a measurable argmax pair. The prior (B) enters
\emph{exactly once}: because the budget constrains each $x$ separately,
$I(\beta)$ is \emph{decomposable} (splicing two members along any
measurable set stays feasible), and an interchange-of-sup-and-integral
lemma in the mechanism of Rockafellar's interchange theorem (splicing
directedness plus a countable envelope) upgrades the trivial ``$\le$'' to
equality. Any cross-$x$ prior (O3 smoothness, recurrence) breaks
decomposability, and only ``$\le$'' remains valid
(Remark~A.12). The binary corollary is then clip
algebra on the explicit endpoint selections. The normalization anchor
$D_\Delta(0;R)=0$ (hence the refutation $f\equiv0$ of the additive
conjecture) holds for diagonal-vanishing kernels
(Remark~A.17).

\begin{proof}[Proof of the binary formulas]
Write $\underline\eta=\max(0,\eta_0-\beta)$,
$\overline\eta=\min(1,\eta_0+\beta)$, so $B_\beta(x)=[\underline\eta,
\overline\eta]$ and the pointwise diameter under the mean-gap kernel
$G(x,a,b)=|\eta_a-\eta_b|$ (class (C), diagonal-vanishing) is
$\delta_\beta^{\mathrm{mean}}=w_\beta=\overline\eta-\underline\eta$,
attained by the \emph{explicit} endpoint pair
$q^*(1\mid\cdot)=\Pi_{[0,1]}(\eta_0+\beta)=\overline\eta$,
$q'^*(1\mid\cdot)=\Pi_{[0,1]}(\eta_0-\beta)=\underline\eta$.
Both are measurable everywhere-selections needing no selection theory. The
reduction identity (the general part of the theorem) then gives
$D_{\mathrm{mean}}(\beta;R)=\mathbb E_\mu[w_\beta]$. The clip identity
is algebra on the endpoints: substituting
$\min(1,u)=u-(u-1)^+$ at $u=\eta_0+\beta$ and
$\max(0,\eta_0-\beta)=(\eta_0-\beta)+(\beta-\eta_0)^+$,
\[
w_\beta=2\beta-(\eta_0+\beta-1)^+-(\beta-\eta_0)^+,
\qquad\text{whence}\qquad
D_{\mathrm{mean}}(\beta;R)=2\,\mathbb E_\mu[\beta]-c(\beta;R)
\]
with $c(\beta;R)=\mathbb E_\mu[(\eta_0+\beta-1)^++(\beta-\eta_0)^+]\ge0$.
The sandwich follows from the pointwise bound
$0\le(\eta_0+\beta-1)^++(\beta-\eta_0)^+\le\beta$ (a three-case check
using $\beta\le1$, $\eta_0\in[0,1]$), so $c\in[0,\mathbb E_\mu\beta]$ and
$\mathbb E_\mu[\beta]\le D_{\mathrm{mean}}\le2\,\mathbb E_\mu[\beta]$,
with $c=0$ iff $\beta\le\eta_0\le1-\beta$ $\mu$-a.e.\ (the interior
regime). For decision flips, $B_\beta(x)$ meets the cell
$\{a:\eta_a\ge\tfrac12\}$ iff $\eta_0+\beta\ge\tfrac12$ and the cell
$\{a:\eta_a\le\tfrac12\}$ iff $\eta_0-\beta\le\tfrac12$, so the
pointwise flip indicator is $\mathbf 1[|\eta_0-\tfrac12|\le\beta]$
(tie-inclusive) resp.\
$\mathbf 1[\eta_0-\beta<\tfrac12\le\eta_0+\beta]$ (tie-broken), and
averaging gives $D_{01}(\beta;R)=\mu(|\eta_0-\tfrac12|\le\beta)$, the two
conventions differing only when $\mu(\eta_0-\beta=\tfrac12)>0$.
\end{proof}

\medskip\noindent\emph{Formal version and the measurability machinery:}
Theorem~A.14, Appendix~A of Online Appendix~1; the binary clip
algebra above is Corollary~A.16.

An immediate reading is that in the calibrated interior regime the wall
$2\,\mathbb E[\beta]$ is a restatement of the assumed budget with zero data
content. The operational quantities are instead the \emph{decision-flip
fraction} $D_{01}(\beta;R)=\Pr_\mu[\,|\eta_0(x)-\tfrac12|\le\beta(x)\,]$ (the
diameter under the decision discrepancy $\Delta_{01}$; how often the drift
could change the decision) and the \emph{price of freeze}
$\rho_\beta(x)=2\max\big(0,\beta(x)-|\eta_0(x)-\tfrac12|\big)$
(worst-case regret of trusting the frozen model). Both are computable from
$(\eta_0,\beta)$ and both are used by the local-correction screening rule of Part~III.

\subsection{Irreducibility: The Wall Is Assumed, Not Measured}
\label{sec:irred}

\begin{theorem}[Label-free irreducibility]\label{thm:irred}
Assume the standing conditions (S1)--(S7) of Appendix~\ref{app:B} with binary
$\mathcal Y$: Polish $\mathcal X$, Borel kernel $p_0$ and budget $\beta$,
i.i.d.\ unlabeled draws $x_{1:n}\sim p_W$ with $y$ never observed, the
mechanism prior (B), and randomized Borel estimators. Then for every
$n$, including $n=\infty$, over the worlds
$\mathcal W=\{p_W\otimes q:q\in I(\beta)\}$ and the region-mean functional
$\theta(P)=\mathbb E_\mu[\eta_q]$,
\[
\inf_{\hat\theta_n\ \mathrm{randomized}}\;\sup_{P\in\mathcal W}\;
\mathbb E\big|\hat\theta_n-\theta(P)\big|
\;=\;\tfrac12\,D_{\mathrm{mean}}(\beta;R)
\;=\;\tfrac12\,\mathbb E_{p_W(\cdot\mid R)}[w_\beta],
\]
\emph{constant in $n$}. The lower bound is attained already at the
two-point subfamily $\{P_-,P_+\}$ built from
$\eta_\pm=\Pi_{[0,1]}(\eta_0\pm\beta)$, and the upper bound by the
zero-data midpoint estimator. The risk is strictly positive whenever
$p_W(R\cap\{\beta>0\})>0$, and the wall height is a functional of the
prior $(\beta,p_0,p_W)$ alone. Moreover, any test of one drift explanation
$q$ against another $q'$ satisfies
$\mathbb E_{P'}[\psi_n]=\mathbb E_P[\psi_n]$ (power equals size,
Corollary~B.9), and any confidence procedure with uniform coverage
$1-\alpha$ has expected $\theta$-diameter at least
$(1-2\alpha)\,D_{\mathrm{mean}}(\beta;R)$ for every $n$
(Corollary~B.10, outer-expectation form, under its measurability
convention).
\end{theorem}

\begin{proof}
A two-point argument in the tradition of Le~Cam
\citep{lecam1986,tsybakov2009introduction}, resting on an exact
ancillarity. Throughout, $\mu=p_W(\cdot\mid R)$ and
$\underline\eta=\max(0,\eta_0-\beta)$,
$\overline\eta=\min(1,\eta_0+\beta)$, so $w_\beta=\overline\eta-\underline\eta$.

\emph{Worlds and ancillarity.} Each world is the joint law
$P=p_W\otimes q$ on $\mathcal X\times\mathcal Y$, defined by
$P(E)=\sum_y\int_{E_y}q(y\mid x)\,p_W(dx)$, a probability measure with
$x$-marginal $p_W$ and $p_W$-a.e.-unique disintegration $q$
(\hyperref[a:B1]{Lemma~B.1}). Push the $n$-fold product through the
observation map $\Pi_A:(x,y)\mapsto x$. The product--pushforward identity
$(\Pi_A^{(n)})_\#P^{\otimes n}=((\Pi_A)_\#P)^{\otimes n}$
(\hyperref[a:B2]{Lemma~B.2}, via agreement on measurable rectangles and
Dynkin's $\pi$--$\lambda$ theorem) together with $(\Pi_A)_\#P=p_W$
(which uses only $q(\mathcal Y\mid x)=1$) gives
\[
Q^{(n)}_P\;=\;p_W^{\otimes n}
\qquad\text{for every }q\in I(\beta)\text{ and every }n,
\]
including $n=\infty$ (infinite product measure plus the cylinder
$\pi$-system; \hyperref[a:B3]{Proposition~B.3}). Marginalizing the never-observed label
erases $q$ from the model before any data are drawn. The unlabeled sample
is exactly ancillary for the conditional, and
$D_{\mathrm{TV}}(Q^{(n)}_P,Q^{(n)}_{P'})=0$ for any two worlds.

\emph{Lower bound.} The clipped functions
$\eta_\pm=\Pi_{[0,1]}(\eta_0\pm\beta)$ define Borel Markov kernels
$q_\pm$ with $q_\pm(\cdot\mid x)\in B_\beta(x)$ for \emph{every} $x$
(by the binary TV identity, $|\eta_+-\eta_0|=\min(1-\eta_0,\beta)\le\beta$
and symmetrically), so $q_\pm\in I(\beta)$, and since
$\eta_\pm=\overline\eta$ resp.\ $\underline\eta$,
$\theta(P_+)-\theta(P_-)=\mathbb E_\mu[w_\beta]$
(\hyperref[a:B6]{Lemma~B.6}). This is the full mean-gap diameter
$D_{\mathrm{mean}}(\beta;R)$, because every $q\in I(\beta)$ has
$\eta_q\in[\underline\eta,\overline\eta]$ $\mu$-a.e.\
(\hyperref[a:B7]{Lemma~B.7}). The risk of any randomized estimator
factorizes through $(Q^{(n)}_P,\theta(P))$ by change of variables for
pushforwards, so Le~Cam's two-point inequality
(\hyperref[a:B5]{Lemma~B.5}, proved there in a form covering arbitrary
Markov-kernel randomized rules) with $D_{\mathrm{TV}}=0$ gives, already
over the subfamily $\{P_-,P_+\}$,
\[
\sup_{P\in\{P_-,P_+\}}\mathbb E\big|\hat\theta_n-\theta(P)\big|
\;\ge\;\tfrac12\,\big|\theta(P_+)-\theta(P_-)\big|\,(1-0)
\;=\;\tfrac12\,\mathbb E_\mu[w_\beta].
\]
Enlarging the supremum to all of $\mathcal W$ preserves the bound.

\emph{Upper bound.} No data are needed. For every $q\in I(\beta)$,
$\theta(P_q)\in[\mathbb E_\mu\underline\eta,\,\mathbb E_\mu\overline\eta]$,
an interval of length $\mathbb E_\mu[w_\beta]$, so the zero-data midpoint
estimator $\hat\theta^\star=\mathbb E_\mu[(\underline\eta+\overline\eta)/2]$
has worst-case error $\tfrac12\mathbb E_\mu[w_\beta]$. The two bounds
meet, constant in $n$. Strict positivity when
$p_W(R\cap\{\beta>0\})>0$ follows from $w_\beta\ge\beta$ (valid also at
the both-clips boundary via $\beta\le1$), and the value is a functional
of the prior $(\beta,p_0,p_W)$ alone.

\emph{Corollaries.} Power $=$ size is immediate. Any test $\psi_n$ of
$q$ against $q'$ has
$\mathbb E_{P'}[\psi_n]=\mathbb E_P[\psi_n]
=\mathbb E_{p_W^{\otimes n}}[\psi_n]$, since all observation laws
coincide (\hyperref[a:B9]{Corollary~B.9}). For confidence sets: uniform coverage
$1-\alpha$ puts each of $q_\pm$ in $C_n$ with probability $\ge1-\alpha$
under the common observation law $p_W^{\otimes n}$, hence both with
probability $\ge1-2\alpha$ (union bound). On that event
$\operatorname{diam}_\theta C_n\ge D_{\mathrm{mean}}(\beta;R)$, and
taking (outer) expectations, since the diameter need not be measurable,
gives the $(1-2\alpha)\,D_{\mathrm{mean}}(\beta;R)$ bound for every $n$
(\hyperref[a:B10]{Corollary~B.10}, outer-expectation form).
\end{proof}

\medskip\noindent\emph{Formal version:} \hyperref[a:B8]{Theorem~B.8}
with \hyperref[a:B9]{Corollaries~B.9}--\hyperref[a:B10]{B.10},
Appendix~\ref{app:B}. The constant is unaffected even by an oracle revealing
$p_W$ itself, because within (B) the observation functional factors through
$p_W$ alone (\hyperref[a:B11]{Remark~B.11}).

The height of the wall is an assumption, not a measurement. An unlabeled
monitor that reports small drift is therefore reporting the chosen prior.

The wall shrinks only when the model includes structure that the label-free
observation map cannot see. The program contract (Online Appendix~1,
``the formal program'') enumerates these cases. \emph{(a) Labels:}
writing $\bar\beta_R=\mathbb E_\mu[\beta]$, $k$ anchor labels drawn
i.i.d.\ from the region pin its mean to
$\min\!\big(2\bar\beta_R,\ \tilde O(k^{-1/2})\big)$ (the pure
$k^{-1/2}$ collapse operating only once $k\gtrsim\bar\beta_R^{-2}$, below
which the budget still binds), while \emph{pointwise} recovery needs
anchors plus Lipschitz-in-$\phi$ smoothness at the nonparametric rate
$\tilde O(k^{-\alpha(R)})$, $\alpha(R)\asymp1/(2+d_R)$ with $d_R$ (the
analogue of $d$ in (A3)) the intrinsic dimension of the region's support
in $\phi$-space. In this branch, covariate geometry governs the rate. \emph{(b) A drift-transport prior}
$\beta(x)=\Psi(r_W(x))$, with $r_W$ the window's density ratio against the
source, ties the budget to the estimable density ratio,
letting unlabeled data constrain the effective wall (the coupling
that Section~\ref{sec:dichotomy} makes learnable from history). \emph{(c) (O3)
recurrence} plus a stationarity prior amounts to transferred anchors.
Part~III quantifies these label requirements.

Two canonical mechanism classes illustrate such a collapse of the wall.

\begin{theorem}[Canonical collapse classes]\label{thm:covshift}
\emph{(i) Covariate shift, prior-driven.} For the class
$\mathcal C_{\mathrm{cov}}=\{p\otimes p_0\}$, every observed $p_W$ pins
$I_{\mathcal C_{\mathrm{cov}}}(p_W)=\{p_0(\cdot\mid\cdot)\}$, so
$D^{\mathcal C_{\mathrm{cov}}}_\Delta(R)=0$ for every diagonal-vanishing
discrepancy, every $R$, every $\beta$, already at $n=0$. The frozen
conditional remains valid and the collapse class C5 applies.
\emph{(ii) Label shift, data-pinned.} Fix known class-conditionals
$\{\mu_y\}$, linearly independent in $M(\mathcal X)$ (equivalently, the
mixture map $\pi\mapsto\sum_y\pi(y)\mu_y$ is injective). Then for
$p_W$ in the class the target prior is identified from unlabeled
covariates alone by the linear moment system
$\pi_W=M^{-1}\mathbb E_{p_W}[T]$, for a bounded Borel moment vector $T$
with nonsingular $M_{jy}=\mathbb E_{\mu_y}[T_j]$
(the moment system of \citealp{lipton2018bbse}). The Bayes formula pins
$I_{\mathcal C_{\mathrm{ls}}}(p_W)=\{q_{\pi_W}\}$ and again
$D^{\mathcal C_{\mathrm{ls}}}_\Delta(R)=0$, with $\pi_W$ estimated from
$n$ unlabeled draws at rate
$\mathbb E\|\hat\pi-\pi_W\|_2\le
\|M^{-1}\|_{\mathrm{op}}\sqrt{|\mathcal Y|/(4n)}$
(Proposition~B.17).
\end{theorem}

\begin{proof}
\emph{(i) is pointwise pinning.} By \hyperref[a:B1]{Lemma~B.1} the
disintegration of a joint law over its covariate marginal is $p_W$-a.e.\
unique; since every member of $\mathcal C_{\mathrm{cov}}$ has conditional
$p_0$, whatever $p_W$ is observed the identified set is the singleton
$\{p_0(\cdot\mid\cdot)\}$, on which every diagonal-vanishing discrepancy
vanishes: $D^{\mathcal C_{\mathrm{cov}}}_\Delta(R)=0$ for every $R$ and
$\beta$, with no data required.

\emph{(ii) is a moment system.} First, linear independence of
$\{\mu_y\}$ is \emph{equivalent} to injectivity of the mixture map: any
relation $\sum_yc_y\mu_y=0$ among probability measures has
$\sum_yc_y=(\sum_yc_y\mu_y)(\mathcal X)=0$ automatically, so affine and
linear independence coincide, and a nonzero relation would produce two
distinct priors with equal mixtures (\hyperref[a:B14]{Lemma~B.14}).
Next, the moment vector exists. Evaluation functionals
$\nu\mapsto\nu(A)$ separate finite measures (Hahn--Jordan), so a
point-separating subset of the dual of the $|\mathcal Y|$-dimensional
span $V=\mathrm{span}\{\mu_y\}$ spans $V^*$, and one selects sets
$A_1,\dots,A_{|\mathcal Y|}$ whose evaluations form a basis. The
indicator moments $T_j=\mathbf 1_{A_j}$ are bounded Borel with
$M_{jy}=\mathbb E_{\mu_y}[T_j]$ nonsingular ($Mc=0$ forces
$\sum_yc_y\mu_y$ to be annihilated by a basis of $V^*$, hence $c=0$). Moment
linearity gives $\mathbb E_{p_W}[T]=M\pi$ for every mixture, which
inverts to $\pi_W=M^{-1}\mathbb E_{p_W}[T]$. Bayes pinning
(\hyperref[a:B15]{Lemma~B.15}, with everywhere-nonnegative
Radon--Nikodym versions $g_y$ so the formula defines a Markov kernel at
every $x$) exhibits $q_{\pi_W}(y\mid x)=\pi_W(y)g_y(x)/m_{\pi_W}(x)$ as
the $p_W$-a.e.-unique disintegration, so
$I_{\mathcal C_{\mathrm{ls}}}(p_W)=\{q_{\pi_W}\}$ and
$D^{\mathcal C_{\mathrm{ls}}}_\Delta(R)=0$
(\hyperref[a:B16]{Theorem~B.16}). The estimation rate is
\hyperref[a:B17]{Proposition~B.17}, where $\hat\pi=M^{-1}\hat b$ with
$\hat b_j$ the empirical moment has coordinate variances $\le1/(4n)$,
and Jensen plus the operator norm give the stated bound. The two
collapses thus have different currencies: (i) is prior-driven at $n=0$;
(ii) uses unlabeled data, at the $n^{-1/2}$ rate.
\end{proof}

\medskip\noindent\emph{Formal version:}
\hyperref[a:B13]{Proposition~B.13} and
\hyperref[a:B16]{Theorem~B.16} (with Lemmas~B.14--B.15 and
Proposition~B.17), Appendix~\ref{app:B}; the unfalsifiability claim
below is \hyperref[a:B20]{Proposition~B.20}.

The agnostic class itself cannot be falsified from unlabeled data.
If $p_W$ is consistent with label shift, the worlds $p_W\otimes q$
under (B), the covariate-shift world, and the label-shift world all induce
the same $p_W^{\otimes n}$ for every $n$, so every label-free test between
mechanism classes has power equal to size. The defensible posture is therefore
to read $D$ as the \emph{maximal honest wall}: an upper envelope over
mechanism priors that respect the budget
(any class with $I_{\mathcal C}(p_W)\subseteq I(\beta)$ satisfies
$D^{\mathcal C}\le D$; \hyperref[a:B18]{Theorem~B.18}), which the two
collapse classes lower to zero. The scoping matters. A \emph{budget-free}
prior can sit above $D$ (\hyperref[a:B19]{Remark~B.19}), so ``maximal''
is claimed only among budget-respecting priors.

\subsection{The Second Wall: What the Representation Discarded}
\label{sec:wallB}

The wall of Sections~\ref{sec:idset}--\ref{sec:irred} exists for any
representation. A second, distinct obstruction is contributed by the frozen
representation itself.

\begin{theorem}[Wall B and orthogonality]\label{thm:wallB}
Let $\eta_W(x)=\mathbb E_{p_W}[y\mid x]$,
$\bar\eta_W=\mathbb E[\eta_W\mid\mathcal F_\phi]$ (the $L^2(p_W)$
projection onto $\mathcal F_\phi=\sigma(\phi)$, the $\phi$-generated
$\sigma$-algebra), and define
$\mathrm{WallB}=\mathbb E[\operatorname{Var}(\eta_W\mid\mathcal F_\phi)]
=\|\eta_W-\bar\eta_W\|^2_{L^2(p_W)}$, equivalently in log-loss the
discarded conditional information
$\mathrm{WallB}_{\log}=I(x;y\mid\phi(x))=I(x;y)-I(\phi(x);y)\ge0$: the
part of the target that no readout of the frozen embedding can express,
labels or not (a functional of $(p_W,\phi)$ alone, in which $\beta$ and
$p_0$ never appear). Then for any deployed predictor
$\hat p=g\circ\phi$, with
$L^\star_x=\mathbb E_{p_W}[\eta_W(1-\eta_W)]$ the full-covariate Bayes
squared-error risk,
\[
\mathbb E_{p_W}(y-\hat p)^2-L^\star_x
=\underbrace{\|\eta_W-\bar\eta_W\|^2}_{\mathrm{WallB}\,\in\,\mathcal F_\phi^\perp}
+\underbrace{\|\bar\eta_W-g\circ\phi\|^2}_{\text{reachable error}\,\in\,\mathcal F_\phi},
\]
and the deployable Wall~A (a supremum over $I(\beta)$ of
$\mathcal F_\phi$-measurable acts) lives entirely in $\mathcal F_\phi$, so
the two walls are $L^2$-orthogonal \emph{subspace components}:
geometric orthogonality, not statistical independence. They are
certifiably distinct: there is a world with $\phi$ invertible
($\mathrm{WallB}=0$) and $\beta>0$ interior, where
$\mathrm{WallA}=\mathbb E[\operatorname{diam}B_\beta]>0$; and a
one-fiber world whose $\phi$-level conditional is identified by anchors,
where $\mathrm{WallA}=0$ and
$\mathrm{WallB}=\operatorname{Var}(\eta_W\mid\phi)>0$. They are coupled
only through $\mathrm{WallB}\le\mathbb E[\beta^2]$ (the drift ball is
centered at the fiber-constant $p_0$), so $\beta\to0$ caps both walls.
For injective $\phi$ (generic at $D=192$) the exact
$\mathrm{WallB}$ vanishes identically; the deployable object is the
resolution-$\rho$ deficit $\mathrm{WallB}_\rho$ defined below and
used in Section~\ref{sec:attribution-procedure}, tied to the readout's
operating scale, and every empirical Wall-B statement in this paper is at a stated
$\rho$.
\end{theorem}

\begin{proof}
Pure $L^2(p_W)$ projection geometry. Since
$\mathcal F_\phi\subseteq\mathcal F_x$, the square-integrable
$\phi$-measurable functions form a closed subspace of $L^2(p_W)$, and
$\bar\eta_W=\mathbb E[\eta_W\mid\mathcal F_\phi]$ is the orthogonal
projection of $\eta_W$ onto it. The residual $\eta_W-\bar\eta_W$ is
orthogonal to every $g\circ\phi$, since
$\mathbb E[(\eta_W-\bar\eta_W)\,g(\phi)]=0$ by the tower property. Any
deployed predictor $\hat p=g\circ\phi$ lies in the subspace, so
\[
\|\eta_W-g\circ\phi\|^2=\|\eta_W-\bar\eta_W\|^2
+\|\bar\eta_W-g\circ\phi\|^2
\qquad\text{(Pythagoras)},
\]
and adding the pointwise decomposition
$\mathbb E_{p_W}(y-\hat p)^2=L^\star_x+\|\eta_W-\hat p\|^2$ gives the
displayed identity. The deployable Wall~A is a supremum over $I(\beta)$
of differences of $\mathcal F_\phi$-measurable acts, hence lives
entirely in the subspace. The log-loss form is the data-processing
deficit $I(x;y)-I(\phi(x);y)=I(x;y\mid\phi(x))\ge0$.

\emph{Distinctness.} World (i): $\mathcal X=\{a,b\}$,
$\phi=\mathrm{id}$ (invertible), $p_0\equiv0.5$, $\beta=0.3$ interior.
Fibers are singletons, so $\mathrm{WallB}=0$, while
$\mathrm{WallA}=\mathbb E[\operatorname{diam}B_\beta]=0.6>0$. World
(ii): a single fiber ($\phi$ constant), $\eta_W(a)=0.9$,
$\eta_W(b)=0.1$ with equal weights, so $\bar\eta_W\equiv0.5$ and
$\mathrm{WallB}=\operatorname{Var}(\eta_W\mid\phi)=0.16$
($\mathrm{WallB}_{\log}=0.368$ nats). Anchors identifying the
$\phi$-level conditional pin the one deployable number, so
$\mathrm{WallA}=0$. This world must be realized by anchors, \emph{not}
by $\beta=0$, because of the coupling next.

\emph{Coupling.} The drift ball is centered at $p_0$, which is
fiber-constant ($\eta_0=h_0\circ\phi$ is $\mathcal F_\phi$-measurable),
so $|\eta_W-\eta_0|\le\beta$ pointwise, and since the conditional
variance minimizes conditional squared deviation over
$\mathcal F_\phi$-measurable centers,
$\operatorname{Var}(\eta_W\mid\mathcal F_\phi)
\le\mathbb E[(\eta_W-\eta_0)^2\mid\mathcal F_\phi]
\le\mathbb E[\beta^2\mid\mathcal F_\phi]$. Taking expectations,
$\mathrm{WallB}\le\mathbb E[\beta^2]$, so $\beta\to0$ caps both walls and
the orthogonality is geometric (subspace), not statistical independence.
\end{proof}

\medskip\noindent Two qualifications apply to the split:
it is exact \emph{pointwise-in-world} (the worst-case certificate
obeys only ``$\le$''), and $\phi$-resolution is a \emph{shared} lever, since
refining $\phi$ weakly decreases $\mathrm{WallB}$ (tower property) yet
can \emph{increase} deployable Wall~A on a sub-fiber region
(Theorem~\ref{thm:borrowed}); the reading of the two walls as mutually
inert interventions holds only for fiber-union regions.

\medskip\noindent\emph{Formal version:} the two-walls orthogonality
block of the program contract
(Online Appendix~1, ``the formal program'').
The deployable fiber-wise form is developed in Appendix~F of Online Appendix~1
(theorem cluster F.4: the reduction and dichotomy theorems).

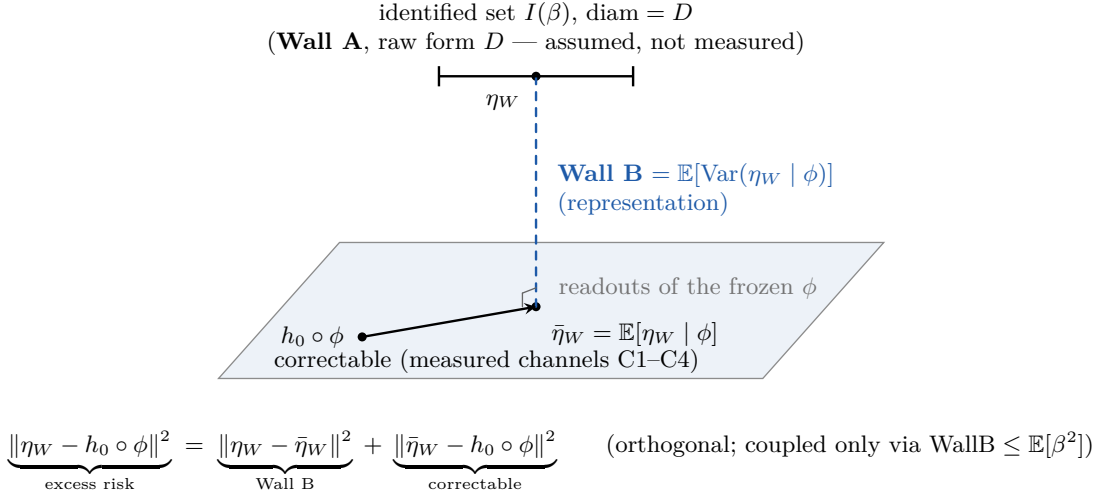
\begin{figure}[t]
\centering
\input{figures/fig_twowalls}
\caption{The two walls are orthogonal in the $L^2(p_W)$ projection
geometry (geometric orthogonality, not statistical independence), coupled
through $\mathrm{WallB}\le\mathbb E[\beta^2]$. The excess risk of a deployed
predictor splits Pythagorean-wise into an out-of-plane Wall-B component
(what no readout of the frozen $\phi$ can express, labels or not) and an
in-plane component reachable by the channels of Part~I. Wall~A is not a
distance but a \emph{diameter}: the identified set $I(\beta)$ of worlds
that no unlabeled quantity distinguishes, sitting at $\eta_W$ regardless
of the representation. Formal statements: the program contract and
Appendix~F of Online Appendix~1.}
\label{fig:twowalls}
\end{figure}

One measurement qualification comes from the program contract: the
resolution-$\rho$ deficit of the theorem's last sentence is
$\mathrm{WallB}_\rho=\mathbb E[\operatorname{Var}(\eta_W\mid
\phi\text{ coarsened to }\rho\text{-balls})]$, monotone in $\rho$ and
converging to $\mathrm{WallB}$ as $\rho\to0$.

The practical content is an \emph{attribution}
(Figure~\ref{fig:twowalls}): residual error decomposes
into a Wall-A share (acquire labels), a Wall-B share (change or fine-tune the
representation), and a correctable share (run the channels of Part~I).
Because the two walls respond to different interventions, conflating them
(as ``adaptation failed'' reports implicitly do) prescribes the wrong
remedy.

One further result makes Wall~A deployable on the frozen readout, whose
predictions are constant on the fibers of $\phi$:

\begin{theorem}[Borrowed certainty]\label{thm:borrowed}
Assume the standing conditions (P1)--(P3) of Appendix~F of Online
Appendix~1 (standard Borel
$\mathcal X$, binary $\mathcal Y$, Borel $\phi,h_0,\beta$ with $\beta$
\emph{not} required $\phi$-measurable, and a proper regular conditional
distribution $\rho_z$ on the fibers of $\nu=\phi_\#p_W$, not to be confused
with the resolution $\rho$) together with
the modeling choice (D-proj): the deployable act conditions on the whole
window, $\bar\eta_q=\mathbb E_{p_W}[\eta_q\mid\mathcal F_\phi]$ (the
frozen readout is one $\phi$-function of the deployment window, ignorant
of $R$). Then, writing
$\bar w_\beta(z)=\mathbb E_{p_W}[w_\beta\mid\phi=z]$, $\mu_R=p_W(\cdot\mid R)$,
and $\nu_R=\phi_\#\,\mu_R$:
\emph{(i)} for any region $R$,
\[
\mathrm{WallA}(\beta;R)
=\mathbb E_{z\sim\nu_R}\big[\,\mathbb E_{p_W}[w_\beta\mid\phi=z]\,\big],
\]
the inner conditional expectation running over the \emph{full} fiber
population (fiber points outside $R$ included), attained by the
explicit pair $(q_+,q_-)$;
\emph{(ii)} on fiber-unions $R=\phi^{-1}(B_R)$ the deployable and raw
walls coincide, $\mathrm{WallA}(\beta;R)=D_{\mathrm{mean}}(\beta;R)$;
\emph{(iii)} on sub-fiber regions the gap
\[
D_{\mathrm{mean}}(\beta;R)-\mathrm{WallA}(\beta;R)
=\mathbb E_{z\sim\nu_R}\big[\mathbb E_{\mu_R}[w_\beta\mid\phi=z]
 -\mathbb E_{p_W}[w_\beta\mid\phi=z]\big]
\]
is \emph{signed}: strictly positive (deployable wall smaller) when
$R$ over-selects high-$w_\beta$ points whose fiber-mates are pinned (the
drifting cell borrows certainty), negative when $R$ under-selects. The
resulting quantity is a label-free plug-in estimand given
$(\phi,\eta_0,\beta,p_W)$ (Corollary~F.5; its estimand identity is
proven, its coarse-$\phi$ estimator rate is graded partial).
\end{theorem}

\begin{proof}
The engine is the conditional Aumann-interval description of the
reachable acts (Lemma~F.1): with
$\eta_\pm=\Pi_{[0,1]}(\eta_0\pm\beta)$ and
$\bar\eta_\pm(z)=\int\eta_\pm\,d\rho_z$ ($\rho_z$ the proper regular
conditional distribution on the fiber over $z$), for $\nu$-a.e.\ $z$, as
$q$ ranges over $I(\beta)$ the fiber-average $\bar\eta_q(z)$ sweeps
\emph{exactly} the interval $[\bar\eta_-(z),\bar\eta_+(z)]$ of width
$\bar w_\beta(z)$, both endpoints attained by the explicit
everywhere-selections $q_\pm$. (Feasibility off a $p_W$-null set
transfers to $\nu$-a.e.\ fibers since
$\int\rho_z(N_q)\,d\nu=p_W(N_q)=0$; surjectivity comes from the convex
combinations $\lambda\eta_++(1-\lambda)\eta_-$, measurable selections of
the convex-valued $\Gamma_\beta$.)

\emph{(i), general $R$.} For $q,q'\in I(\beta)$ the lemma gives
$|\bar\eta_q(z)-\bar\eta_{q'}(z)|\le\bar w_\beta(z)$ $\nu$-a.e. Since
$\nu_R=\phi_\#\mu_R\ll\nu$, integrating against $\nu_R$ bounds the
deployable wall by $\int\bar w_\beta\,d\nu_R$, and the pair $(q_+,q_-)$
attains the bound because
$\bar\eta_{q_+}-\bar\eta_{q_-}=\bar w_\beta\circ\phi\ge0$ pointwise.
Neither step uses $R\in\mathcal F_\phi$, so
$\mathrm{WallA}(\beta;R)=\mathbb E_{z\sim\nu_R}
[\mathbb E_{p_W}[w_\beta\mid\phi=z]]$, the inner expectation over the
full fiber population.

\emph{(ii), fiber-unions.} When $R=\phi^{-1}(B_R)$, conditioning on the
$\mathcal F_\phi$-event $R$ leaves each within-fiber law $\rho_z$
($z\in B_R$) unchanged, and a tower computation collapses (i) to the raw
wall: $\int_{\phi^{-1}(B)}w_\beta\,d\mu_R
=\int_{\phi^{-1}(B)}\bar w_\beta\circ\phi\,d\mu_R$ for every Borel $B$,
whence $\mathbb E_{\mu_R}[\bar w_\beta\circ\phi]
=\mathbb E_{\mu_R}[w_\beta]=D_{\mathrm{mean}}(\beta;R)$.

\emph{(iii), sub-fiber regions.} Exactly that tower step fails:
$\bar w_\beta\circ\phi=\mathbb E_{p_W}[w_\beta\mid\mathcal F_\phi]$ is
\emph{not} a version of $\mathbb E_{\mu_R}[w_\beta\mid\mathcal F_\phi]$
when $R$ is sub-fiber, because conditioning on $R$ tilts the
within-fiber law. Disintegrating both walls under $\nu_R$ leaves
precisely the stated signed gap (the discrepancy between the
$R$-tilted and full-population fiber averages, with the sign of the
reweighting).

\emph{Check on the worked example}
(Corollary~F.3, Figure~\ref{fig:borrowed}): a two-point fiber
with $\eta_0\equiv0.5$, $\beta=(0.3,0)$, equal weights, $R$ the drifting
cell $\{a\}$. Pointwise the cell admits any $\eta_q(a)\in[0.2,0.8]$ (raw
wall $0.60$), but the pinned mate forces
$\bar\eta_q(z_0)=\tfrac12\eta_q(a)+0.25\in[0.35,0.65]$, so
$\mathrm{WallA}=0.30$, half the raw wall, with gap $0.30$ matching (iii).
\end{proof}

\medskip\noindent The multi-class structure stays in Appendix~F of
Online Appendix~1 (cluster F.4). The whole binary picture carries over verbatim to
TV/mean-gap discrepancies (concentric within-fiber balls share a common
worst direction), while non-monotone discrepancies such as $\ell_2$ incur
an additional strict directional-misalignment suppression.

\medskip\noindent\emph{Formal version:} the proof above transcribes
Theorems~F.2--F.3, Appendix~F of Online Appendix~1. Remaining
appendix-grade pieces are Corollary~F.5 (the label-free
estimand form) and
Proposition~F.4 with cluster F.4 (multi-class
structure). Figure~\ref{fig:borrowed} draws the theorem's worked example
(Corollary~F.3) to scale.

One hypothesis deserves separate mention because the phenomenon rests on it.
(D-proj) is semantic, not a theorem. Under the alternative act model
$\bar\eta_q=\mathbb E_{\mu_R}[\eta_q\mid\mathcal F_\phi]$ the inner and
outer measures coincide and borrowing vanishes for every $R$. But it is
forced by the definition of the deployable act itself (a frozen
readout emits one $\phi$-function of the whole deployment window and does
not know the analyst's region), so the appendix's classification of (D-proj) as a modeling choice is
conservative rather than a weakness (Remark~F.3). The same remark records the
shared-lever effect quoted under Theorem~\ref{thm:wallB}. As $\phi$
refines, $\mathrm{WallA}\to D_{\mathrm{mean}}(\beta;R)$, while coarsening
merges drifting cells with pinned mates and lowers the deployable wall.

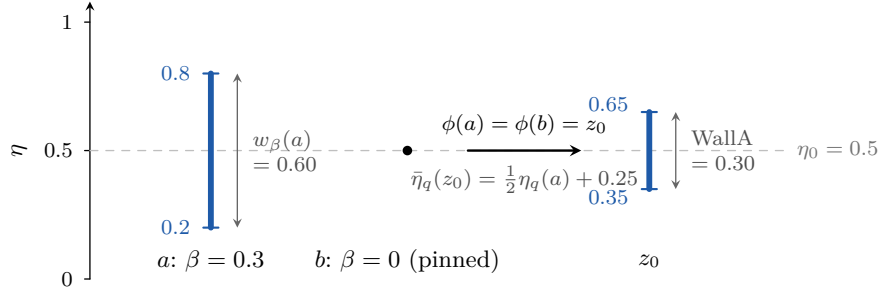
\begin{figure}[t]
\centering
\input{figures/fig_borrowed}
\caption{Borrowed certainty, drawn to the numbers of
Corollary~F.3. Left: pointwise, the drifting cell $a$ admits
any $\eta\in[0.2,0.8]$ (raw wall $0.60$) while its fiber-mate $b$ is
pinned at $0.5$. Right: the frozen readout emits one number per fiber,
so the deployable identified interval at $z_0$ is $[0.35,0.65]$, which gives
$\mathrm{WallA}=0.30$, half the raw wall. On fiber-unions deployable and
pointwise walls coincide (Theorem~F.2); on sub-fiber regions
the drifting cell borrows certainty from its pinned mates
(Theorem~F.3).}
\label{fig:borrowed}
\end{figure}

\subsection{Learning the Wall from History}
\label{sec:dichotomy}

The wall is a prior within one window, but deployments see many windows,
and occasionally labels. The question is whether the wall's height can be
\emph{learned} from labeled history and transferred forward, converting the
assumption into a measurement. A threshold calculation identifies one sufficient regular
branch and a corresponding, still incomplete lower-bound program. Both are
governed by one geometric quantity: how much drift mass sits near the
decision boundary.

Let $d_0$ be the (deduplicated) intrinsic dimension of the embedding
(instantiating the model parameter $d$ of assumption (A3)),
$\alpha_s$ the smoothness of the drift-response field, and $\gamma$ the
margin exponent of (Mgn$_\gamma$)
($\Pr[0<|\bar\eta-\eta_0|\le t]\lesssim t^\gamma$). Define
\[
a \;=\; \frac{\alpha_s}{2\alpha_s+d_0}\in\Big(0,\tfrac12\Big],
\qquad
\gamma^\star \;=\; \frac1{2a}-1\;=\;\frac{d_0}{2\alpha_s},
\qquad\text{so that}\quad a(1+\gamma^\star)=\tfrac12\ \text{identically}.
\]
The estimand is the covered-part wall
$\Dcov=2\int_\J\Psi\,d\nu$, $\Psi(r)=\E[\,|\bar\eta_J-\eta_0|\mid R=r\,]$
(here $R$ is the novelty index of Section~\ref{sec:assumptions}, not a
region; $\J$ is the covered ratio interval of (A2) and the subscript $J$
indexes the new window),
of a new anchor-free window with novelty law $\nu$ and representer
$w=d\nu/d\barmu$.

\begin{theorem}[Regular branch of the margin boundary]\label{thm:dichotomy}
Assume (A0)--(A7), (A5$'$), the interior regime, (T), and
(Mgn$_\gamma$), with $K$ historical windows of $k$ MCAR anchor labels
each ($N=Kk$), the fold-balanced anchor-level cross-fitting of
Definition~G.1.2 (including held-out Stage-1 evaluation), the Stage-1
sup-norm rate $\varepsilon_1\asymp(\log k/k)^a$, and the DML side conditions
(H1)--(H4) and (H6) of Theorem~G.1.4: overlap, cross-fitting, an
$o_P(N^{-1/2})$ product rate for the two nuisances, proportional target
size, and a sup-norm density-ratio rate (H6). If $\gamma>\gamma^\star$ (drift stays clear
of the boundary), the wall is estimable at the parametric rate: the
sign-corrected, cross-fitted one-step estimator built on the efficient
influence function $\varphi=2[\,w(b-\Psi)+(\Psi-\theta)\,]$, where
$b:=|\bar\eta_J-\eta_0|$ is the drift magnitude and $\theta:=\Dcov/2$, satisfies
\[
\sqrt N\big(\hat D_{\mathrm{cov}}-\Dcov\big)\Rightarrow
\mathcal N(0,V),\qquad
V=4\,\E_{\barmu}[w^2\sigma_b^2]+4\tau_{NT}\Var_\nu(\Psi),
\]
where $\sigma_b^2(r):=\Var(b\mid R{=}r)$ is the conditional variance of the
drift magnitude and $\tau_{NT}:=\lim N/m$ is the anchor-to-target size ratio
of (H4),
and no regular estimator improves on the rate. By the H\'ajek--Le~Cam
convolution and local-asymptotic-minimax theorems
\citep{hajek1970,lecam1986},
$\liminf_N N\Var(T_N)\ge V^\star_{A1}>0$ for every regular $T_N$, where
$V^\star_{A1}\le V^\star$ is the efficient constant on the
stationarity-constrained tangent (attainment by inverse-variance
pooling is graded partial in the appendix; $=V^\star$ only under
disjoint novelty coverage). For fixed $K$, (A6)
overlap, (T), and $\gamma>\gamma^\star$ are sufficient for this conclusion
modulo the Appendix-G regime-(I) qualifications (Theorem~I.3.2). At equality the exponents meet, but
the displayed Stage-1 rate retains a logarithmic factor and the present
argument does not establish $\sqrt N$-regularity.
\end{theorem}

\paragraph{Conditional lower-bound program (not part of the proved regular theorem).}
One open analytic input conditions the program's $\gamma=0$ results.
Gap~(2c) is a quantitative resummation estimate for the
permutation-mixture affinity of the derandomized lower-bound
construction of Appendix~K of Online Appendix~1: its band structure and first band
are proven, and exact-rational enumeration supports the general
estimate, which remains open. The extension to
$0<\gamma<\gamma^\star$ additionally requires a margin adaptation that
Remark~\ref{rmk:H8}(iv) leaves open. We therefore state only the
$\gamma=0$ consequences (no margin: (Mgn$_0$) is vacuous), explicitly
conditional on gap~(2c).

\begin{proposition}[Fixed-$K$ lower bound at $\gamma=0$; conditional on
gap~(2c)]\label{prop:fixedKlower}
Under the hypotheses of Theorem~\ref{thm:dichotomy} with $\gamma=0$ and
each fixed $K\ge1$, and conditional on the Appendix-K affinity
estimate~(2c) together with the cited transfer steps in Theorem~E.3,
\[
\inf_{\hat\theta}\ \sup\ \E|\hat\theta-\theta|
\ \gtrsim\ (N/K)^{-a}/\sqrt K\qquad(\text{up to logarithmic factors}).
\]
Because $K$ is fixed and $a<1/2$, this conditional floor is larger than
$N^{-1/2}$ and hence implies non-regularity. At $K=1$ it meets the
available Cai--Low upper bound up to logarithmic factors.
\end{proposition}

\begin{proposition}[Growing-$K$ rate at $\gamma=0$; conditional on
gap~(2c)]\label{prop:growingK}
Under the hypotheses of Theorem~\ref{thm:dichotomy} with $\gamma=0$ and
$k^{\delta}\le K\le\mathrm{poly}(k)$ for some fixed $\delta>0$, and
conditional on the graded affinity estimate~(2c) together with the cited
transfer steps in Theorem~E.3, the minimax rate is
\[
\inf_{\hat\theta}\ \sup\ \E|\hat\theta-\theta|\ \asymp\ (N/K)^{-a}
\quad(\text{up to polylog}),\qquad\text{\emph{not} } k^{-a}/\sqrt K.
\]
\emph{Conditionally on gap~(2c), accumulating more windows does not improve
the worst-case rate}:
independent-window averaging reduces the stochastic error by $\sqrt K$
unconditionally, but an adversary that drifts coherently across windows
keeps the non-smoothness bias at the single-window floor.
\end{proposition}

\begin{proof}[Proof of Theorem~\ref{thm:dichotomy}, given the Stage-1 inputs]
The estimand is \emph{linear} in $\Psi$:
$\theta=\langle\Psi,w\rangle_{L^2(\barmu)}$ with $\Dcov=2\theta$, so it
is pathwise differentiable with the two-sample efficient influence
function $\varphi$ above (source gradient $w(b-\Psi)$, target
gradient $\Psi-\theta$; Proposition~G.1.1). The estimator is
the cross-fitted one-step (DML2) of Definition~G.1.2: split every
window's anchors across all folds with equal per-window quotas, train
$(\hat\Psi^{(-\ell)},\hat w^{(-\ell)})$ and each Stage-1
$\hat\eta_j^{(-\ell)}$ off the evaluation anchors in fold $\ell$,
and read the held-out drift through the recovered sign,
$\hat b^{\mathrm{obs}}=\hat s\,(y-\eta_0)$, where $g:=\bar\eta_J-\eta_0$,
$s:=\operatorname{sign}g$, and $\hat s:=\operatorname{sign}(\hat\eta_J-\eta_0)$.

\emph{Step 1: sign recovery; the threshold computation.} On the Stage-1
sup-norm event $\|\hat\eta_J-\bar\eta_J\|_\infty\le\varepsilon_1$
(probability $\ge1-\sum_jk_j^{-2}$,
\hyperref[lem:C11prime]{Lemma~C.11$'$}), the \emph{straddle} holds
pointwise: a sign flip forces $0<b\le\varepsilon_1$
(Lemma~\ref{lem:H1}). Hence under (Mgn$_\gamma$) the weighted sign
error obeys
$\E[b\,\mathbf 1(\hat s\ne s)]\le\varepsilon_1\Pbar(0<b\le\varepsilon_1)
\le C\varepsilon_1^{1+\gamma}$, and, via the pointwise identity
$(\hat s-s)g=-2b\,\mathbf 1(\hat s\ne s)$, the \emph{entire} outcome
bias of $\hat b^{\mathrm{obs}}$ is
$\|\beta_1^{\mathrm{sign}}\|_{L^1(\barmu)}\le2C\varepsilon_1^{1+\gamma}$.
Wherever the sign is right, the raw residual is conditionally unbiased,
so no Jensen/smoothing bias is ever incurred. With
$\varepsilon_1\asymp k^{-a}$ (logs absorbed by the strict inequality)
and $k=N/K$, $K$ fixed,
\[
\varepsilon_1^{1+\gamma}=o(N^{-1/2})
\iff(1+\gamma)a>\tfrac12
\iff\gamma>\gamma^\star=\tfrac1{2a}-1=\tfrac{d_0}{2\alpha_s}.
\]
This computation \emph{is} the threshold.

\emph{Step 2: one-step decomposition.} Write
$\hat D_{\mathrm{cov}}-\Dcov
=2[(\nu_m-\nu)\hat\Psi+(\Pbar_N-\Pbar)(\hat w(\hat b^{\mathrm{obs}}
-\hat\Psi))]+2R_2$ with
$R_2=R_{\mathrm{prod}}+R_{\mathrm{ratio}}+R_{\mathrm{stage1}}^{\mathrm{sign}}$.
The Neyman-orthogonal \emph{product} remainder is
(Lemma~G.1.3)
\[
\tfrac12\,\E[\hat D_{\mathrm{cov}}\mid\hat\Psi,\hat w]-\theta
=\int(\hat\Psi-\Psi)(w-\hat w)\,d\barmu,\qquad
|R_{\mathrm{prod}}|\le\|\hat\Psi-\Psi\|_{L^2(\barmu)}
\|\hat w-w\|_{L^2(\barmu)}
\]
(double robustness in $(\Psi,w)$), so the cube-root shape rate of
the isotonic $\hat\Psi$ enters only \emph{multiplied} by the fast,
one-dimensional representer error, and (H3) gives
$R_{\mathrm{prod}}=o_P(N^{-1/2})$; (H6) handles $R_{\mathrm{ratio}}$;
Step 1 gives
$|R_{\mathrm{stage1}}^{\mathrm{sign}}|\le\Gcov
\|\beta_1^{\mathrm{sign}}\|_{L^1}=o(N^{-1/2})$ above threshold.

\emph{Step 3: CLT.} Cross-fitting replaces any Donsker condition.
\hyperref[lem:G14prime]{Lemma~G.14$'$} proves the fold-conditional
empirical-to-population step under the window-stratified (non-i.i.d.)
fold-balanced design with equal quotas $k_j\equiv k$. The published i.i.d.\ DML
theorem \citep[Thm.~3.1]{chernozhukov2018dml} does not apply as stated, and only its Step-2 device is
imported. What remains is the empirical process of the fixed $\varphi$
on two independent bounded mean-zero averages. Lindeberg--Feller
\citep[Prop.~2.27]{vandervaart1998} plus Slutsky give
$\sqrt N(\hat D_{\mathrm{cov}}-\Dcov)\Rightarrow\mathcal N(0,V)$, where
the variance identification $V=4\E_{\barmu}[w^2\sigma_b^2]
+4\tau_{NT}\Var_\nu(\Psi)$ \emph{requires} stationarity (A1): with
nonvanishing per-window score means the stratified variance is strictly
below the mixture $V$ (the CI constant is conservative; the estimator
stays centered on the pooled $\theta$).

\emph{Step 4: no regular estimator improves.} The stratified anchor
experiment is LAN in the interior regime, so the H\'ajek--Le~Cam
convolution and LAM theorems applied on the stationarity-constrained
tangent $\mathcal T_{A1}$ give
$\liminf_NN\Var(T_N)\ge V^\star_{A1}>0$ for every regular $T_N$
(Theorem~G.3.2; for the non-i.i.d.\ experiment via the convolution
theorem of \citealp{mcneney2000convolution}).

\emph{Remainders left to the appendix, as noted there:} the
efficient \emph{constant} $V^\star_{A1}$ under overlapping novelty
coverage (and the GMM/inverse-variance pooling attaining it) is
partial; and the ratio channel keeps (H6) at
$\varepsilon_r=o(N^{-1/2})$ unconditionally,
$\varepsilon_r^2$ only under (A5$'$) plus the conjectural DML-in-$r$
orthogonalization (Conjecture~G.2.3).
\end{proof}

\medskip\noindent\textbf{The conditional lower-bound program
(Appendix~\ref{app:H}; Appendices~I and~K of Online Appendix~1).}
The candidate non-regularity mechanism lives entirely in the
labels$\to|\cdot|$ map. The constructions behind
Propositions~\ref{prop:fixedKlower} and~\ref{prop:growingK} have three
stages; their dependence on gap~(2c) is carried throughout.

\emph{Single window} (Theorems~E.3 and~\ref{thm:H6}). By LP
duality of the moment problem, two symmetric priors $\nu_0,\nu_1$ on
$[-1,1]$ matching moments to degree $D$ can differ in $\E|v|$ by
$2\delta_D$, $\delta_D=\beta_*/D\,(1+o(1))$ with Bernstein's constant
$\beta_*\approx0.2802$, the extremal pair sitting on the Chebyshev
alternation points of the best degree-$D$ approximant of $|x|$
(Proposition~\ref{prop:H5}, following \citealp{cailow2011}; see also
\citealp{lepski1999}). Partition the $d_0$-dimensional
$\phi$-support into $m_h=h^{-d_0}$ cells of side $h=k^{-1/(2\alpha_s+d_0)}$.
On each cell set the drift $g=Av_c$ at H\"older-maximal amplitude
$A=h^{\alpha_s}=k^{-a}$, with i.i.d.\ cell heights $v_c\sim\nu_\iota$, $\iota\in\{0,1\}$.
With $A^2\cdot(\text{anchors per cell})=\Theta(1)$ and
$D\sim\log k/\log\log k$ so that $(D{+}1)!\gtrsim m_h$, the mixture
$\chi^2$ stays $O(1)$ while the functional gap is
$A\cdot2\delta_D\asymp k^{-a}(\log\log k/\log k)$. Conditional on the
Appendix-K affinity estimate, Le~Cam's two-point method gives the $K=1$
floor (the sharp constant and the
uniform-over-cells Bernoulli$\leftrightarrow$Gaussian transfer are
cited results, \citealp{cailow2011,nussbaum1996}; Theorem~E.3 carries
this external dependence).

\emph{Fixed $K$} (Theorem~\ref{thm:H7}). Place the $K$ windows on
\emph{disjoint} novelty cells (Stage-1 cannot pool), each of $\nu$-mass
$1/K$, and let each window be active with probability $p$ under
$\Lambda_0$ and $p'$ under $\Lambda_1$, $|p'-p|=c/\sqrt K$, the active
state being the fuzzy prior above, coherently oriented. The affinities
tensorize, $\chi^2(\Lambda_1\Vert\Lambda_0)=\prod_j(1+\chi^2_j)-1\le
e^{\sum_j\chi^2_j}-1$ with $\sum_j\chi^2_j\asymp K\cdot(c^2/K)\cdot O(1)=O(1)$, so $\mathrm{TV}\le\tfrac12$, while the
functional separates by $(p'-p)\cdot2\delta_DA\asymp(N/K)^{-a}/\sqrt K$ up
to the logarithmic factor $\log\log k/\log k$ of the single-window case.
Conditional on gap~(2c), Le~Cam gives the fixed-$K$ floor, which is enough
for non-regularity, since
$(k^{-a}/\sqrt K)/N^{-1/2}=k^{1/2-a}\to\infty$ and no logarithmic factor
changes that. Assouad is inapplicable
here, because the functional sees the sign pattern $\tau\in\{0,1\}^K$ only through the
scalar $\sum_j\tau_j$, a Hamming \emph{upper} bound where Assouad needs a lower
one (Remark~\ref{rmk:H8}), which is exactly why this bound carries
the $\sqrt K$ deficit.

\emph{Growing $K$: degree-tuning}
(Theorem~I.2.1 and Proposition~I.2.2). The
amplitude-tuned bound loses $\sqrt K$ in the gap. The general-$K$ bound
keeps the amplitude and tunes the \emph{degree}. At the Cai--Low
bandwidth the per-window affinity $\chi^2_j\asymp m_h/(D+1)!$ falls
\emph{super-exponentially} in the matched degree while the gap
$\asymp A\,\delta_D\asymp A/D$ falls only \emph{polynomially}. Choosing
$(D{+}1)!\asymp m_hK/c_0$ with $c_0\le\ln2$ drives $\chi^2_j\asymp c_0/K$, so
that on independent windows
$(1+\chi^2_j)^K-1\le e^{c_0}-1\le1$ and $\mathrm{TV}\le\tfrac12$
(Ingster--Suslina tensorization; \citealp{ingster2003nonparametric}),
while raising the degree from $D_1$ ($(D_1{+}1)!\asymp m_h$) costs the
gap only the factor $1+O(\log K/\log k)$. Conditional on the
permutation-mixture affinity estimate, placing the \emph{same}
moment-matched prior coherently on all $K$ windows therefore retains a
gap of order $k^{-a}/\mathrm{polylog}$ at $\mathrm{TV}\le\tfrac12$, and
the functional concentrates (variance $\asymp k^{-2a}/(Km_h)\ll$ the
squared gap), so the two-fuzzy-hypotheses Le~Cam bound
\citep[Thm.~2.15]{tsybakov2009introduction} yields the floor
$r_k/\mathrm{polylog}(k)$, $r_k:=(k/\log k)^{-a}$ ($=(N/K)^{-a}$ up to
the $(\log k)^a$ factor), for
$k^\delta\le K\le\mathrm{poly}(k)$,
which would refute the $k^{-a}/\sqrt K$ averaging gain. The competing Cauchy--Schwarz
ceiling $\chi^2_j=(\delta_j/r_k)^2$, $\delta_j$ the per-window functional
gap, which would force the
$r_k/\sqrt K$ cap, holds only along the amplitude curve and is not a
valid universal lower-bound cap (Proposition~I.2.2).

\emph{Admissibility of the coherent construction.} The construction's
admissibility is established in
Appendix~K of Online Appendix~1, at every $K$. A naive i.i.d.-heights construction is
\emph{inadmissible} at every $K\ge1$ (its realized response curve is
a.s.\ non-monotone); the admissible version is a derandomized,
phase-complementary construction with \emph{exactly constant} $\Psi$.
The derandomization trades the i.i.d.\ tensorization for a
permutation-mixture affinity bound, proven modulo the single graded
gap (2c). Because the derandomization is what makes the construction admissible
already at $K=1$, Propositions~\ref{prop:fixedKlower}
and~\ref{prop:growingK} carry the same conditioning. If (2c) failed
entirely, both conditional floors would lapse; the bracket
$[r_k/\sqrt K,\,r_k]$ of Remark~\ref{rmk:H8}(ii) would then remain only as
the candidate range for the growing-$K$ rate, with neither end established.

\medskip\noindent\emph{Reconciliation and refinements.} The
certified-sign regime of Appendix~G of Online Appendix~1, where the signed drift
is observed and the Stage-1 bias vanishes for every $\gamma$, is the
\emph{perfect-margin endpoint} of the regular-side calculation. Its side condition
$\varepsilon_1^{1+\gamma}=o(N^{-1/2})$ is algebraically identical to
$\gamma>\gamma^\star$, and its ``outcome-bias hypothesis fails for all
$K\ge1$'' is exactly the no-margin instance $\gamma=0$
(Theorem~\hyperref[thm:H11]{H.11}). Below threshold the candidate
rate-relevant achiever is the polynomial-debiased $U$-statistic of
Proposition~\hyperref[prop:H10]{H.10}, which confines poly-debiasing to
the near-crossing layer (mass $\lesssim\varepsilon_1^\gamma$) and runs the
linear one-step on the bulk. It is \emph{not} the naive plug-in
$|\hat\eta_J-\eta_0|$, whose one-signed Jensen bias is coherent on every
window and therefore never averages. Two refinements
(Corollary~\ref{cor:H4}(b); Theorem~I.3.2 of Online Appendix~1): for
\emph{growing} $K=K_N$ with
$\log K/\log N\to\lambda\in(0,\tfrac12)$, the parametric branch is
established only for $\gamma>\gamma^\star(\lambda)
=\tfrac1{2a(1-\lambda)}-1>\gamma^\star$, leaving the interval
$(\gamma^\star,\gamma^\star(\lambda))$ with a proven $N^{-1/2}$ lower
bound but no proven attaining estimator; and rates of non-smooth
functionals routinely jump at phase boundaries, so continuity across
$\gamma^\star$ is not a valid selection principle. The below-threshold
rate is fixed by the two-sided bracket, not by matching the parametric
branch at the threshold. Finally, the deployment reading of the
$\sqrt K$ (Remark~I.3.3): the gain is real but sits
on the wrong error component. It averages the stochastic error
unconditionally and helps on \emph{bias-incoherent} recurrences, while a
worst-case adversary drifts coherently and pays no averaging at all. Regime
recurrence therefore supplies a $\sqrt K$ lever \emph{under an
incoherence prior on the recurrences}, degrading to the bare per-window
price under coherent drift. Section~\ref{sec:map}
examines how the resulting threshold behaves under stream-level proxies and
an explicit sensitivity scenario; it does not identify a worst-case price
for every data set.

\noindent\emph{Formal statements and proofs:} the regular branch is
Theorem~G.1.4 (DML2 CLT) with the matching regularity bound
Theorem~G.3.2 and its regime-(II) extension
Theorem~\ref{thm:H3}. The conditional $\gamma=0$ lower-bound statements
are Theorem~\ref{thm:H7} (fixed $K$) and Theorem~I.2.1,
Proposition~I.2.2, and Theorem~I.3.1 (general $K$). The positive-margin adaptation below the threshold remains
open. Appendix~\ref{app:H} and Appendices~G, I, and~K of Online Appendix~1
give the full status record.

\emph{Status.} The upper rate, the threshold identity, and the
sign-recovery reduction behind the magnitude case are theorems under
their stated hypotheses. One hypothesis deserves separate
mention: Theorem~\ref{thm:dichotomy} \emph{assumes} the sup-norm
density-ratio rate (H6), and the program that would establish (H6) from
primitives is itself graded partial in Online Appendix~1, so the theorem
certifies the implication, not the availability of its premise. The
efficient \emph{constant} under overlapping windows and the second-order
behavior of the density-ratio channel are graded partial in the
appendix. The growing-$K$ admissibility step of the coherent construction is
settled exactly (Appendix~K of Online Appendix~1, via the derandomized
exactly-constant construction).
What remains there is the single graded affinity estimate (gap (2c)); it
conditions the $\gamma=0$ lower bounds at every $K$, fixed and growing
alike (Propositions~\ref{prop:fixedKlower} and~\ref{prop:growingK}), and
the margin-$\gamma$ adaptation of
the fixed-$K$ bound remains open (Remark~\ref{rmk:H8}).

Estimating the wall requires estimating $|\bar\eta-\eta_0|$, and the
absolute value is non-smooth where drift crosses the boundary. The margin
condition keeps this kink rare enough for sign recovery to restore a smooth
problem. For deployment, the wall is learnable from history at low label
cost when drift stays away from decisions. The conditional lower-bound program explains why
crossings are the difficult case, but it does not yet prove that conclusion
throughout the positive-margin sub-threshold region.

\subsection{Proxy Estimates and Sensitivity Analysis on Twenty-Three Data Sets}
\label{sec:map}

The threshold depends on three quantities that are difficult to estimate
jointly: intrinsic dimension $d_0$, drift smoothness $\alpha_s$, and margin
exponent $\gamma$. We therefore separate two evidence levels. On eight
TabReD streams, labels permit stream-level proxy estimates of all three
quantities, although the smoothness and margin estimates are noisy. On
fifteen further shift benchmarks, labels for the relevant temporal drift
field are unavailable; those data sets contribute only $d_0$ and are shown
under the favorable rough-field scenario $(\alpha_s,\gamma)=(1,1)$.
Figure~\ref{fig:map} visualizes both levels and Table~\ref{tab:J-map} records
the inputs. Appendix~\ref{app:J} gives the protocol and its limitations.

\begin{figure}[t]
\centering
\includegraphics{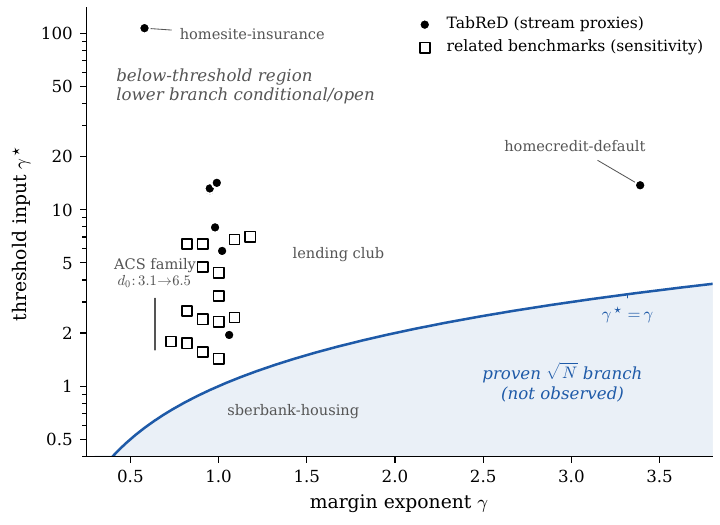}
\caption{Proxy estimates and sensitivity scenarios relative to the
margin boundary (log vertical scale). Filled circles use the stream-level
$\widehat\gamma$ and $\widehat\alpha_s$ available for eight TabReD streams.
Open squares are not measured phase locations: they fix
$(\alpha_s,\gamma)=(1,1)$ and plot $\gamma^\star=d_0/2$ for fifteen related
benchmarks. Horizontal dodging of the squares (at most $0.27$) is purely
for legibility. The shaded region $\gamma>\gamma^\star$ is the strict
sufficient branch for $\sqrt N$ regularity; equality is not classified by
the present theory. All proxy points and all sensitivity scenarios fall on
the below-threshold side, conditional on their stated inputs; the lower
branch there is conditional or open, depending on $\gamma$.}
\label{fig:map}
\end{figure}

\begin{table}[t]
\centering
\scriptsize
\begin{tabular}{llrrrrl}
\toprule
data set & domain & $d_0$ & $\alpha_s$ input & $\gamma$ input & $\gamma^\star$ & status \\
\midrule
\textsc{ecom-offers}        & TabReD (cls) & 3.61  & 0.14 & 0.95 & 12.89 & proxy below \\
\textsc{sberbank-housing}   & TabReD (reg) & 3.01  & 0.77 & 1.06 & 1.95  & proxy below \\
\textsc{weather}            & TabReD (reg) & 11.10 & 0.95 & 1.02 & 5.84  & proxy below \\
\textsc{cooking-time}       & TabReD (reg) & 11.12 & 0.70 & 0.98 & 7.94  & proxy below \\
\textsc{homesite-insurance} & TabReD (cls) & 10.69 & 0.05 & 0.58 & 106.90& proxy below \\
\textsc{homecredit-default} & TabReD (cls) & 13.76 & 0.50 & 3.39 & 13.76 & proxy below \\
\textsc{delivery-eta}       & TabReD (reg) & 16.63 & 0.63 & 0.95 & 13.20 & proxy below \\
\textsc{maps-routing}       & TabReD (reg) & 17.61 & 0.62 & 0.99 & 14.20 & proxy below \\
\midrule
\textsc{cicids}             & security & 2.86 & [1] & [1] & 1.43 & sensitivity below \\
\textsc{acs\_employment}    & census   & 3.11 & [1] & [1] & 1.56 & sensitivity below \\
\textsc{acs\_unemployment}  & census   & 3.49 & [1] & [1] & 1.75 & sensitivity below \\
\textsc{unsw\_nb15}         & security & 3.58 & [1] & [1] & 1.79 & sensitivity below \\
\textsc{ieee\_fraud}        & fraud    & 4.64 & [1] & [1] & 2.32 & sensitivity below \\
\textsc{acs\_publiccoverage}& census   & 4.78 & [1] & [1] & 2.39 & sensitivity below \\
\textsc{bike\_sharing}      & reg      & 4.90 & [1] & [1] & 2.45 & sensitivity below \\
\textsc{credit\_card\_fraud}& fraud    & 5.32 & [1] & [1] & 2.66 & sensitivity below \\
\textsc{acs\_income}        & census   & 6.49 & [1] & [1] & 3.25 & sensitivity below \\
\midrule
\textsc{diabetes\_readmission} & medical & 8.79  & [1] & [1] & 4.40 & sensitivity below \\
\textsc{brfss\_diabetes}       & medical & 9.46  & [1] & [1] & 4.73 & sensitivity below \\
\textsc{sepsis}                & medical & 12.75 & [1] & [1] & 6.38 & sensitivity below \\
\textsc{baf}                   & fraud   & 12.81 & [1] & [1] & 6.41 & sensitivity below \\
\textsc{lending\_club}         & credit  & 13.55 & [1] & [1] & 6.78 & sensitivity below \\
\textsc{mimic\_iv\_mortality}  & medical & 14.05 & [1] & [1] & 7.03 & sensitivity below \\
\bottomrule
\end{tabular}
\caption{Inputs to Figure~\ref{fig:map}. In the top block,
$\widehat\alpha_s$ and $\widehat\gamma$ are noisy stream-level proxies and
$\gamma^\star=d_0/(2\widehat\alpha_s)$ is recomputed from the displayed
values. In the lower blocks, bracketed values are assumptions for a
sensitivity analysis, so $\gamma^\star=d_0/2$ by construction. All fifteen
related benchmarks use deduplicated $d_0$; the six high-$d_0$ rows are
$99.97$--$100\%$ unique, so their correction is bounded by subsample noise
(Appendix~\ref{app:J}). ``Below''
means $\gamma<\gamma^\star$ for the stated inputs, not a population-level
phase classification.}
\label{tab:J-map}
\end{table}

For the TabReD streams, the filled-circle locations should be read as
estimated inputs, not exact coordinates. Six of the eight
$\widehat\gamma$ values lie in $[0.95,1.06]$, while the two-range slope check of
Appendix~\ref{app:J} disagrees on five streams. Nevertheless, if the
rough-field bound $\alpha_s\le1$ is accepted and the measured margins are
representative, then $2\alpha_s\widehat\gamma<d_0$ for all eight streams;
the closest is \textsc{sberbank-housing}, with a favorable left side of
$2.12$ (at $\alpha_s=1$; $1.63$ at the measured $\widehat\alpha_s=0.77$)
against $d_0=3.01$. This is conditional evidence for the difficult
branch, not an uncertainty-quantified phase assignment. The known
biases run in the favorable direction for this reading: the $\alpha_s$
proxy over-states smoothness (Appendix~\ref{app:J}), inflating the left
side, so the below-threshold comparison is conservative.

For the fifteen related benchmarks, neither $\alpha_s$ nor $\gamma$ is
observed. Their open squares answer a narrower question, namely whether, under the
generic transversal-margin hypothesis $\gamma=1$ and the favorable roughness
value $\alpha_s=1$, the estimated dimension alone would place the scenario
below threshold. The answer is yes for every displayed $d_0>2$. Changing either
assumption can change that answer, so these rows cannot support the claim
that all fifteen data sets have measured non-regular walls. All fifteen
were re-estimated after deduplication: the low-uniqueness members of the
nine re-measured data sets moved materially (two left the low-dimensional
reading entirely), while the six
high-$d_0$ rows are duplicate-free to four decimals ($99.97$--$100\%$
unique) and move by at most $0.10$, comparable to the TabReD-8
robustness check (at most $0.08$).

The study therefore supplies a consistent empirical warning, not a
population law: the observed TabReD proxies and a favorable
sensitivity scenario do not populate the regular branch. Establishing an
empirical phase map would require repeated labeled windows, uncertainty
intervals for $(d_0,\alpha_s,\gamma)$, and sensitivity to the dimension and
variogram estimators. The semi-synthetic experiment of the next subsection
supplies what the stream proxies cannot: controlled inputs whose effective
exponents can be calibrated and checked numerically. The ACS ordering (employment, unemployment, public
coverage, income) remains a useful descriptive pattern in $d_0$, but is not
by itself a validation of the theoretical threshold.

\subsection{Semi-Synthetic Illustration of the Threshold Mechanism}
\label{sec:phasesim}

The stream proxies of Section~\ref{sec:map} cannot manipulate
$(d_0,\alpha_s,\gamma)$; a generated control can vary them separately.
Covariates are
$Z\sim\mathrm{Unif}[0,1]^{d_0}$ with $\phi$ the identity, so the
doubling dimension is $d_0$. The signed drift is
$g(z)=\mathrm{sign}(W(z_1))\,|W(z_1)|^{1/\gamma}$, a function of the first
coordinate $z_1$ only,
where $W$ is a normalized Weierstrass field of exponent
$\alpha_0=\gamma\alpha_s\le1$. The construction targets a
H\"older-$\alpha_s$ drift magnitude and margin exponent $\gamma$; finite-grid
self-tests report the corresponding effective slopes rather than certify
the asymptotic exponents. Labels are
Bernoulli with $\eta_j=\tfrac12+0.45\,g$ over four windows. The novelty
representer is $w\equiv1$ and known, which removes the
representer/ratio channel entirely; the simulation therefore
isolates the sign-recovery mechanism that creates the threshold, not
the DML nuisance theory. The estimator is the specified two-stage procedure:
per-window anchor split, Stage-1 $\kappa$-NN with the non-adaptive
tuning $\kappa\asymp k^{2\alpha_s/(2\alpha_s+d_0)}$ (up to the
logarithmic factor) of \hyperref[lem:C11prime]{Lemma~C.11$'$}, cross-fitted sign recovery, sign-corrected residual
mean. Four arms vary one parameter each across
$\gamma^\star=d_0/(2\alpha_s)$. The design, prediction, and falsification criterion were fixed before
the sweep.

\begin{center}\small
\begin{tabular}{lccccll}
\toprule
arm & $d_0$ & $\alpha_s$ & $\gamma$ & $\gamma^\star$ & predicted & measured \\
\midrule
A (above) & 1 & 0.75 & 1 & 0.67 & $-0.500$ & $-0.509\pm0.024$ \\
B (below) & 8 & 0.75 & 1 & 5.33 & $-0.158$ & $-0.138\pm0.002$ \\
C (below) & 1 & 0.375 & 0.75 & 1.33 & $-0.375$ & $-0.305\pm0.008$ \\
D (above) & 1 & 0.375 & 2 & 1.33 & $-0.500$ & $-0.515\pm0.023$ \\
\bottomrule
\end{tabular}
\end{center}

Slopes are weighted log-log fits of RMSE against $N$ over the largest
five sample sizes ($N$ up to $131{,}072$; $100$--$400$ seeds per
point). The above-threshold arms are variance-dominated with slopes at
the regular $-\tfrac12$; the below-threshold arms are bias-dominated
and sit $181$ and $24$ standard errors away from $-\tfrac12$, with
point values on the same side of $-\tfrac12$ as, though not at, the
candidate bias-rate prediction
$-a(1+\gamma)$. The experiment illustrates the sign-recovery mechanism; it
does not establish the conditional lower bound. Arms A/B
differ only in $d_0$, arms C/D only in $\gamma$: each parameter alone moves
the scaling across the boundary. The generating code, its outputs, and the pre-registered design are
part of the reproduction deposit (Section~\ref{sec:setup}).

%% file: figures/fig_twowalls.tex
\begin{tikzpicture}[>=stealth, line cap=round]
  \definecolor{figblue}{HTML}{1E5AA8}
  \fill[figblue!8] (0,0) -- (7.2,0) -- (8.8,1.8) -- (1.6,1.8) -- cycle;
  \draw[black!45, line width=0.5pt] (0,0) -- (7.2,0) -- (8.8,1.8) -- (1.6,1.8) -- cycle;
  \node[black!55, font=\footnotesize, anchor=east] at (8.05,1.20)
    {readouts of the frozen $\phi$};
  \coordinate (etabar) at (4.2,0.95);
  \coordinate (eta)    at (4.2,4.0);
  \coordinate (hzero)  at (1.9,0.55);
  \fill (etabar) circle (1.7pt);
  \node[font=\footnotesize, anchor=north west, xshift=2pt, yshift=-2pt] at (etabar)
    {$\bar\eta_W=\mathbb E[\eta_W\mid\phi]$};
  \fill (eta) circle (1.7pt);
  \node[font=\footnotesize, anchor=north east, xshift=-2pt, yshift=-3pt] at (eta) {$\eta_W$};
  \fill (hzero) circle (1.7pt);
  \node[font=\footnotesize, anchor=east, xshift=-3pt] at (hzero) {$h_0\circ\phi$};
  \draw[dashed, line width=0.9pt, figblue] (eta) -- (etabar);
  \node[figblue, font=\footnotesize, anchor=west, align=left] at (4.35,2.5)
    {\textbf{Wall B} $=\mathbb E[\operatorname{Var}(\eta_W\mid\phi)]$\\ (representation)};
  \draw[black!60, line width=0.5pt] (4.02,1.13) -- (4.20,1.20);
  \draw[black!60, line width=0.5pt] (4.02,1.13) -- (4.02,0.95);
  \draw[->, line width=0.9pt] (hzero) -- (etabar);
  \node[font=\footnotesize, anchor=north] at (3.55,0.53)
    {correctable (measured channels C1--C4)};
  \draw[|-|, line width=0.9pt] (2.9,4.0) -- (5.5,4.0);
  \node[font=\footnotesize, align=center, anchor=south] at (4.2,4.15)
    {identified set $I(\beta)$, $\operatorname{diam}=D$\\
     (\textbf{Wall A}, raw form $D$ --- assumed, not measured)};
  \node[font=\footnotesize, anchor=north] at (4.4,-0.55)
    {$\underbrace{\|\eta_W-h_0\circ\phi\|^2}_{\text{excess risk}}
      \;=\;\underbrace{\|\eta_W-\bar\eta_W\|^2}_{\text{Wall B}}
      \;+\;\underbrace{\|\bar\eta_W-h_0\circ\phi\|^2}_{\text{correctable}}
      \qquad(\text{orthogonal; coupled only via }\mathrm{WallB}\le\mathbb E[\beta^2])$};
\end{tikzpicture}

%% file: figures/fig_borrowed.tex
\begin{tikzpicture}[>=stealth, x=1cm, y=3.4cm, line cap=round]
  \definecolor{figblue}{HTML}{1E5AA8}
  \draw[->, line width=0.5pt] (0,0) -- (0,1.08);
  \foreach \y in {0,0.5,1} {
    \draw[line width=0.5pt] (-0.07,\y) -- (0,\y);
    \node[left, font=\scriptsize] at (-0.09,\y) {\y};
  }
  \node[font=\footnotesize, rotate=90, anchor=south] at (-0.7,0.5) {$\eta$};
  \draw[black!25, dashed, line width=0.5pt] (0,0.5) -- (9.2,0.5);
  \node[right, font=\scriptsize, black!55] at (9.22,0.5) {$\eta_0=0.5$};
  \draw[figblue, line width=2.2pt] (1.6,0.2) -- (1.6,0.8);
  \draw[figblue, line width=0.8pt] (1.5,0.2) -- (1.7,0.2);
  \draw[figblue, line width=0.8pt] (1.5,0.8) -- (1.7,0.8);
  \node[left, font=\scriptsize, figblue] at (1.46,0.8) {0.8};
  \node[left, font=\scriptsize, figblue] at (1.46,0.2) {0.2};
  \draw[<->, line width=0.5pt, black!60] (1.95,0.2) -- (1.95,0.8);
  \node[right, font=\scriptsize, black!70, align=left] at (2.02,0.5)
    {$w_\beta(a)$\\ $=0.60$};
  \node[font=\footnotesize] at (1.6,0.07) {$a$:\ $\beta=0.3$};
  \fill (4.2,0.5) circle (1.8pt);
  \node[font=\footnotesize] at (4.2,0.07) {$b$:\ $\beta=0$ (pinned)};
  \draw[->, line width=0.9pt] (5.0,0.5) -- (6.5,0.5);
  \node[above, font=\scriptsize, align=center, yshift=2pt] at (5.75,0.5)
    {$\phi(a)=\phi(b)=z_0$};
  \node[below, font=\scriptsize, align=center, black!70, yshift=-2pt] at (5.75,0.5)
    {$\bar\eta_q(z_0)=\tfrac12\eta_q(a)+0.25$};
  \draw[figblue, line width=2.2pt] (7.4,0.35) -- (7.4,0.65);
  \draw[figblue, line width=0.8pt] (7.3,0.35) -- (7.5,0.35);
  \draw[figblue, line width=0.8pt] (7.3,0.65) -- (7.5,0.65);
  \node[left, font=\scriptsize, figblue] at (7.26,0.68) {0.65};
  \node[left, font=\scriptsize, figblue] at (7.26,0.32) {0.35};
  \draw[<->, line width=0.5pt, black!60] (7.75,0.35) -- (7.75,0.65);
  \node[right, font=\scriptsize, black!70, align=left] at (7.82,0.5)
    {$\mathrm{WallA}$\\ $=0.30$};
  \node[font=\footnotesize] at (7.4,0.07) {$z_0$};
\end{tikzpicture}

%% file: sec_beyond.tex
\section{Beyond the Boundary: Diagnostics and Decision Rules}
\label{sec:beyond}

This section develops three procedures from quantities that remain
computable under the boundary results.

\subsection{The Local-Correction Screening Rule}
\label{sec:certificate}

The local-correction screening rule is a per-region, per-window computation, and it asks for
no labels beyond those the corrector it gates already consumes.
\emph{Inputs requiring no labels of their own.} The frozen readout
$\eta_0$ comes with the model. The drift-budget profile $\beta$ is a
prior, by Theorem~\ref{thm:irred}, and not a measurement, unless it is
upgraded to the history-transferred estimate $\hat D_{\mathrm{cov}}$ of
Section~\ref{sec:dichotomy} under the remaining assumptions of
Section~\ref{sec:assumptions} that Theorem~\ref{thm:dichotomy} invokes
((A1) stationarity, (A4)
shape, (A5$'$) exogeneity, (A6) coverage, (A0) MCAR anchors, with the
uncovered mass $\hat\pi_{\mathrm{out}}$ reported alongside). The remaining
inputs are the embedding $\phi$ and the local geometry: neighborhood
density and freshness of the labeled buffer, and the deduplicated local
intrinsic dimension of Section~\ref{sec:dedup}.
\emph{The noise floor.} For classification it is the Bernoulli quantity
$\mathbb E[p(1-p)]$ over the frozen predictions, which needs no labels at
all. For regression it is the short-range decoupled semivariance over the
stream's residuals, read at the smallest available distances among pairs
at least $1{,}000$ steps apart and without extrapolation to zero distance
(Section~\ref{sec:twins} estimates the extrapolated nugget instead, and
reports it as a range). Our retrospective evaluation computes that floor over
the whole stream; a deployment would compute it over the revealed prefix.
From these, three quantities are computed per region.

\begin{enumerate}[leftmargin=*]
\item \emph{Decision-flip fraction}
$\Pr[\,|\eta_0-\tfrac12|\le\beta\,]$: the mass on which drift within
budget could change the decision (the identified-set diameter in the
decision discrepancy $\Delta_{01}$). Where it is near zero, even the full
assumed budget cannot move decisions, and freezing is safe in the decision
currency whatever the wall's height in the mean currency.
\item \emph{Price of freeze}
$\rho_\beta(x)=2\max\!\big(0,\beta(x)-|\eta_0(x)-\tfrac12|\big)$, the
worst-case regret of trusting the frozen model at $x$. Its region average
prices the freeze option, so adaptation is worth its risk only where that
price is material.
\item \emph{Deployable Wall~A}, via the fiber-average formula of
Theorem~\ref{thm:borrowed}:
$\mathrm{WallA}(\beta;R)=\mathbb E_{z\sim\nu_R}\!\big[\mathbb E_{p_W}[\,w_\beta\mid\phi=z\,]\big]$,
where $w_\beta=\min(1,\eta_0+\beta)-\max(0,\eta_0-\beta)$ is the post-clip
interval width and the inner average runs over the full fiber population.
For injective $\phi$ (generic at embedding dimension $192$) fibers
are singletons and this is the plain plug-in average of $w_\beta$ over the
region's unlabeled draws, accurate to $O_P(m^{-1/2})$ given the budget
profile. For a coarse or resolution-limited $\phi$ the inner average is
estimated by $k$-NN in $\phi$ over the \emph{full} window (the operational
form of the borrowing), at a rate graded partial in Appendix~F of Online
Appendix~1 (Corollary~F.5). The
mandatory order of operations is to average the post-clip width across the
fiber, never to average $\beta$ first and clip after. The post-clip width
is concave in $\beta$, so the latter overstates the wall.
\end{enumerate}

These quantities define the \emph{local-correction screening rule}. The rule
adapts where the local estimate is well determined against the calibrated floor and the flip
fraction is material, and freezes where the identified set is wide (the wall
output is then a restatement of the prior, and acting on it would be
acting on an assumption) or the buffer sparse or stale. When the budget is
transferred rather than assumed, its semantics must be reported with it.
The mean-response forecast is a calibrated forecast of typical drift that
upper-bounds the mean realized width, not a worst-case certificate.

The two gates that decide every verdict in the evaluation below (the
calibrated noise floor and the local-determinacy statistic) are empirical
instruments, not consequences of the Part~II theorems. What the theory
contributes to the rule is interpretation: Theorem~\ref{thm:irred} makes
the budget $\beta$ a prior, the three quantities above inherit that
status, and the theory thereby tells the operator which of the rule's
readings are measurements and which restate an assumption. On the
evaluated streams the prior-dependent quantities corroborate verdicts the
empirical gates have already forced (the zero flip fraction on
\textsc{homecredit-default} confirms a freeze the determinacy gate
decided on its own); they decide none.

\noindent\emph{The rule on all eight streams.}
We call this a screening rule rather than a certificate because it has no
finite-sample false-adapt guarantee. The evaluation below is retrospective;
Section~\ref{sec:conclusion} collects the corresponding soundness
limitations.
We apply the rule to all eight streams, not only to the harm case, so that
its false-freeze behavior is measured alongside its false-adapt
behavior. Table~\ref{tab:certificate} instantiates the two
label-free gates against the realized outcomes of
Table~\ref{tab:channels}. The first gate is room above the calibrated
floor ($1-\text{floor}$); the second is a local estimate determined
out-of-sample under the deployed strictly-past retrieval (past-only
buffer, $k=20$, unweighted neighbor mean). The rule is
fixed before the outcomes are consulted, and the eight verdicts are
unchanged over the whole box $\tau_{\mathrm{room}}\in[0,0.108)$,
$\tau_{\mathrm{sig}}\in[0,0.0616)$ of the two thresholds. The operating
point is $\tau_{\mathrm{room}}=0.10$ with $\tau_{\mathrm{sig}}=0$, that is,
a strict positivity test on the local statistic. \textsc{weather} is the
stream nearest a boundary of that box: its room is $0.108$, so a room
threshold of $0.2$ would freeze it.

In this retrospective evaluation, \emph{the rule makes no false adapt on any
stream}: every stream it clears improved under
adaptation ($+3.94\%$ and $+3.80\%$ RMSE, $+0.74$ AUC points), and the one
stream that is harmed is frozen. On \textsc{homecredit-default} the
freeze is decided by the local-determinacy gate alone
($R^2_{\mathrm{causal}}=-0.021$): the room gate passes, because a
Bernoulli floor of $0.85$ still leaves $0.15$ of nominal room. The
decision-currency quantities corroborate the same verdict independently.
The flip fraction is identically zero for every budget up to $\beta=0.25$
(the predictions concentrate around a mean of $0.019$ with no mass in
$[0.25,0.75]$), so within-budget drift cannot move a single decision, and
the price of freeze is zero. The calibrated floor changes how the stream
\emph{reads} rather than what is decided here: a variogram floor of the
kind the regression streams use reads $0.73$ on this stream against the
Bernoulli calibration's $0.85$, that is $0.27$ of nominal room against
$0.15$, which is the difference between a stream that looks worth
adapting and one that barely does. Even so, here the local-determinacy
gate refuses either way.

One of the three clearances warrants a caution the other two do not. On
\textsc{homesite-insurance} the clearing statistic is
$R^2_{\mathrm{causal}}=+0.070$ on the growing buffer, but the stream's
strict-split ceiling is $0.016$ and falls to $0.006$ under the
temporal-gap stress test of Section~\ref{sec:leak}, so the local
squared-error signal that clears the gate may be partly residual temporal
autocorrelation rather than durable spatial structure. The realized
$+0.74$pp gain traveled through the ranking channel C4, which the gate
does not observe: the verdict is right, but for reasons the rule cannot
claim credit for. A stress-tested variant of the determinacy gate
(scoring the local $R^2$ under a temporal gap) would demote this stream
to a marginal call and is the natural tightening; we flag rather than
adopt it, since changing the gate after seeing the outcomes would forfeit
the fixed-in-advance status of Table~\ref{tab:certificate}.

The price of that safety is conservatism, and it has two distinct causes
worth separating. Three regression streams are frozen although adaptation
helped them ($+0.83\%$, $+0.66\%$, $+0.59\%$ RMSE forgone). Two of the three
are attributed by Table~\ref{tab:channels} to the \emph{global} debias
channel C1 (on \textsc{delivery-eta}, $93\%$ of the probe corrector's gain is a running
mean), and the third, \textsc{maps-routing}, to a C1/partial-scale-C3
mixture; in all three the rule reports that the \emph{local} field
carries no out-of-sample signal, which for the first two is the full explanation and for the third is a partial one. A gate for the inexpensive
global channels is a different and easier instrument, which we do not
build here. The fourth false freeze is the one our own theory predicts:
\textsc{ecom-offers} gains $+0.86$ AUC points through ranking
recalibration (C4) while its squared-error local field is empty, and a
loss-denominated gate cannot see a ranking gain. This is the currency
limitation of Section~\ref{sec:gates} recurring inside the rule
itself, now with a price attached. For ranking deployments the
loss-denominated components are necessary but not sufficient, and a
practitioner who freezes on them forgoes gains of this size.

Across the eight retrospective streams, each of the four forgone gains
is of the same order as the single avoided harm. The rule's case rests on loss asymmetry rather than on
expected value; it is built for deployments where one undetected
degradation costs more than several missed sub-point gains.

\begin{table}[t]
\centering\small
\begin{tabular}{llrrlll}
\toprule
stream & task & floor & $R^2_{\mathrm{causal}}$ & verdict & realized & outcome \\
\midrule
\textsc{weather}            & reg & $0.89$ & $+0.052$ & adapt  & $+3.94\%$ & gain kept \\
\textsc{sberbank-housing}   & reg & $0.66$ & $+0.021$ & adapt  & $+3.80\%$ & gain kept \\
\textsc{homesite-insurance} & cls & $0.88$ & $+0.070$ & adapt  & $+0.74$pp & gain kept \\
\textsc{homecredit-default} & cls & $0.85$ & $-0.021$ & freeze & $\mathbf{-2.28}$\textbf{pp} & harm avoided \\
\textsc{cooking-time}       & reg & $0.92$ & $-0.009$ & freeze & $+0.83\%$ & forgone (C1) \\
\textsc{delivery-eta}       & reg & $0.97$ & $-0.018$ & freeze & $+0.66\%$ & forgone (C1, $93\%$ probe) \\
\textsc{maps-routing}       & reg & $0.88$ & $-0.006$ & freeze & $+0.59\%$ & forgone (C1/C3 mix) \\
\textsc{ecom-offers}        & cls & $0.87$ & $-0.022$ & freeze & $+0.86$pp & forgone (C4 ranking) \\
\bottomrule
\end{tabular}
\caption{The local-correction screening rule on all eight streams, with the thresholds
fixed before the outcomes were consulted, under the rule stated
above: adapt only
if $1-\text{floor}\ge\tau_{\mathrm{room}}$ and
$R^2_{\mathrm{causal}}>\tau_{\mathrm{sig}}(1-\text{floor})$, at the
operating point $\tau_{\mathrm{room}}=0.10$, $\tau_{\mathrm{sig}}=0$;
verdicts are constant over $\tau_{\mathrm{room}}\in[0,0.108)$,
$\tau_{\mathrm{sig}}\in[0,0.0616)$. Floors are the calibrated estimators
of Section~\ref{sec:certificate}, as fractions of residual variance;
$R^2_{\mathrm{causal}}$ is the local-determinacy statistic of
Appendix~\ref{app:J}, the strictly-past $k$-NN $R^2$ under the deployed
retrieval; it is the out-of-sample $R^2$ of the streaming corrector
itself (block $64$, first $100$ points reserved; Appendix~\ref{app:J}) and
not the growing-buffer oracle ceiling of Section~\ref{sec:leak}, which on
\textsc{weather} reads $0.067$ against the $+0.052$ printed here.
Gate inputs and the realized column are computed on the
single-prior probe streams (Appendix~\ref{app:J}). The rule speaks only to the local
channel C3.}
\label{tab:certificate}
\end{table}

\noindent\emph{Formal statements:} the decision-flip and price-of-freeze
functionals are the canonical discrepancy instances of the program
contract (Online Appendix~1, ``the formal program'');
the deployable estimand, its plug-in, and the clip/average
non-commutation are Theorem~F.3, Corollary~F.5, and
Remark~F.5, Appendix~F of Online Appendix~1; the budget-transfer
option, its error budget, and its forecast-not-certificate semantics are
Theorem~C.15 and Remark~C.16,
Appendix~C.

\subsection{The Degradation-Attribution Procedure}
\label{sec:attribution-procedure}

Given a labeled audit sample (periodic or delayed labels suffice), the
orthogonal split of Theorem~\ref{thm:wallB} operationalizes as a three-way
attribution of a window's excess error: the \emph{Wall-B share}
$\widehat{\mathbb E}[\operatorname{Var}(\eta_W\mid\phi)]$ (at a stated fiber
resolution), the \emph{correctable share} (the time-respecting-oracle ceiling of the
combined channels; the worked example below measures its local, C3 slice,
folding the temporal and ranking channels into the remainder), and the
\emph{Wall-A remainder}. Two estimates are required; the third share follows by subtraction.

\emph{Step 1 (Wall-B share, needs the audit labels).} Because an injective
embedding has singleton fibers and hence raw
$\mathrm{WallB}\equiv 0$, the deployable object is the resolution-$\rho$
deficit $\mathrm{WallB}_\rho$: coarsen $\phi$ to $\rho$-balls, estimate
the conditional mean of the audit labels at the finest resolution the
sample supports, and report the average variance of that mean within the
$\rho$-cells. The resolution $\rho$ must be stated
with the estimate ($\mathrm{WallB}_\rho$ is monotone in $\rho$ and
recovers $\mathrm{WallB}$ as $\rho\to 0$), so the share is meaningful
only relative to the fiber scale at which the readout is actually
operated; Appendix~\ref{app:J} gives the selection protocol that ties
$\rho$ to the deployed retrieval scale. A nonzero Wall-B share does not disable the history-transfer
machinery of Section~\ref{sec:dichotomy}. That theory applies verbatim to
the $\phi$-reachable drift, with the within-fiber discrepancy priced by
$\mathrm{WallB}^{1/2}$ and reported separately (the (A3) disclosure of
Appendix~C, Online Appendix~1).

\emph{Step 2 (Wall-A envelope, label-free).} The borrowed-certainty
plug-in of Section~\ref{sec:certificate}, computed from
$(\phi,\eta_0,\beta)$ alone, prices the Wall-A remainder \emph{ex ante}.
It is the diameter of the identified set under the assumed budget (an
upper envelope on what unidentified drift could contribute, not an
estimate of the realized share). The realized Wall-A share is then the
remainder of the audit window's excess error after subtracting the Wall-B
share and the correctable share, the latter already measured as the time-respecting
(leak-free) ceiling of the combined channels under the strict out-of-time
protocol of Section~\ref{sec:leak}.

Each share suggests a different intervention: representation changes for
Wall~B, engineering changes for the correctable share, and additional labels
for the Wall-A remainder. Table~\ref{tab:shares} works
the attribution on all eight streams, in the one currency common to all three
shares (squared error as a fraction of residual variance). The measured Wall-B share is indistinguishable
from zero on every stream: the fine-scale arm is strictly worse out of
sample everywhere, so no conditional-mean structure below the deployed
retrieval scale is detectable at these sample sizes, and representation
work is not the binding constraint. The main cross-stream differences lie in the split between the local correctable share and
the remainder: \textsc{weather} and \textsc{homesite-insurance} are the
only streams where the local channel has anything to recover ($0.025$
of $0.108$ and $0.016$ of $0.125$),
\textsc{sberbank-housing} has the largest room ($0.34$) with none of it
locally correctable (its realized gain traveled through freshness, C2,
which a spatial split cannot see), and on \textsc{homecredit-default}
the entire $0.15$ is remainder and naive correction harms. The remainder necessarily bundles the realized Wall-A share with the
temporal channels (C1, C2) and ranking-currency gains (C4) that a strict
spatial split cannot register. The measurements therefore do not support a
uniform remedy across streams.

\begin{table}[t]
\centering\small
\setlength{\tabcolsep}{5pt}%
\begin{tabular}{lrrrr}
\toprule
stream & room & correctable (local) & $\mathrm{WallB}_\rho$ & remainder \\
\midrule
\textsc{weather}            & $0.108$ & $0.025$ & $0.000$ & $0.084$ \\
\textsc{sberbank-housing}   & $0.337$ & $0.000$ & $0.000$ & $0.337$ \\
\textsc{cooking-time}       & $0.075$ & $0.000$ & $0.000$ & $0.075$ \\
\textsc{delivery-eta}       & $0.032$ & $0.000$ & $0.000$ & $0.032$ \\
\textsc{maps-routing}       & $0.117$ & $0.000$ & $0.000$ & $0.117$ \\
\textsc{ecom-offers}        & $0.134$ & $0.000$ & $0.000$ & $0.134$ \\
\textsc{homesite-insurance} & $0.125$ & $0.016$ & $0.000$ & $0.108$ \\
\textsc{homecredit-default} & $0.146$ & $0.000$ & $0.000$ & $0.146$ \\
\bottomrule
\end{tabular}
\caption{Worked three-way attribution on the eight probe streams, as
fractions of residual variance on the evaluation half. Room is $1$ minus
the calibrated floor (Table~\ref{tab:certificate}); the correctable
column is the strict out-of-time $k{=}20$ ceiling of
Table~\ref{tab:channels} (its local, C3 slice); $\mathrm{WallB}_\rho$
is the additional variance a finest-supported ($m{=}5$) readout explains
over the deployed scale, out of sample, at the protocol resolution
$\rho$ (median $20$th-neighbor distance, $1.9$--$4.0$ across streams;
Appendix~\ref{app:J}); the remainder bundles the realized Wall-A share
with the temporal and ranking channels a spatial split cannot see. Room
is measured on the full rule-evaluation window and treated as a common
fraction under approximate stationarity of the residual variance across
halves.
Rows may differ in the last digit from column arithmetic; shares are
rounded independently; the computation is part of the reproduction
deposit (Appendix~\ref{app:J}).}
\label{tab:shares}
\end{table}

\noindent\emph{Formal statements:} the Pythagorean split, the
$\mathrm{WallB}_\rho$ object, and the orthogonality qualifications are the
two-walls block of the program contract
(Online Appendix~1, ``the formal program''); the
$\phi$-reachable estimand and its Wall-B disclosure are assumption (A3) of
Appendix~C of Online Appendix~1; the deployable Wall-A envelope is
Corollary~F.5, Appendix~F.

\subsection{The Price of Labels}
\label{sec:pricelist}

When the Wall-A share dominates, additional labels are required. The
boundary theory gives two rates.

\emph{Region average, parametric.} $k$ anchors drawn i.i.d.\ from the
region pin the region-mean wall to
$\min\!\big(2\bar\beta_R,\ \tilde O(k^{-1/2})\big)$. Averaging a bounded
quantity over a fixed region is a parametric problem, so no geometry
enters the rate. The pure $k^{-1/2}$ collapse is operative only once
$k\gtrsim\bar\beta_R^{-2}$. Below that threshold the assumed budget still
binds, so the pin sits at the prior level $2\bar\beta_R$, not at the
anchor rate.

\emph{Pointwise, nonparametric.} Recovering the wall's height point by
point requires anchors \emph{plus} Lipschitz-in-$\phi$ smoothness, at the
local nonparametric rate $\tilde O(k^{-1/(2+d_0)})$. Each point can only
be estimated from anchors nearby in the embedding, and the intrinsic
dimension $d_0$ of the region's support in $\phi$-space governs how many
anchors are ever local. This is the same deduplicated $d_0$ that governs channel C3 and the margin
threshold. Higher intrinsic dimension both weakens unlabeled correction and
raises the label requirement.

Up to suppressed constants and logarithmic factors, the rates imply the
following costs at two points of the proxy map
(Table~\ref{tab:J-map}): past the $\bar\beta_R^{-2}$ threshold, halving
region-averaged uncertainty costs $4\times$ the labels, whatever the
geometry. At $d_0=3$ (the floor among the measured streams, where
\textsc{sberbank-housing} sits), halving \emph{pointwise} uncertainty
costs $2^{2+3}=32\times$. At $d_0=13$ (the
\textsc{homecredit-default} / \textsc{lending\_club} neighborhood, $d_0=13.8$
and $13.6$; the map's top, \textsc{maps-routing}, sits at $17.6$), it costs
$2^{2+13}\approx 3\times 10^4\times$. Equivalently, $N=10^4$ labels yield
region-averaged resolution of order $N^{-1/2}=10^{-2}$ on either stream,
but pointwise resolution of order $N^{-1/5}\approx 0.16$ at $d_0=3$ and
$N^{-1/15}\approx 0.54$ at $d_0=13$. On the high-dimensional streams a
pointwise drift map is out of reach at any realistic labeling budget,
while the regional audit is unaffected. Anchors yield the most where
the pointwise diameter is large and the region straddles the decision
boundary. Under the measured $d_0\approx 3$--$18$, pointwise certification is
costly, whereas region-averaged certification is comparatively inexpensive.
The labeling budget should therefore target regional wall audits and
decision-flip monitoring rather than pointwise drift maps.

\noindent\emph{Formal statements:} both rates are Branch~2(a) of the
decisive lemma in the program contract
(Online Appendix~1, ``the formal program''), with the
generic intrinsic dimension there instantiated by the deduplicated
$d_0$ of Table~\ref{tab:J-map} (Appendix~\ref{app:J}); the multi-window
refinement of the pointwise line (the anchor rate
$k^{-\alpha_s/(2\alpha_s+d_0)}$ and its margin-graded version) is given by
the margin-boundary results of Section~\ref{sec:dichotomy},
Appendix~\ref{app:H}, and Appendices~C, G, and~I of Online Appendix~1.

\subsection{Deployment Recommendations}
\label{sec:posture}

Together, the three procedures suggest the following deployment practice. Deploy the combined channels where an
out-of-time validation shows a positive gain. Screen the local channel
with the rule of Section~\ref{sec:certificate}: on the eight
retrospective streams it made no false adapt, at the measured price of
freezing four streams whose gains traveled wholly or partly through
channels it does not see; gates for the temporal and ranking channels are open instruments,
so in ranking deployments the deployment currency itself must be
monitored (Section~\ref{sec:gates}). Attribute residual degradation with
the three-way split, and allocate labels according to the price of labels when the
attribution puts the mass on the identifiability wall.

\section{Conclusion}
\label{sec:conclusion}

We studied two opposing claims about adaptation under drift: that it recovers
drift, and that unlabeled adaptation cannot help. The empirical results
support neither claim as a general rule. Across eight temporal streams, the
observed improvements and one harm case can be attributed to a small set of
mechanisms under a leakage-corrected protocol. Under the stated agnostic
drift class, the theory then separates uncertainty caused by
non-identifiability from error caused by the frozen representation. The
margin-indexed analysis proves the regular branch and states the remaining
condition in the lower-bound program explicitly. TabReD proxies and a
sensitivity analysis on fifteen additional benchmarks place the examined
settings on the difficult side under their stated assumptions, but do not
establish a universal empirical phase law. These distinctions determine
whether a deployment should improve the correction method, change the
representation, or acquire labels.

Several limitations bound these claims. The equality case
$\gamma=\gamma^\star$ is not covered by the current log-rate argument.
Above the threshold, the efficient constant under overlapping windows is
graded partial, and the rate statement itself holds only under the
stated nuisance-rate hypotheses, (H6) included, whose derivation from
primitives is graded partial. Below the threshold, the lower bound at $\gamma=0$ is
conditional on one graded affinity estimate (gap (2c)), for fixed and
growing $K$ alike; the margin adaptation for
$0<\gamma<\gamma^\star$ remains open (Remark~H.8); Online Appendix~1
records the status of every statement. On the empirical side, the proxy map estimates drift
smoothness and margin from
single streams, so the point estimates are noisy even where the verdict,
resting only on an inequality under the rough-field bound, is insensitive
to that noise. All measurements use a single
backbone family, so the map's $d_0$ floor should be re-established per
representation. Schema and semantic drift, where the covariate space itself
changes, lie outside our formalism altogether, and are in our view the
deepest open formal gap in the area.

The screening rule carries qualifications of its own. Its soundness can be
checked only with labels the deployment lacks, so it must be validated on
semi-synthetic controls with matched covariate laws; Section~\ref{sec:phasesim}
gives the generator we would start from, and a community-scale study is
future work. One-sidedness on all eight
streams in retrospective evaluation (Table~\ref{tab:certificate}) is an
observation on one deployment, not a bound: soundness on unseen streams is
not established. Finally, the loss-denominated gates forgo ranking gains of
measured size, so a sound ranking-currency gate with the same one-sided
profile remains open, as does a ranking-currency analogue of the wall
theory itself.

%% file: appendix/app_guide.tex
\section*{Reader's guide to the appendices}
\phantomsection\label{app:guide}
\addcontentsline{toc}{section}{Reader's guide to the appendices}

\paragraph{Where each appendix appears.} This document keeps the appendices a
reader needs to check the headline results: \textbf{B} (the label-free
irreducibility core behind Theorem~\ref{thm:irred}), \textbf{H} (the fixed-$K$
margin boundary behind Theorem~\ref{thm:dichotomy}), the auxiliary
Lemmas~C.11$'$ and G.14$'$ that Theorem~\ref{thm:dichotomy} quotes, and
\textbf{J} (the complete measurement protocol). Appendices \textbf{A, C,
D--F, G, I, K} and the program contract appear only in Online Appendix~1,
the supplementary document that accompanies this arXiv submission as the
ancillary file \texttt{anc/supplementary.pdf}. Its first half reproduces
all of Appendices A--K (B, H, and J included), the program contract, and
the two auxiliary lemmas, with their original letters and statement
numbers unchanged, so a statement cited by a letter absent from this
document (for example Theorem~G.1.4 or Corollary~F.5) appears there under
the same number. The lettering in this document therefore runs B, H, J by
design.

The appendices can be read at three depths. \textbf{(i) Contract level:}
read the statement-status conventions (Online Appendix~1, Appendices
A--K), the status note closing Appendix~\ref{app:H}, and the per-section
status tables of Online Appendix~1, \S S1; this suffices to verify the
grade (\emph{proven}, \emph{partial}, \emph{conditional},
\emph{cited-standard}, or \emph{conjecture}) of every claim cited from the main text.
\textbf{(ii) Pillar level:} additionally read the six pillar proofs,
namely Thm.~A.14 (exact reduction identity), Thm.~B.8 (label-free
irreducibility), Thm.~F.2--Cor.~F.5 (deployable Wall~A reduction),
Thm.~G.1.4 ($\sqrt N$ CLT), Thm.~H.9 (margin-indexed regular branch;
graded partial, see the status note of Appendix~\ref{app:H}), and
Thm.~I.3.2 (the strict branches above the threshold and the growing-$K$
remainder).
\textbf{(iii) Full:} everything, including Appendix~K of Online
Appendix~1 (exact admissibility and the permutation-mixture affinity
bound) and the auxiliary-lemmas section here (Lemmas~C.11$'$ and
G.14$'$); at this depth most of the reading is in Online Appendix~1: its
first half carries the verified statements of all of A--K (proof outlines
and consumed outputs kept, rendered numbering unchanged), and the full
developments live in \S S7 (the selection and measurability machinery
Lemmas~A.1--A.13 of Appendix~A, with the auxiliary OP1 remark C.4 of
Appendix~C), \S S5 (the A5$'$ orthogonalization program of Appendix~D),
\S S6 (the C.20 program of Appendix~E), and \S S3 (the complete staged
proof of Lemma~B.5$'$, Appendix~K's permutation-mixture affinity bound,
numbered Lemma~K.8 there and unrelated to Lemma~B.5 of Appendix~B,
together with the proofs of K.1, K.3, and K.4). Proofs that were promoted
into the main text (Section~\ref{sec:wallA}) are replaced in Appendices~B,
F, and the contract by one-line pointers; the statements and status tags
are unchanged.

Online Appendix~1 also records, per appendix, the status of every
statement, the ranked open points, and a status summary (\S S1); the
comparison with the published theorems behind Lemmas~C.11$'$ and G.14$'$
(\S S2); the complete staged proof of Lemma~B.5$'$ with its machine-checked
cases (\S S3); a one-page index of the verification programs (\S S4); and
the full developments of the abridged Appendices~D and~E (\S S5, \S S6). The
verification code and its outputs are part of the reproduction deposit
(Section~\ref{sec:setup}); they are not part of this arXiv source package.

%% file: appendix/app_B.tex
\section{Ancillarity, the Le Cam two-point wall, and its (B)-contingent scope}
\label{app:B}

\subsection*{B.0 Standing assumptions}

\textbf{(S1)} $(\mathcal X,\mathcal B_{\mathcal X})$ Polish with Borel $\sigma$-algebra; $\mathcal Y$ finite with $2^{\mathcal Y}$; primary case $\{0,1\}$. \textbf{(S2)} $p_0(\cdot\mid x)=h_0(\phi(x))$ a Markov kernel ($\phi,h_0$ Borel); $\eta_0(x)=p_0(1\mid x)$. \textbf{(S3)} $\beta:\mathcal X\to[0,1]$ Borel (used by the extremal pair of Lemma B.6). \textbf{(S4)} (O1) i.i.d. $x_1,\dots,x_n\sim p_W$; observation map $\Pi_A:(x,y)\mapsto x$; $y$ never observed; the observation law of a world $P$ is $Q_P^{(n)}:=(\Pi_A^{(n)})_\#P^{\otimes n}$. \textbf{(S5)} (B) $I(\beta)$ as in Appendix~A of \OAname{} (and Section~\ref{sec:idset} of the main paper). \textbf{(S6)} $R\in\mathcal B_{\mathcal X}$, $p_W(R)>0$; $\mu:=p_W(\cdot\mid R)$. \textbf{(S7)} An estimator is a Borel map $\hat\theta_n:\mathcal X^n\to\mathbb R$ (for the infinite-sequence claims: on $\mathcal X^{\mathbb N}$); a randomized estimator is Borel on $\mathcal X^n\times[0,1]$ (resp. $\mathcal X^{\mathbb N}\times[0,1]$) with independent $U\sim\mathrm{Unif}[0,1]$; general randomized procedures given by Markov kernels are covered by the kernel-rule case of Lemma B.5. Binary TV identity: $D_{\mathrm{TV}}(a,b)=|\eta_a-\eta_b|$; $B_\beta(x)=[\underline\eta,\overline\eta]$ with $\underline\eta=\max(0,\eta_0-\beta)$, $\overline\eta=\min(1,\eta_0+\beta)$, width $w_\beta=\overline\eta-\underline\eta\in[\beta,2\beta]$.

\subsection*{B.1 Construction and ancillarity}

\begin{astmt}{Lemma B.1 (world construction). [proven]}\phantomsection\label{a:B1}
For $p_W\in\mathcal P(\mathcal X)$ and any Markov kernel $q$, $P=p_W\otimes q$ defined by $P(E)=\sum_y\int_{E_y}q(y\mid x)p_W(dx)$ is a probability measure on $(\mathcal X\times\mathcal Y,\mathcal B_{\mathcal X}\otimes2^{\mathcal Y})$ with $X$-marginal $p_W$ and $p_W$-a.e.-unique disintegration $q$.
\end{astmt}

\begin{proof}
Since $\mathcal Y$ is finite, every product-measurable $E$ decomposes as $\bigcup_yE_y\times\{y\}$ with Borel sections; countable additivity holds per $y$ by monotone convergence, and the total mass is $1$. For uniqueness, note that for each $y$ both versions are Radon--Nikodym derivatives of $A\mapsto P(A\times\{y\})\le p_W(A)$ with respect to $p_W$.
\end{proof}

\begin{astmt}{Lemma B.2 (product pushforward). [proven]}\phantomsection\label{a:B2}
$(\Pi_A^{(n)})_\#P^{\otimes n}=((\Pi_A)_\#P)^{\otimes n}$.
\end{astmt}

\begin{proof}
The two sides agree on measurable rectangles, which form a generating $\pi$-system, and both are probability measures, so they coincide by Dynkin's $\pi$--$\lambda$ theorem.
\end{proof}

\begin{astmt}{Proposition B.3 (ancillarity, $n\le\infty$). [proven]}\phantomsection\label{a:B3}
For every Markov kernel $q$ and every $n$: $Q^{(n)}_{p_W\otimes q}=p_W^{\otimes n}$; hence $D_{\mathrm{TV}}(Q_P^{(n)},Q_{P'}^{(n)})=0$ for all $q,q'\in I(\beta)$. The same holds for the infinite-sequence law on $(\mathcal X^{\mathbb N},\mathcal B_{\mathcal X}^{\otimes\mathbb N})$.
\end{astmt}

\begin{proof}
The identity $(\Pi_A)_\#P=p_W$ uses only $q(\mathcal Y\mid x)=1$, and Lemma B.2 gives the finite case. For $n=\infty$, the i.i.d. infinite product $P^{\otimes\mathbb N}$ exists on an arbitrary probability space by the infinite-product-measure theorem (equivalently Ionescu--Tulcea with constant kernels; Kolmogorov extension is also available under (S1)); the coordinatewise map $\Pi_A^{(\mathbb N)}$ is product-measurable; cylinder sets form a generating $\pi$-system on which the two observation laws agree by the finite case, and Dynkin's theorem concludes.
\end{proof}

\begin{astmtrm}{Remark B.4 (terminology).}\phantomsection\label{a:B4}
With $p_W$ fixed and $q\in I(\beta)$ the parameter, the full data $x_{1:n}$ is \textbf{ancillary} for $q$ (exact, uniform in $n$): marginalizing the never-observed $y$ erases $q$ from the model before any data are drawn.
\end{astmtrm}

\subsection*{B.2 Le Cam two-point and the irreducibility theorem}

\textbf{The functional.} $\theta(P):=\mathbb E_{\mu}[\eta_P]$, $\eta_P:=dP(\cdot\times\{1\})/dp_W$ (a Radon--Nikodym derivative, which exists because $P(A\times\{1\})\le p_W(A)$ with both measures finite, and satisfies $0\le\eta_P\le1$ a.e.). By Lemma B.1's uniqueness and $\mu\ll p_W$, $\theta(p_W\otimes q)=\mathbb E_\mu[\eta_q]$ is a well-defined functional of the world.

\begin{astmt}{Lemma B.5 (two-point inequality; randomized and kernel rules). [proven]}\phantomsection\label{a:B5}
For probability measures $Q_0,Q_1$ on $(\Omega,\mathcal F)$ and $\theta_0,\theta_1\in\mathbb R$: every randomized estimator satisfies $\max_iE|\hat\theta-\theta_i|\ge\tfrac12|\theta_0-\theta_1|(1-D_{\mathrm{TV}}(Q_0,Q_1))$.
\end{astmt}

\begin{proof}
Let $\nu=Q_0+Q_1$, $g_i=dQ_i/d\nu$, and $d\mu_\wedge=\min(g_0,g_1)d\nu$, so that $\mu_\wedge(\Omega)=1-D_{\mathrm{TV}}(Q_0,Q_1)$. Pointwise, $|\hat\theta-\theta_0|+|\hat\theta-\theta_1|\ge|\theta_0-\theta_1|$, and $\mu_\wedge\le Q_i$; integrating against $\mu_\wedge$, summing, and halving gives the bound. For randomization, replace $Q_i$ by $Q_i\otimes\lambda$, which has the same TV distance. \textbf{Kernel rules:} for a Markov kernel rule $K(\omega,dt)$, $\int(|t-\theta_0|+|t-\theta_1|)K(\omega,dt)\ge|\theta_0-\theta_1|$ pointwise; integrate against $\mu_\wedge$ to obtain the same bound. So ``no label-free procedure improves'' covers all randomized procedures, not only $\mathrm{Unif}[0,1]$-randomized maps.
\end{proof}

\begin{astmt}{Lemma B.6 (extremal pair). [proven]}\phantomsection\label{a:B6}
$\eta_\pm=\Pi_{[0,1]}(\eta_0\pm\beta)$ define Borel Markov kernels $q_\pm$ with $q_\pm(\cdot\mid x)\in B_\beta(x)$ for \textbf{every} $x$ (binary TV identity: $|\eta_+-\eta_0|=\min(1-\eta_0,\beta)\le\beta$, symmetrically), so $q_\pm\in I(\beta)$; and $\theta(P_+)-\theta(P_-)=\mathbb E_\mu[w_\beta]$. $\square$
\end{astmt}

\begin{astmt}{Lemma B.7 (binary mean diameter: direct). [proven]}\phantomsection\label{a:B7}
$D_{\mathrm{mean}}(\beta;R)=\mathbb E_\mu[w_\beta]$, attained by $(q_-,q_+)$.
\end{astmt}

\begin{proof}
For the upper bound, $\eta_q,\eta_{q'}\in[\underline\eta,\overline\eta]$ $p_W$-a.e., hence $\mu$-a.e., and integrating gives the bound; attainment is Lemma B.6.
\end{proof} \emph{(This instantiates Theorem A.14 with $\delta^{\mathrm{mean}}_\beta=w_\beta$; the argument does not use Theorem A.14.)}

\begin{astmt}{Theorem B.8 (label-free irreducibility within (B)). [proven]}\phantomsection\label{a:B8}
Assume (S1)--(S7), binary $\mathcal Y$. For every $n$ (including $n=\infty$), over $\mathcal W=\{p_W\otimes q:q\in I(\beta)\}$:
\[
\inf_{\hat\theta_n\ \mathrm{randomized}}\ \sup_{P\in\mathcal W}\ E\big|\hat\theta_n-\theta(P)\big|\ =\ \tfrac12\,D_{\mathrm{mean}}(\beta;R)\ =\ \tfrac12\,E_{p_W(\cdot\mid R)}[w_\beta],
\]
constant in $n$, attained at the two-point subfamily $\{P_-,P_+\}$ (lower bound) and by the zero-data midpoint estimator $\hat\theta^\star=\mathbb E_\mu[(\underline\eta+\overline\eta)/2]$ (upper bound). If $p_W(R\cap\{\beta>0\})>0$ the bound is strictly positive ($w_\beta\ge\beta$, valid also at the both-clips boundary via $\beta\le1$). The wall height is a functional of the prior $(\beta,p_0,p_W)$ alone.
\end{astmt}

\begin{proof}
The proof has been promoted to the main text: the proof of Theorem~\ref{thm:irred} (Section~\ref{sec:irred}) transcribes it in full, and the inputs it consumes (Lemmas B.1--B.7 above) are unchanged.
\end{proof}

\begin{astmt}{Corollary B.9 (power $=$ size). [proven]}\phantomsection\label{a:B9}
Any test $\psi_n:\mathcal X^n\to[0,1]$ of $q$ vs $q'$ satisfies $E_{P'}[\psi_n]=\mathbb E_P[\psi_n]=E_{p_W^{\otimes n}}[\psi_n]$. $\square$
\end{astmt}

\begin{astmt}{Corollary B.10 (confidence sets cannot shrink; outer-expectation form). [proven]}\phantomsection\label{a:B10}
Let $C_n:\mathcal X^n\to2^{I(\beta)}$ satisfy the measurability convention ($\{x_{1:n}:q'\in C_n\}$ measurable for each fixed $q'$) and uniform coverage $P_q^{(n)}(q\in C_n)\ge1-\alpha$ for all $q\in I(\beta)$. Then, with $\operatorname{diam}_\theta C=\sup_{q,q'\in C}|\theta(P_q)-\theta(P_{q'})|$, the pointwise minorant bound
\[
\operatorname{diam}_\theta C_n\ \ge\ D_{\mathrm{mean}}(\beta;R)\cdot\mathbf 1_{\{q_+\in C_n\}\cap\{q_-\in C_n\}}
\]
holds, the minorant is measurable by the convention, and taking (inner, hence also outer) expectations under $\mu_n:=p_W^{\otimes n}$ with the union bound $\mu_n(q_+\in C_n,\,q_-\in C_n)\ge1-2\alpha$ (each marginal event has $Q$-probability $\ge1-\alpha$ by Prop. B.3, since all observation laws coincide) gives $E_*[\operatorname{diam}_\theta C_n]\ge(1-2\alpha)D_{\mathrm{mean}}(\beta;R)$ for every $n$. As a supremum over an uncountable data-dependent family, $\operatorname{diam}_\theta C_n$ itself need not be measurable, which is why the statement takes inner/outer-expectation form. $\square$
\end{astmt}

\begin{astmtrm}{Remark B.11 (uniformity; oracle $p_W$; other functionals).}\phantomsection\label{a:B11}
The constant is the same for every $n$, including $n=\infty$, and is unaffected by an oracle revealing $p_W$ itself: within (B) the observation functional $P\mapsto Q_P^{(n)}$ factors through $p_W$ alone. For any real functional $\theta$ on $\mathcal W$, the same proof gives minimax risk $\ge\tfrac12\sup_{q,q'}|\theta(P_q)-\theta(P_{q'})|$; the extremal pairs for $\Delta_{01}$ and $\rho_\beta$ are deferred (\OAname{}, \S S1, open point 6).
\end{astmtrm}

\subsection*{B.3 Scope delimitation: the scope theorems}

\begin{astmtrm}{Definition B.12 (class-restricted identified set). [proven-well-posed as scoped]}\phantomsection\label{a:B12}
For a set $\mathcal C$ of joint laws, $I_{\mathcal C}(p_W):=\{q:\ q\text{ is a disintegration of some }P\in\mathcal C\text{ with }(\Pi_A)_\#P=p_W\}$ (modulo $p_W$-a.e.\ equality), and $D^{\mathcal C}_\Delta(R):=\sup_{q,q'\in I_{\mathcal C}(p_W)}\mathbb E_\mu[\Delta(f_q,f_{q'})]$ with $\sup\emptyset:=0$. \emph{Scope note:} for non-singleton classes and general $\Delta$ this presupposes the measurability of the integrand, so carry the kernel hypotheses (S5) of Appendix~A of \OAname{} (not the mechanism prior (S5) above); in this section only singleton/empty cases are used, where the integrand is $\equiv0$ under (D0), which we assume of any discrepancy here.
\end{astmtrm}

\begin{astmt}{Proposition B.13 (covariate-shift collapse: prior-driven). [proven]}\phantomsection\label{a:B13}
$\mathcal C_{\mathrm{cov}}=\{p\otimes p_0\}$: for every observed $p_W$, $I_{\mathcal C_{\mathrm{cov}}}(p_W)=\{p_0(\cdot\mid\cdot)\}$ and $D^{\mathcal C_{\mathrm{cov}}}_\Delta(R)=0$ at $n=0$, for every (D0)-discrepancy, every $R$, every $\beta$.
\end{astmt}

\begin{proof}
The proof has been promoted to the main text (proof of Theorem~\ref{thm:covshift}(i), Section~\ref{sec:irred}); it rests on the uniqueness clause of Lemma B.1.
\end{proof}

\begin{astmt}{Lemma B.14 (moment identification of $\pi_W$). [proven]}\phantomsection\label{a:B14}
Fix known $\mu_y\in\mathcal P(\mathcal X)$, $y\in\mathcal Y=\{1,\dots,|\mathcal Y|\}$. The following are equivalent: (1) $\pi\mapsto\sum_y\pi(y)\mu_y$ is injective on $\Delta(\mathcal Y)$; (2) $\{\mu_y\}$ are linearly independent in $M(\mathcal X)$; (3) there is a bounded Borel $T:\mathcal X\to\mathbb R^{|\mathcal Y|}$ with nonsingular moment matrix $M_{jy}=\mathbb E_{\mu_y}[T_j]$ and $\mathbb E_{p_W}[T]=M\pi$ for every $p_W=\sum_y\pi(y)\mu_y$. In particular $\pi_W=M^{-1}\mathbb E_{p_W}[T]$.
\end{astmt}

\begin{proof}
The proof has been promoted to the main text and is reproduced inside the proof of Theorem~\ref{thm:covshift}(ii) (Section~\ref{sec:irred}). That proof assembles the total-mass observation (affine $\iff$ linear independence for probability measures), the two-priors converse, the Hahn--Jordan/annihilator selection of the indicator moments $T_j=\mathbf 1_{A_j}$ with nonsingular $M$, and moment linearity; (3)$\Rightarrow$(1) follows from nonsingularity.
\end{proof}

\begin{astmt}{Lemma B.15 (Bayes pinning; nonnegative RN versions). [proven]}\phantomsection\label{a:B15}
$\lambda=\sum_y\mu_y$, $g_y=d\mu_y/d\lambda$ \textbf{chosen nonnegative everywhere} (replace any version by $\max(g_y,0)$, another version, and one needed so that $q_\pi(\cdot\mid x)\in\Delta(\mathcal Y)$ at every $x$, as Lemma B.1's kernel framework requires); $m_\pi=\sum_y\pi(y)g_y$; $q_\pi(y\mid x)=\pi(y)g_y(x)/m_\pi(x)$ on $\{m_\pi>0\}$, $:=\pi$ elsewhere. Then $q_\pi$ is a Markov kernel and the $p_\pi$-a.e.-unique disintegration of $P_\pi(A\times\{y\})=\pi(y)\mu_y(A)$.
\end{astmt}

\begin{proof}
One has $p_\pi(\{m_\pi=0\})=0$; that $q_\pi$ is a Markov kernel and a disintegration of $P_\pi$ is a direct verification, and uniqueness follows from Lemma B.1.
\end{proof}

\begin{astmt}{Theorem B.16 (label-shift collapse: data-pinned). [proven]}\phantomsection\label{a:B16}
Under Lemma B.14's condition (2) and $p_W\in\mathcal M_{\mathrm{ls}}$, the set of $X$-marginals of the label-shift family $\mathcal C_{\mathrm{ls}}:=\{P_\pi:\pi\in\Delta(\mathcal Y)\}$ of Lemma B.15 (else the class is falsified and $D^{\mathcal C}:=0$ vacuously): $\pi_W$ is unique, recovered by $\pi_W=M^{-1}\mathbb E_{p_W}[T]$, $I_{\mathcal C_{\mathrm{ls}}}(p_W)=\{q_{\pi_W}\}$, and $D^{\mathcal C_{\mathrm{ls}}}_\Delta(R)=0$ for all $R$ and (D0)-$\Delta$. Binary: $q_W(1\mid x)=\pi_W(1)g_1(x)/(\pi_W(1)g_1(x)+\pi_W(0)g_0(x))$. $\square$
\end{astmt}

\begin{astmt}{Proposition B.17 ($\sqrt n$-estimability of $\pi_W$). [proven]}\phantomsection\label{a:B17}
$\hat\pi=M^{-1}\hat b$, $\hat b_j=n^{-1}\sum_iT_j(x_i)$: $\mathbb E\|\hat\pi-\pi_W\|_2\le\|M^{-1}\|_{\mathrm{op}}\sqrt{|\mathcal Y|/(4n)}$ (the variance is $\le1/(4n)$ per coordinate since each $T_j$ is $[0,1]$-valued; Jensen's inequality and the operator-norm bound then give the result). $\square$
\end{astmt}

\noindent\emph{Rate transfer to the $\theta$-functional requires a lower bound on the mixture density $m_{\pi_W}$ on the relevant region. That lower bound is \textbf{[conjecture]}, with the failure mode (boundary $\pi_W$, barely-overlapping class-conditional supports) explicitly delimiting it.}

\begin{astmt}{Theorem B.18 ((B)-contingency of irreducibility; scoped maximality). [proven]}\phantomsection\label{a:B18}
Fix the observation model (S4). Over the identical channel:
\begin{enumerate}
\item (B): wall $D_{\mathrm{mean}}(\beta;R)=\mathbb E_\mu[w_\beta]>0$ whenever $p_W(R\cap\{\beta>0\})>0$, invariant to $n$ (Theorem B.8);
\item $\mathcal C_{\mathrm{cov}}$: $D^{\mathcal C}=0$ identically, at $n=0$ (Prop. B.13);
\item $\mathcal C_{\mathrm{ls}}$ (known, linearly independent class-conditionals): $D^{\mathcal C}=0$ at population level, $\pi_W$ recovered at rate $n^{-1/2}$ (Thm. B.16, Prop. B.17);
\item (maximality, \textbf{scoped}): any mechanism class $\mathcal C$ over the same channel whose compatible conditionals embed in the budget, $I_{\mathcal C}(p_W)\subseteq I(\beta)$, satisfies $D^{\mathcal C}_\Delta(R)\le D_\Delta(\beta;R)$ (monotonicity of the pairwise supremum). Thus (B) is maximal \textbf{among budget-$\beta$-respecting mechanism priors}. Note $\mathcal C_{\mathrm{ls}}$ does not embed in $I(\beta)$ in general. Its collapse is by the direct computation (3), not by (4).
\end{enumerate}
Hence label-free irreducibility is a property of the mechanism prior (B), not of label-free observation per se. $\square$
\end{astmt}

\begin{astmtrm}{Remark B.19 (the unrestricted ``maximal wall'' reading does not hold: heuristic only).}\phantomsection\label{a:B19}
The budget-free class $\mathcal C_{\mathrm{all}}=\{p_W\otimes q:q\text{ any kernel}\}$ is a \emph{less} committal prior over the same channel; by Lemma B.1 every kernel is compatible with every marginal, so $I_{\mathcal C_{\mathrm{all}}}(p_W)$ is all kernels and $D^{\mathcal C_{\mathrm{all}}}_{\mathrm{mean}}(R)=1$ (take $q\equiv\delta_1$, $q'\equiv\delta_0$), exceeding $\mathbb E_\mu[w_\beta]$ whenever the latter is $<1$. So ``any defensible coupling prior can only lower the wall'' holds only under clause (4)'s embedding restriction; the ``upper envelope'' reading of Section~\ref{sec:irred} should be read as scoped to budget-respecting priors, and a precise envelope formulation over non-embedding priors (re-centered or enlarged balls) remains \textbf{unformulated} (\OAname{}, \S S1, open point 8).
\end{astmtrm}

\begin{astmt}{Proposition B.20 (the mechanism class is itself unfalsifiable from unlabeled data). [proven]}\phantomsection\label{a:B20}
If $p_W\in\mathcal M_{\mathrm{ls}}$, the worlds $p_W\otimes q$ ($q\in I(\beta)$, e.g. $q=p_0$), $p_W\otimes p_0\in\mathcal C_{\mathrm{cov}}$, and $P_{\pi_W}\in\mathcal C_{\mathrm{ls}}$ all induce $p_W^{\otimes n}$ for every $n$; every label-free test between mechanism classes has power $=$ size. The choice among (B), covariate shift, label shift is a prior, not a measurable fact of the window. $\square$
\end{astmt}

\begin{astmtrm}{Remark B.21 (two flavors of collapse; the currency).}\phantomsection\label{a:B21}
$\mathcal C_{\mathrm{cov}}$ collapses \textbf{prior-driven} ($n=0$); $\mathcal C_{\mathrm{ls}}$ collapses \textbf{data-pinned} (the class reduces the unknown to $\pi\in\Delta(\mathcal Y)$, identified and $n^{-1/2}$-estimated from unlabeled data). Under (B) neither route exists: the observation law is constant in the unknown. The currency of class (3) of Theorem B.18 is ``known $\{\mu_y\}$'', label-derived side information obtained outside the window: exactly the currency to which the program's deliverable is directed.
\end{astmtrm}

%% file: appendix/app_H.tex
\section{The magnitude-wall margin boundary (Conjecture E.6 of \OAname{}, regime (II))}
\label{app:H}

\noindent\textbf{Inherited results.} Section~G resolved the covered-part covariate wall
$\Dcov=2\int_\J\Psi\,d\nu$, $\Psi(r)=\E[b\mid R{=}r]$, at the parametric rate in the
\emph{direct/sign-coherent} regime~(I), and left \emph{one} decisive obstruction: the
\emph{raw magnitude} wall $b=|\bar\eta_J-\eta_0|$ of regime~(II), where the non-smooth
$|\cdot|$ Stage-1 map makes the outcome-bias hypothesis (H5) fail ``for all $K\ge1$''
(Rmk.~G.1.5 of Appendix~G, \OAname{}). Section~H establishes the regular branch and
records the conditional/open status of the lower-bound program. It
assembles two derivations (Route~A constructive, Route~B lower-bound)
and their synthesis, and records two qualifications: (i) no $\sigma^2/b$ heavy tail, and hence
no $\varepsilon_1^2$ bias-cap / ``$\as>d/2$'' clause, arises in the synthesis for Conjecture~E.6
(remark following Lem.~\ref{lem:H1}), and (ii)
a general-$K$ \emph{Assouad} lower bound at the exact rate $(N/K)^{-a}$ does not apply
(the functional is a scalar projection $\sum_j\tau_j$; the fuzzy
Le~Cam bound of Thm.~\ref{thm:H7} loses a $\sqrt K$ factor). All notation is inherited from Sections~C and~G.

\medskip\noindent\textbf{Inherited notation and standing conditions.}
$g(x):=\bar\eta_J(x)-\eta_0(x)$, $b:=|g|$, $s:=\sgn g$, plug-in sign
$\hat s:=\sgn(\hat\eta_J-\eta_0)$ from the per-window within-split Stage-1 fit $\hat\eta_J$;
$\E[y-\eta_0\mid\mathcal F_\phi]=g$ (observation condition O2 of Appendix~C, \OAname{}, plus Stage-1 in $\phi$-coordinates); representer
$w=d\nu/d\barmu\le\Gcov$ (A6); efficient influence function (EIF) $\phi=2[w(b-\Psi)+(\Psi-\theta)]$ (Prop.~G.1.1). Set
\[
a:=\frac{\as}{2\as+d}\in\Big(0,\tfrac12\Big)\ \ (\text{since }d\ge1),\quad
\varepsilon_1\asymp\Big(\frac{\log k}{k}\Big)^{a}\ \ (\text{Lemma~\ref{lem:C11prime}}),\quad k=N/K,\quad N=Kk.
\]
Standing: \textbf{(T)}~$\Pbar(g=0)=0$ (crossing set is $\barmu$-null); \textbf{(Mgn$_\gamma$)}~
$\Pbar(0<b\le t)\le Ct^{\gamma}$. Throughout $K$ is fixed and $k\to\infty$ unless a growing-$K$
regime is named. Define the \emph{margin threshold}
\[
\boxed{\ \gamma^\star:=\frac1{2a}-1=\frac{d}{2\as}\ }\qquad\Longleftrightarrow\qquad a(1+\gamma^\star)=\tfrac12.
\]

\bigskip\hrule\bigskip
\subsection{\texorpdfstring{Route A: $\sqrt N$ above the margin threshold}{Route A: above the margin threshold}}

The construction estimates $b$ by \emph{recovering its sign} and reading the \emph{raw}
Bernoulli residual, so no Jensen/smoothing bias is ever incurred; the only error is
misclassifying the sign, which the margin makes second-order.

\setcounter{theorem}{0}%
\begin{lemma}[weighted sign error; the exact outcome-bias identity]\label{lem:H1}\textbf{[proven]}
On the Stage-1 sup-norm event $E_1=\{\|\hat\eta_J-\bar\eta_J\|_\infty\le\varepsilon_1\}$
(prob.\ $\ge1-\sum_jk_j^{-2}$, Lemma~\ref{lem:C11prime}) the \emph{straddle} holds pointwise,
\[
\hat s(x)\ne s(x)\ \Longrightarrow\ 0<b(x)\le\varepsilon_1,\tag{$\star$}
\]
because $\hat g-g=\hat\eta_J-\bar\eta_J$ gives $|\hat g-g|\le\varepsilon_1$, so a sign flip
forces $0<b=|g|\le|g-\hat g|\le\varepsilon_1$. Hence, under (Mgn$_\gamma$),
\[
\Pbar(\hat s\ne s)\le\Pbar(0<b\le\varepsilon_1)\le C\varepsilon_1^{\gamma},
\qquad
\boxed{\ \E_{\Pbar}\!\big[b\,\mathbf 1(\hat s\ne s)\big]\ \le\ \varepsilon_1\,\Pbar(0<b\le\varepsilon_1)\ \le\ C\varepsilon_1^{1+\gamma}.\ }
\]
The sign-recovered signed observation $\hat b^{\mathrm{obs}}:=\hat s\,(y-\eta_0)$ has,
via the pointwise identity $(\hat s-s)g=-2\,b\,\mathbf 1(\hat s\ne s)$,
\[
\begin{aligned}
\beta_1^{\mathrm{sign}}(R)
  &:=\E[\hat b^{\mathrm{obs}}-b\mid R]
    =\E[(\hat s-s)g\mid R]\\
  &=-2\,\E[b\,\mathbf 1(\hat s\ne s)\mid R],
\qquad
\|\beta_1^{\mathrm{sign}}\|_{L^1(\barmu)}
  \le 2C\varepsilon_1^{1+\gamma}.
\end{aligned}
\]
Off $E_1$ the contribution is $\le O(1)\sum_jk_j^{-2}\asymp K^3N^{-2}=o(N^{-1/2})$
($K$ fixed). A dyadic-shell refinement of $(\star)$ gives the same order with a
$\gamma$-independent constant.
\end{lemma}

\noindent\emph{Remark (plug-in smoothing bias versus sign-error bias).} The naive plug-in
$|\hat\eta_J-\eta_0|$ carries the deterministic Stage-1 \emph{smoothing} bias
$\asymp\varepsilon_1$ for \emph{all} $\gamma$ (this is exactly Rmk.~G.1.5(a)) and is
\textbf{not} the estimator here. Route~A uses the \emph{raw} label, for which
$\E[y-\eta_0\mid\mathcal F_\phi]=g$ is unbiased and the \emph{only} bias is the weighted
sign error above; consequently $\|\beta_1^{\mathrm{sign}}\|_{L^1}\lesssim\varepsilon_1^{1+\gamma}$
\emph{with no cap at $\varepsilon_1^2$}. (No heavy tail $\sigma^2/b$ arises: for
sub-Gaussian noise the pointwise $|\cdot|$-bias decays
super-exponentially for $b\gg\sigma$, so the integrated Jensen bias converges to
$\varepsilon_1^{1+\gamma}$ for every $\gamma>0$; no $\as>d/2$ clause is needed.)

\setcounter{theorem}{1}%
\begin{proposition}[linearization: regime (II) $\to$ regime (I)]\label{prop:H2}\textbf{[proven]}
With the sign recovered, $b=sg$ is linear in $\bar\eta_J$ on the correctly-signed region, so
$\hat b^{\mathrm{obs}}=\hat s(y-\eta_0)$ \emph{is} the regime-(I) direct/sign-coherent object
of Section~G: conditionally on the frozen $\hat s$ the label fluctuation $y-\bar\eta_J$ is
mean-zero given $R$ (variance, not bias), and the entire outcome bias is
$\beta_1^{\mathrm{sign}}$ of Lem.~\ref{lem:H1}. The fold-balanced anchor-level split of
Def.~G.1.2 freezes $\hat s$ off the evaluation fold, so $\hat b^{\mathrm{obs}}_{ji}$ has
conditionally-independent, mean-zero fluctuation given $R$ (this neutralizes the second,
in-sample-fluctuation obstruction of Rmk.~G.1.5(b)). The EIF
$\phi=2[w(b-\Psi)+(\Psi-\theta)]$, the DML2 one-step $\hat D_{\mathrm{cov}}$ (Def.~G.1.2 with
$\hat b^{\mathrm{obs}}$), and the double-robust product remainder $R_{\mathrm{prod}}
=\int(\hat\Psi-\Psi)(w-\hat w)\,d\barmu$ (Lem.~G.1.3) apply verbatim.
\end{proposition}
\noindent\emph{Remark (no orthogonality in the sign nuisance).} We do \emph{not} invoke
``Neyman-orthogonality in the sign nuisance'': $s\in\{\pm1\}$ admits no score/tangent, so the
Gateaux-derivative-in-$s$ language is not defined. The $\varepsilon_1^{1+\gamma}$ bias is
established self-containedly by the exact identity and margin bound of Lem.~\ref{lem:H1};
nothing downstream uses orthogonality-in-$s$.

\setcounter{theorem}{2}%
\begin{theorem}[Route~A $\sqrt N$ CLT above threshold]\label{thm:H3}\textbf{[partial:
new argument proven; CLT inherits the Section-G regime-(I) foundation]}
Assume (A0)--(A7), (A5$'$), interior regime, (T), (Mgn$_\gamma$) with $\gamma>\gamma^\star$,
$K$ fixed, $k\to\infty$, the within-window split (Def.~G.1.2), the sup-norm rate
$\varepsilon_1\asymp(\log k/k)^{a}$ (Lemma~\ref{lem:C11prime}), and (H1)--(H4), (H6). Then the sign-corrected
DML2 one-step
\[
\hat D_{\mathrm{cov}}=\frac2m\sum_l\hat\Psi\big(\hat r_{K+1}(\tilde x_l)\big)
+\frac2N\sum_\ell\sum_{(j,i)\in I_\ell}\hat w^{(-\ell)}(\hat r_j)\big(\hat b^{\mathrm{obs}}_{ji}-\hat\Psi^{(-\ell)}(\hat r_j)\big)
\]
satisfies $\sqrt N(\hat D_{\mathrm{cov}}-\Dcov)\Rightarrow\mathcal N(0,V)$,
$V=4\E_{\barmu}[w^2\sigma_b^2]+4\tau_{NT}\Var_\nu(\Psi)$ with $\tau_{NT}:=\lim N/m$ (Prop.~G.1.1; efficient constant
$V^\star_{A1}$ under the generalized-method-of-moments pooling of Lem.~G.3.1), for the magnitude wall
$\Dcov=2\int_\J\Psi\,d\nu$, $\Psi(r)=\E[|\bar\eta_J-\eta_0|\mid R{=}r]$. The Hájek--Le~Cam
local-asymptotic-minimax (LAM) bound (Thm.~G.3.2(a)) supplies the matching $N^{-1/2}$ floor at fixed super-threshold
$\gamma$, so the rate is two-sided there.
\end{theorem}
\begin{proof}[Key steps and the threshold computation]
Decompose as in Thm.~G.1.4:
$\hat D_{\mathrm{cov}}-\Dcov=2[(\nu_m-\nu)\hat\Psi+(\Pbar_N-\Pbar)(\hat w(\hat b^{\mathrm{obs}}-\hat\Psi))]+2R_2$,
$R_2=R_{\mathrm{prod}}+R_{\mathrm{ratio}}+R_{\mathrm{stage1}}^{\mathrm{sign}}$, with
$R_{\mathrm{stage1}}^{\mathrm{sign}}=\int\hat w\,\beta_1^{\mathrm{sign}}\,d\barmu$. Then
$R_{\mathrm{prod}}=o_P(N^{-1/2})$ by (H3) (Lem.~G.1.3), $R_{\mathrm{ratio}}=o_P(N^{-1/2})$ by
(H6), and $|R_{\mathrm{stage1}}^{\mathrm{sign}}|\le\Gcov\|\beta_1^{\mathrm{sign}}\|_{L^1}
\le2C\Gcov\varepsilon_1^{1+\gamma}$. With $\varepsilon_1\asymp k^{-a}$ (logs absorbed by the
strict inequality) and $k=N/K$,
\[
\begin{gathered}
\varepsilon_1^{1+\gamma}\asymp(N/K)^{-(1+\gamma)a}\overset{K\text{ fixed}}{\asymp}N^{-(1+\gamma)a},\\
N^{-(1+\gamma)a}=o(N^{-1/2})\iff(1+\gamma)a>\tfrac12\iff
\boxed{\ \gamma>\gamma^\star=\frac1{2a}-1=\frac{d}{2\as}.\ }
\end{gathered}
\]
Cross-fitting replaces any Donsker condition; Lindeberg--Feller on the two independent bounded
mean-zero averages plus Slutsky give the CLT (Lemma~\ref{lem:G14prime} with the
Lindeberg--Feller step, \hyperref[app:discharges]{auxiliary-lemmas section}).
\end{proof}

\noindent\emph{Qualifications.} Theorem~\ref{thm:H3} is a
\emph{reduction}, not an unconditional closure. Its \emph{new} content is \emph{proven},
namely the weighted sign error (Lem.~\ref{lem:H1}), the bias identity, and the threshold
$\gamma^\star$. Its CLT conclusion inherits the two Section-G regime-(I) qualifications graded
\emph{partial}, (i) and (ii), and one input, (iii): (i) the efficient constant is $V^\star_{A1}$ (GMM/inverse-variance pooling),
$=V^\star$ only under disjoint novelty coverage; (ii) the ratio channel (A5$'$) is
\emph{not} removed, so (H6) stands as $\varepsilon_r=o(N^{-1/2})$ (Conj.~G.2.3 open); (iii)
the straddle $(\star)$ uses the \emph{sup-norm} Stage-1 rate, supplied by
Lemma~\ref{lem:C11prime} (\hyperref[app:discharges]{auxiliary-lemmas section}) under the minimal-mass
reading of (A3); an $L^2$-only Stage-1 rate would give a different
Audibert--Tsybakov exponent. Thus the correct
reading is ``\emph{reduces regime (II) to regime (I) above the margin threshold, modulo the
Section-G regime-(I) qualifications}.''

\setcounter{theorem}{3}%
\begin{corollary}[reach limitation and the growing-$K$ threshold]\label{cor:H4}\textbf{[proven]}
\emph{(a) Generic-margin reach.} Under (T) with a genuine crossing (bounded positive density
of $g$ at $0$, which is the generic interior case and the only one where the sign is truly
ambiguous), the margin exponent is $\gamma=1$ (the tube $\{0<b\le t\}$ has mass $\asymp t$).
Then $\gamma>\gamma^\star$ requires $\gamma^\star<1$, i.e.\ $\boxed{d<2\as}$. So at the
generic transversal margin Route~A's $\sqrt N$ is non-vacuous only in \emph{low intrinsic
dimension / high smoothness}; for large-$d$ representations with $d\ge2\as$ the generic
crossing lies \emph{at or below} threshold: below it falls under Route~B, and at
equality no claim is made (Thm.~\ref{thm:H9}(c)). This is a scope
statement, not a defect: $\gamma^\star>0$ always ($a<\tfrac12$), and the naive
magnitude reconstruction is the $\gamma=0$ case, which needs $a>\tfrac12$ (impossible for
$K\ge1$). That is the exact failure of (H5) that Route~A removes.
\emph{(b) Growing $K$.} If $K=K_N$ with $\log K/\log N\to\lambda\in[0,\tfrac12)$
(so $k\asymp N^{1-\lambda}$), solving $(N/K)^{-(1+\gamma)a}=o(N^{-1/2})$ gives
$\gamma^\star(\as,d,\lambda)=\frac1{2a(1-\lambda)}-1$. The restriction $\lambda<\tfrac12$ is
\emph{necessary}: the off-$E_1$ failure budget is $\sum_jk_j^{-2}=Kk^{-2}\asymp N^{3\lambda-2}$,
which is $o(N^{-1/2})$ iff $\lambda<\tfrac12$; for
$\lambda\ge\tfrac12$ the centering is uncontrolled unless the Stage-1 tail is strengthened
from $k^{-2}$ to $k^{-c}$ with $Kk^{-c}=o(N^{-1/2})$ (supplied by the tail generalization
of Lemma~\ref{lem:C11prime}, \hyperref[app:discharges]{auxiliary-lemmas section}). Along growing $K$
the CLT additionally requires a variance-stabilization hypothesis; see
the hypothesis check following Lemma~\ref{lem:G14prime} (\hyperref[app:discharges]{auxiliary-lemmas section}).
\end{corollary}

\bigskip\hrule\bigskip
\subsection{Route B: the Cai--Low floor below the margin threshold}

Below threshold the wall is a genuine non-smooth-functional object. We state the mechanism,
the conditional $\gamma=0$ floor, and the remaining gaps: the affinity estimate needed by
the lower-bound tool, gap~(2c) (the quantitative resummation estimate for the
permutation-mixture affinity, Appendix~K of \OAname{}), is open, and the
positive-margin adaptation is unwritten.

\setcounter{theorem}{4}%
\begin{proposition}[polynomial-approximation moment duality]\label{prop:H5}\textbf{[proven]}
With $\delta_D:=\inf_{\deg p\le D}\sup_{|v|\le1}\big||v|-p(v)\big|=\beta_*/D\,(1+o(1))$
(Bernstein constant $\beta_*\approx0.2802$),
\[
\sup\Big\{\E_{\nu_1}|v|-\E_{\nu_0}|v|:\ \nu_0,\nu_1\in\mathcal P([-1,1]),\ \textstyle\int v^\ell d\nu_0=\int v^\ell d\nu_1,\ \ell=0,\dots,D\Big\}=2\delta_D
\]
by LP duality of the moment problem; the extremal pair sits on the Chebyshev alternation
points of the best degree-$D$ approximant. Two drift-height priors matching $D$ moments are
near-indistinguishable from Bernoulli labels, yet the $|\cdot|$-functional they induce differs
by $\asymp1/D$. This non-regularity lives entirely in the labels$\to b$ map, not in the
linear $\int(\cdot)\,d\nu$ (Prop.~G.1.5).
\end{proposition}

\setcounter{theorem}{5}%
\begin{theorem}[single-window floor, $K=1$]\label{thm:H6}\textbf{[partial: cited-standard
(Thm.~E.3); shape-class admissibility is supplied by the Appendix-K derandomization already at $K=1$;
its affinity bound is modulo gap (2c) (Thm.~K.6)]}
For $K=1$, $k$ anchors, over the Hölder$(\as)$ interior class in the zero-crossing
(no-margin) regime, $\inf_{\hat\theta}\sup\E|\hat\theta-\theta|\gtrsim k^{-a}/\mathrm{polylog}(k)$,
i.e.\ $k^{-a}$ up to logarithmic factors. Partition
the $d$-dimensional $\phi$-support into $m=h^{-d}$ cells of side
$h=k^{-1/(2\as+d)}$; on each cell set $g=A v_c$ with $A=h^{\as}=k^{-a}$ (Hölder-maximal); draw
$v_c\stackrel{iid}\sim\nu_\iota$ (the moment-matched pair of Prop.~\ref{prop:H5}). With
$A^2n_c=kh^{2\as+d}=\Theta(1)$ and $D\sim\log k/\log\log k$ so $(D{+}1)!\gtrsim m$, the mixture
$\chi^2=O(1)$ while the functional gap is $A\cdot2\delta_D\asymp k^{-a}(\log\log k/\log k)$;
Le~Cam's two-point method then gives the bound. This is the Cai--Low/Lepski--Nemirovski--Spokoiny $\int|f|$ rate $=$ the pointwise
regression rate (no averaging gain, never $\sqrt N$-parametric); the sketch gives
$k^{-a}\log\log k/\log k$, and the exact logarithmic power together with the
uniform-over-cells Bernoulli$\leftrightarrow$Gaussian transfer are cited from \citet{cailow2011} and
\citet{nussbaum1996} (the cited-standard step Thm.~E.3 already carries).
\end{theorem}

\setcounter{theorem}{6}%
\begin{theorem}[fixed-$K$ lower bound at $\gamma=0$]\label{thm:H7}
\textbf{[conditional: the admissible-class version routes through the
Appendix-K derandomization and requires the open affinity gap (2c) (Thm.~K.6),
in addition to the cited-standard transfer steps in Thm.~E.3]}
Place the $K$ windows on \emph{disjoint} novelty cells $\J_1,\dots,\J_K$ (Stage-1 cannot pool,
Rmk.~C.14), each of $\nu$-mass $1/K$, $w\equiv1$. Let each window be independently ``active''
with probability $p$ under $\Lambda_0$ and $p'$ under $\Lambda_1$, with $|p'-p|=c/\sqrt K$,
active state $=$ the fuzzy prior $\nu_1$ of Prop.~\ref{prop:H5} at budget $k$, coherently
oriented. Then the total data-affinity is bounded,
\[
\chi^2(\Lambda_1\Vert\Lambda_0)=\prod_{j=1}^K\big(1+\chi^2_j\big)-1
\le e^{\sum_j\chi^2_j}-1,\qquad
\sum_{j=1}^K\chi^2_j\asymp K\cdot\frac{c^2}{K}\cdot O(1)=O(1),
\]
so $\mathrm{TV}(\Lambda_1,\Lambda_0)\le\tfrac12$ for $c$ sufficiently
small,
while the functional separates by
$\E_{\Lambda_1}\theta-\E_{\Lambda_0}\theta=(p'-p)\cdot2\delta_D A\asymp
(N/K)^{-a}\,\frac{\log\log k}{\log k}\Big/\sqrt K$, the same logarithmic factor as in
Thm.~\ref{thm:H6} (it comes from $\delta_D\asymp1/D$ at the degree $D\sim\log k/\log\log k$
needed for $\chi^2=O(1)$). Le~Cam gives, conditional on gap~(2c), in the no-margin case
$\gamma=0$,
\[
\inf_{\hat\theta}\ \sup\ \E|\hat\theta-\theta|\ \ge\ c'\,\frac{(N/K)^{-a}}{\sqrt K\,\mathrm{polylog}(k)}
\ =\ c'\,\frac{k^{-a}}{\sqrt K\,\mathrm{polylog}(k)},
\]
with the same logarithmic power as in Thm.~\ref{thm:H6}. For \emph{fixed} $K$ this is
$k^{-a}$ up to logarithmic factors, and
$k^{-a}/(\sqrt K\,\mathrm{polylog}(k))\big/N^{-1/2}=k^{1/2-a}/\mathrm{polylog}(k)\to\infty$
(since $a<\tfrac12$): $\sqrt N$-estimability is \emph{impossible} and $\theta$ is
\emph{non-regular}, for every fixed $K\ge1$.
\end{theorem}

\setcounter{theorem}{7}%
\begin{remark}[why Assouad does not apply, and the $\sqrt K$ gap]\label{rmk:H8}
A general-$K$ derivation might invoke Assouad over
$\tau\in\{0,1\}^K$ with $\theta(P_\tau)=\theta_{\mathrm{base}}+(A\delta_D/K)\sum_j\tau_j$. But
this functional depends on $\tau$ \emph{only through the scalar} $\sum_j\tau_j$: the induced
separation is
$|\theta_\tau-\theta_{\tau'}|=(A\delta_D/K)|\sum(\tau-\tau')|\le(A\delta_D/K)\,\rho_H(\tau,\tau')$,
an \emph{upper} bound by Hamming distance, whereas Assouad's lemma as stated in
\citet[Thm.~2.12]{tsybakov2009introduction} requires a
\emph{lower} bound $\ge2s\,\rho_H$; per-coordinate bit decoding is impossible for a projection
functional. The correct tool is the fuzzy Le~Cam bound of Thm.~\ref{thm:H7}, which costs a
\emph{genuine} $\sqrt K$ relative to the rate $(N/K)^{-a}$. Consequences:
\begin{enumerate}[nosep,label=(\roman*)]
\item \textbf{Fixed $K$:} non-regularity/$\sqrt N$-impossibility is
Thm.~\ref{thm:H7}'s conclusion, conditional on gap~(2c) and the Thm.~E.3 transfer steps; the rate is $k^{-a}$ up to logarithmic factors and the constant $\sqrt K$.
\item \textbf{Growing $K$:} the $\sqrt K$ gap diverges. Only the sandwich
$[\,(N/K)^{-a}/(\sqrt K\,\mathrm{polylog}),\ (N/K)^{-a}\cdot\mathrm{polylog}\,]$ is available: its lower end is
Thm.~\ref{thm:H7}'s floor, conditional on gap~(2c), and its upper end is the minimax
rate conditional on the same gap (Thm.~I.2.1 of Appendix~I, \OAname{};
Proposition~\ref{prop:growingK}). For the
variance-averaged estimator the coherent linear functional does average down by $\sqrt K$.
\item \textbf{Admissibility:} the $|\cdot|$-nonsmoothness caps
$\Psi\lesssim A=k^{-a}$ per cell, so keeping all $2^K$ sign patterns monotone-Lipschitz (A4)
forces $K\cdot2\delta_D A\lesssim A$, i.e.\ $K\lesssim\log k/\log\log k$. Beyond that the
cube is inadmissible in $\mathcal P^L_{\mathrm{mon}}$. This cap applies to the
i.i.d.-heights cube only; the derandomized, phase-complementary construction of
Appendix~K (\OAname{}) is admissible at every $K$; only its
permutation-mixture affinity bound remains modulo gap~(2c).
\item \textbf{Margin realizability:} the zero-crossing-with-density drift gives
$\gamma=1$, which is sub-threshold ($\gamma<\gamma^\star$) iff $\as<d/2$; for $\as\ge d/2$ a
higher-order-vanishing ($\gamma<1$) crossing must be built and its Hölder membership checked
(not carried out here); the margin adaptation of Thm.~\ref{thm:H7} to
$0<\gamma<\gamma^\star$ is unwritten.
\item \textbf{Upper bound:} the crude per-window plug-in has coherent, non-cancelling
$|\cdot|$-Jensen bias $\asymp k^{-a}$ but is \emph{not} a certificate of the floor; the
matching achiever is the poly-debiased U-statistic of Prop.~\ref{prop:H10} (rate-sharp only at
$K=1$).
\end{enumerate}
\end{remark}

\bigskip\hrule\bigskip
\subsection{The regular branch, boundary estimator, and lower-bound status}

\setcounter{theorem}{8}%
\begin{theorem}[regular branch and lower-bound status for the magnitude wall]\label{thm:H9}\textbf{[partial]}
Let $\gamma^\star=d/(2\as)$, $a=\as/(2\as+d)\in(0,\tfrac12)$. Under
(A0)--(A7), (A5$'$), (T), (Mgn$_\gamma$), $K$ fixed, $k\to\infty$, and for the
within-window-split one-step, the following hold:
\begin{enumerate}[nosep,label=(\alph*)]
\item \emph{Above threshold $\gamma>\gamma^\star$ (Route~A):} $\Dcov$ is \emph{regular},
rate $N^{-1/2}$, EIF $\phi=2[w(b-\Psi)+(\Psi-\theta)]$, efficient constant $V^\star_{A1}$
(Thm.~\ref{thm:H3}), \emph{modulo the Section-G regime-(I) qualifications}.
\item \emph{Below threshold $\gamma<\gamma^\star$ (Route~B):} at $\gamma=0$,
Thm.~\ref{thm:H7} gives a conditional fixed-$K$ lower bound
$(N/K)^{-a}/\sqrt K$ up to logarithmic factors, implying non-regularity if gap~(2c) and the cited
transfer steps hold. For $0<\gamma<\gamma^\star$, both the proposed
margin-adaptive upper bound (Prop.~\ref{prop:H10}) and the matching lower-bound
adaptation are open.
\item \emph{Boundary point:} $a(1+\gamma^\star)=\tfrac12$ \emph{identically},
so the two polynomial exponents meet at $\gamma=\gamma^\star$. This identity does not settle the
boundary case: the Stage-1 rate contains a logarithmic factor, and Route~A requires the strict
little-$o$ condition $\varepsilon_1^{1+\gamma}=o(N^{-1/2})$.
\end{enumerate}
Consequently, for fixed $K$ (the first line theorem-grade modulo the regime-(I)
qualifications, in particular hypothesis (H6); the second conditional on gap~(2c)),
\[
\begin{aligned}
\inf_{\hat D}\ \sup\ \E|\hat D-\Dcov|
&\asymp N^{-1/2}
&&\bigl(\gamma>\gamma^\star\bigr),\\
\inf_{\hat D}\ \sup\ \E|\hat D-\Dcov|
&\gtrsim (N/K)^{-a}/(\sqrt K\,\mathrm{polylog}(k))
&&\bigl(\gamma=0;\ \text{conditional}\bigr).
\end{aligned}
\]
\end{theorem}
\noindent\emph{Scope.} The $\gamma=0$ lower side is conditional on gap~(2c) even at
$K=1$, because the admissible construction depends on it.

\setcounter{theorem}{9}%
\begin{proposition}[the boundary estimator: Cai--Low/LNS polynomial debiasing]\label{prop:H10}
\textbf{[partial ($K{=}1$); conjecture ($K{>}1$)]}
Per window and per width-$h$ cell ($n_c=kh^d$ anchors, drift $\mu_C=\bar\eta_j-\eta_0$), the estimator proceeds in four steps:
(1)~take the best minimax polynomial $P_{2D}$ of $|x|$ on $[-M,M]$, $M\asymp k^{-a}$, with uniform error
$\asymp M/D$; (2)~form the unbiased $U$-statistic $\hat U_\ell=\binom{n_c}{\ell}^{-1}\!\sum_{\text{distinct}}
\prod_p(y_{i_p}-\eta_0)$ for $\mu_C^\ell$ ($\eta_0$ known); (3)~set $\widehat{|\mu_C|}=\sum_\ell
c_\ell\hat U_\ell$, which has bias $\le M/D$ and variance $\asymp M^2 2^{O(D)}/n_c$; (4)~choose the degree $D^\star\asymp\log n_c$
so that bias$^2\asymp$var. Aggregating $\tfrac1N\sum_{\text{cells}}\widehat{|\mu_C|}\,w$ attains
$k^{-a}=(N/K)^{-a}$ up to logarithmic factors; with Thm.~\ref{thm:H6}'s lower bound this is
\emph{rate-sharp up to logarithmic factors at $(\gamma{=}0,K{=}1)$, conditional on gap~(2c)}.
\emph{This is the estimator that realizes
the clean $\varepsilon_1^{1+\gamma}$ rate} (raw-label $U$-statistics carry no
$\hat\eta$-smoothing bias), together with Route~A's sign-recovery step, and \emph{not} the
plug-in $|\hat\eta_J-\eta_0|$. For $0<\gamma<\gamma^\star$ the margin-adaptive version applies
poly-debiasing only in the near-crossing layer (mass $\lesssim\varepsilon_1^{\gamma}$) and the
linear-bulk one-step elsewhere, targeting $(N/K)^{-a(1+\gamma)}$ (the aggregation of the $2^{O(D)}/n_c$ cell
variances is not carried out here, so this achievability is partial). The general-$K$ aggregated lower bound is the
$\sqrt K$ gap of Rmk.~\ref{rmk:H8}.
\end{proposition}

\setcounter{theorem}{10}%
\begin{theorem}[regular branch of Conjecture~E.6; Section~G as the perfect-margin endpoint]\label{thm:H11}
\textbf{[partial]}
For fixed $K$, the covered-part covariate wall $\Dcov=2\int_\J\Psi\,d\nu$ is $\sqrt N$-regular
(rate $N^{-1/2}$, EIF $\phi$, efficient constant $V^\star_{A1}$) under
\[
\underbrace{w\le\Gcov}_{\text{(A6) overlap}}\ \wedge\ \underbrace{\Pbar(\bar\eta_J=\eta_0)=0}_{\text{(T)}}\ \wedge\ \underbrace{\gamma>\gamma^\star=\tfrac{d}{2\as}}_{\text{(Mgn}_\gamma)},
\]
\emph{for all $\as$} (no $\as>d/2$ clause; the near-crossing part has mass
$\lesssim\varepsilon_1^{\gamma}$ and the bulk $\int(\bar\eta_J-\eta_0)\,s\,w\,d\nu$ is
linear-in-$\Psi$, hence $\sqrt N$-regular at any smoothness by Prop.~G.1.5, so that no
higher-order influence function is needed).
At $\gamma=0$, the conditional lower bound of Theorem~\ref{thm:H7} implies
non-regularity if gap~(2c) holds. For $0<\gamma<\gamma^\star$, the lower-bound
adaptation remains open; no claim is made here at $\gamma=\gamma^\star$.

\smallskip
\emph{Reconciliation.} Section~G's regime~(I) is the \emph{certified-sign / $\beta_1\equiv0$}
branch: the signed drift is observed, so $\beta_1\equiv0$ for every $\gamma$. It is the
\emph{perfect-margin endpoint} of the magnitude problem in the precise sense that as the
sign-recovery penalty $\varepsilon_1^{1+\gamma}\to0$ (margin $\gamma\to\infty$, or the sign
supplied) Route~A's object converges to regime~(I)'s directly-observed object. Section~G's own
side condition $\varepsilon_1^{1+\gamma}=o(N^{-1/2})$ (Thm.~G.1.4) is \emph{algebraically
identical} to $\gamma>\gamma^\star$; and its ``(H5) fails for all $K\ge1$'' is exactly the
$\gamma=0$ (no-margin) instance, since then $\|\beta_1\|_{L^1}\asymp\varepsilon_1$ needs
$a>\tfrac12$. Thus Section~G sits on the regular side of the frontier:
(H5) holds above $\gamma^\star$; equality remains unresolved, and its failure
below the threshold is conditional/open as described above. With bounded
overlap and transversality, the magnitude wall recovers $\sqrt N$ when drift stays sufficiently
clear of the decision boundary $\bar\eta_J=\eta_0$. The conditional
lower-bound program points to a non-parametric Cai--Low price for
near-boundary drift but does not yet establish it throughout the sub-threshold region.
\end{theorem}

\bigskip\hrule\bigskip
\subsection{Summary of status (Section H)}

The per-claim status record for this section (claims, assumptions and
tools used, weakest steps) is in \S S1 of \OAname{}.
In summary, Conjecture~E.6 (regime~II, the magnitude wall) has an established regular
branch above $\gamma^\star=d/(2\as)$, where the wall is
$\sqrt N$-regular (Route~A, Thm.~\ref{thm:H3}, modulo the Section-G regime-(I)
qualifications). At $\gamma=0$, the fixed-$K$ lower bound is conditional on
the Appendix-K gap~(2c); the positive-margin adaptation below the threshold
is open. Equality is unresolved under the logarithmic Stage-1 rate. The
general/growing-$K$ conditional program is taken up in Appendix~I of \OAname{}.

%% file: appendix/app_discharges.tex
\section*{\texorpdfstring{Auxiliary lemmas C.11$'$ and G.14$'$}{Auxiliary lemmas C.11-prime and G.14-prime}}
\phantomsection\label{app:discharges}
\addcontentsline{toc}{section}{Auxiliary lemmas C.11-prime and G.14-prime}

\noindent\emph{Scope.} This section proves two steps that Appendices C and G of
\OAname{} and Appendix~\ref{app:H} use: the Stage-1 sup-norm $k$-NN rate behind Lem.~C.11
(Lemma~C.11$'$ below; it is non-adaptive, tuning $\kappa$ with known $(\as,d)$,
which suffices downstream), and the empirical-to-population and CLT steps of
Thm.~G.1.4, used again in Thm.~H.3 (Lemma~G.14$'$ and the Lindeberg--Feller step
below). The comparison with the published theorems is in \S S2 of
\OAname{}. Step (3) of Thm.~C.12 is not covered here (closing paragraph).

\paragraph{Item 1: Stage-1 sup-norm $k$-NN (Lem.~C.11).} The published uniform
$k$-NN rates do not cover the sup-norm rate in the form Lem.~C.11 needs
(\S S2 of \OAname{} compares them); the lemma below proves it directly,
following the scheme of \citet{jiang2019uniform} under (A3)-compatible hypotheses
with two standard concentration inequalities.

\begin{astmt}{Lemma C.11$'$ (sup-norm variable-radius $k$-NN under minimal mass). [proven]}\taglabel{C.11$'$}{lem:C11prime}
Let $x_1,\dots,x_n$ be i.i.d.\ $q$ on a region $R\subset\mathbb R^{D}$ ($\phi$-coordinates,
$D=$ ambient representation dimension), $y_i\in[0,1]$, $\eta(x)=\E[y\mid x]$
$(L,\as)$-H\"older on $R$. Assume the \emph{minimal-mass condition}
\[
\textbf{(M)}\qquad q\big(B(x,r)\big)\ \ge\ c_0\,r^{d}\qquad\forall x\in R,\ 0<r\le r_0 .
\]
\emph{Scope of (M).} When the lemma is invoked per window $j=1,\dots,K$ (as in Thm.~H.3's
straddle step), (M) is required with a \emph{single} pair of constants $(c_0,r_0)$
\emph{uniform over $j$}, and must hold on the \emph{full anchor-region support} of $p_j$
(every $x\in R_j$, not merely the realized anchor sites), because H.1's straddle reads the
sign at \emph{population} points ($\beta_1^{\mathrm{sign}}(R)=\E[\cdot\mid R]$ integrates over
$X\sim\Pbar$ on that support).
Let $\hat\eta$ be the $\kappa$-NN average (the convention of \citet{jiang2019uniform}: average of $y_i$ over
$N_\kappa(x)=B(x,r_\kappa(x))\cap\{x_i\}$, $r_\kappa(x)=$ $\kappa$-NN radius; ties enlarge
$N_\kappa$, which only helps) with
$\kappa=\lceil n^{2\as/(2\as+d)}(\log n)^{d/(2\as+d)}\rceil$. (Log convention: throughout
this lemma and its proof, $\log$ denotes the natural logarithm, $\log=\ln$.) Then there is
$n_0=n_0(D,d,\as,c_0,r_0)$ such that for $n\ge n_0$, with probability $\ge1-n^{-2}$,
\[
\sup_{x\in R}\big|\hat\eta(x)-\eta(x)\big|\ \le\
\Big[\,L\big(4/c_0\big)^{\as/d}+2\sqrt{D+4}\,\Big]\Big(\frac{\log n}{n}\Big)^{\as/(2\as+d)} .
\]
\end{astmt}

\begin{proof}
\textbf{Step 0 ((A3)$\Rightarrow$(M)).} If $q$ is measure-doubling with exponent $d$ on $R$
($q(B(x,2r))\le2^{d}q(B(x,r))$, the operative meaning of ``doubling intrinsic dimension
$\le d$'') and $\Delta=\operatorname{diam}R$, then iterating $\ell=\lceil\log_2(\Delta/r)\rceil$
times gives $q(B(x,r))\ge2^{-d\ell}q(B(x,\Delta))\ge(r/2\Delta)^{d}q(R)$: (M) holds with
$c_0=q(R)(2\Delta)^{-d}$, $r_0=\Delta$. Alternatively, density bounds w.r.t.\ an
Ahlfors-$d$-regular reference measure give (M) directly. Assumption (A3) is read here in its minimal-mass form (M); this argument records when (M) follows from the stronger measure-doubling reading printed in Section~\ref{sec:assumptions}.

\textbf{Step 1 (uniform radius; bias).} Balls in $\mathbb R^{D}$ have VC dimension $D+1$. By the
VC relative-deviation bound (\citealp[Thm.~5.1]{bousquet2004introduction}, as paraphrased in
\citealp[Thm.~15]{chaudhuri2010cluster}): w.p.\ $\ge1-\delta/2$, every ball $B$
satisfies $q_n(B)\ge q(B)-\beta_n\sqrt{q(B)}$ with
$\beta_n=\sqrt{(4/n)\big((D{+}1)\ln 2n+\ln(16/\delta)\big)}$. Set
$h:=(2\kappa/(c_0n))^{1/d}$ (require $h\le r_0$ and $\kappa\ge8((D{+}1)\ln2n+\ln(16/\delta))$,
i.e.\ $\beta_n^2\le\kappa/2n$; both hold for $n\ge n_0$). For any $x\in R$,
$q(B(x,h))\ge c_0h^{d}=2\kappa/n$, and $t-\beta_n\sqrt t\ge\kappa/n$ for all $t\ge2\kappa/n$
under $\beta_n^2\le\kappa/2n$; hence $q_n(B(x,h))\ge\kappa/n$, i.e.\
$\sup_{x\in R}r_\kappa(x)\le h$. Since every point of $N_\kappa(x)$ lies within
$r_\kappa(x)\le h$: $\big|\,|N_\kappa(x)|^{-1}\sum_{i\in N_\kappa(x)}\eta(x_i)-\eta(x)\big|
\le L h^{\as}$, uniformly. \emph{(No cover argument or per-cover-point stability of the NN sets is needed.)}

\textbf{Step 2 (design-conditional noise; realizable-NN-set counting).} Condition on the design
$(x_1,\dots,x_n)$: $\xi_i:=y_i-\eta(x_i)$ are independent, mean-zero, $|\xi_i|\le1$. Every
realizable $N_\kappa(x)$ is a subset of the design cut out by a closed Euclidean ball, of size
$\ge\kappa$; by Sauer's lemma (VC dim $D{+}1$) the number of distinct ball-cut subsets is
$M\le(n+1)^{D+1}$. (We use Sauer's bound instead of the arrangement count in
\citet[Lem.~3]{jiang2019uniform}; it is conservative.) Hoeffding's inequality for a fixed set $A$ with $|A|=s\ge\kappa$ gives
$\Pr(|s^{-1}\sum_{i\in A}\xi_i|>t)\le2e^{-st^2/2}\le2e^{-\kappa t^2/2}$. A union bound with
$t=\sqrt{2((D{+}1)\ln(n{+}1)+\ln(4/\delta))/\kappa}$ gives, w.p.\ $\ge1-\delta/2$ conditionally
on \emph{any} design (hence unconditionally),
$\sup_x\big||N_\kappa(x)|^{-1}\sum_{i\in N_\kappa(x)}\xi_i\big|\le t$.

\textbf{Step 3 (combine).} Take $\delta=n^{-2}$. Then $\ln(4/\delta)\le2\ln n+\ln4$, and with
the stated $\kappa$ the bias is $\le Lh^{\as}=L\big(2\kappa/(c_0n)\big)^{\as/d}\le
L(4/c_0)^{\as/d}(\log n/n)^{\as/(2\as+d)}$ (without the ceiling
$(\kappa/n)^{\as/d}=(\log n/n)^{\as/(2\as+d)}$ exactly; the ceiling on $\kappa$ costs at most
a factor $2$, absorbed into $4/c_0$), while the noise is
$\le2\sqrt{D+4}\,(\log n/n)^{\as/(2\as+d)}$ for $n\ge n_0$. The failure probability is $\le n^{-2}$. \qedhere
\end{proof}

\noindent\emph{Tail generalization.} Repeating the proof with $\delta=n^{-c}$ changes only
constants ($\kappa$'s side condition and the noise constant scale with $c$): Lemma~C.11$'$
holds with probability $\ge1-n^{-c}$ for any fixed $c>0$, which supplies the $k^{-c}$ tail
used in Cor.~\ref{cor:H4}(b).

\paragraph{Item 2: empirical-to-population / CLT steps of Thm.~G.1.4 (echoed in
Thm.~H.3).} \citet{chernozhukov2018dml} (henceforth CCDDHNR), Theorem~3.1, assumes an i.i.d.\ sample,
while the paper's anchor design is window-stratified (independent, not identically
distributed; inid), so it does not apply as published (\S S2 of \OAname{}). Its
fold-conditional proof device and its model-free Lemma~6.1 carry over verbatim; the
short inid extension is written out below, and the CLT is then Lindeberg--Feller
\citep[Prop.~2.27]{vandervaart1998}.

\begin{astmt}{Lemma G.14$'$ (fold-conditional empirical-to-population under window-stratified sampling). [proven]}\taglabel{G.14$'$}{lem:G14prime}
Fix fold $\ell$ with evaluation windows $I_\ell$, each contributing $k$ i.i.d.\ anchors
(all windows independent), $N_\ell=|I_\ell|k$, $\Pbar_\ell:=|I_\ell|^{-1}\sum_{j\in I_\ell}P_j$.
Let $\hat f_\ell$ be measurable w.r.t.\ the training data (independent of the evaluation
anchors) with $\|\hat f_\ell\|_{L^2(\Pbar_\ell)}<\infty$. Then
\[
\E\Big[\big((\Pbar_{N_\ell}-\Pbar_\ell)\hat f_\ell\big)^2\,\Big|\,\mathrm{train}\Big]
=\frac1{N_\ell^2}\sum_{j\in I_\ell}k\,\Var_{P_j}(\hat f_\ell)
\ \le\ \frac{\|\hat f_\ell\|^2_{L^2(\Pbar_\ell)}}{N_\ell},
\]
so $(\Pbar_{N_\ell}-\Pbar_\ell)\hat f_\ell=O_P\big(N_\ell^{-1/2}\|\hat f_\ell\|_{L^2(\Pbar_\ell)}\big)$
unconditionally by CCDDHNR Lemma~6.1(b); in particular $=o_P(N^{-1/2})$ whenever
$\|\hat f_\ell\|_{L^2(\Pbar_\ell)}=o_P(1)$ and $N_\ell\asymp N$.
\end{astmt}
\begin{proof}
Conditionally on the training data $\hat f_\ell$ is a fixed function; the evaluation anchors are
independent with $\E[\Pbar_{N_\ell}\hat f_\ell\mid\mathrm{train}]
=|I_\ell|^{-1}\sum_j\E_{P_j}\hat f_\ell=\Pbar_\ell\hat f_\ell$, an identity that \emph{requires
equal quotas} $k_j\equiv k$ matching $\Pbar$'s uniform mixture weights (A0); for unequal $k_j$,
$\Pbar$ must be re-weighted to $\sum_j(k_j/N)P_j$, which is precisely Appendix C's Convention
C.0 (anchor-weighted pool $\barmu=N^{-1}\sum_jk_j\mu_j$): unequal quotas are handled by the
already-weighted mixture, provided $\Pbar$ (and $\barmu$) are read consistently in that
convention; Appendix G's uniform mixture is the equal-$k$ special case. The variance splits over independent terms, $\Var_{P_j}\le\E_{P_j}[\hat f^2]$, and averaging over $j\in I_\ell$ gives the displayed bound. Chebyshev's inequality, followed by Lemma~6.1(b), gives the unconditional rate.
\end{proof}

\noindent\emph{Application closing step (2).} Take $\hat f_\ell=\hat w^{(-\ell)}(b-\hat\Psi^{(-\ell)})
-w(b-\Psi)$; since $b,\hat\Psi,\Psi\in[0,1]$ and $w\le\Gcov$,
$\|\hat f_\ell\|_{L^2(\Pbar_\ell)}\le\|\hat w-w\|_{L^2}+\Gcov\|\hat\Psi-\Psi\|_{L^2}=o_P(1)$ by
(H3). For the target term, $\hat\Psi$ is independent of the target sample and $\nu_m$ is i.i.d., so
$\Var\le m^{-1}\|\hat\Psi-\Psi\|^2_{L^2(\nu)}\le m^{-1}\Gcov\|\hat\Psi-\Psi\|^2_{L^2(\barmu)}$,
and (H4), $m\ge cN$, gives $o_P(N^{-1/2})$. Both displays of step (2) are proven for
fold-balanced designs, in which every window contributes equally to every fold (the
anchor-level design of Def.~G.1.2), so that $\bar P_\ell=\bar P$ for each fold.
The theorem uses this design. Window-level folding is not covered: there a fold-centering term
$\sum_\ell(N_\ell/N)(\bar P_\ell-\bar P)\hat f_\ell$, of first order in the nuisance
error, appears and is not controlled by the three proof steps. Unequal quotas are
handled by Convention~C.0's anchor-weighted mixture, as noted in the proof.

\paragraph{Step (3), the CLT: precise citation and hypothesis check.} The remaining terms are
$2N^{-1/2}\sum_{j,i}\phi_{\rm src}(O_{ji})+2\sqrt{N/m}\cdot m^{-1/2}\sum_l\phi_{\rm tgt}(\tilde O_l)$,
$\phi_{\rm src}=w(b-\Psi)$, $\phi_{\rm tgt}=\Psi-\theta$: a row-independent triangular array of
\emph{bounded} summands ($|\phi_{\rm src}|\le\Gcov$, $|\phi_{\rm tgt}|\le1$). The CLT used
is the Lindeberg--Feller CLT for independent non-identically distributed triangular arrays,
namely \citet[Prop.~2.27, \S2.8]{vandervaart1998} (independent mean-zero rows,
$\sum_i\Var\to\Sigma$ plus the Lindeberg condition
imply $\sum_iY_{ni}\Rightarrow\mathcal N(0,\Sigma)$), plus Cram\'er--Wold and Slutsky; CCDDHNR's own
proof uses exactly this pair. \emph{Hypothesis check:} (i) Independence across windows, anchors,
and the target sample holds by (A0) and two-sample independence. (ii) The Lindeberg condition
holds automatically: summands are bounded by $C/\sqrt N$ after scaling, so the Lindeberg sum vanishes for
$N>(C/(\epsilon\sqrt V))^2$, given $V>0$. (iii) Variance identification is the genuine content:
under the fixed-quota design the CLT variance is $N^{-1}\sum_jk\Var_{P_j}(\phi_{\rm src})$, which
equals the mixture variance $\Var_{\Pbar}(\phi_{\rm src})$ used in Prop.~G.1.1's $V$ \emph{if and
only if the per-window means vanish}, $\E_{P_j}[\phi_{\rm src}]=0$ $\forall j$. That is exactly
what stationarity (A1) delivers ($\E[b\mid R{=}r,J{=}j]=\Psi(r)$ for every $j$). Without (A1) the
stratified variance is $V-K^{-1}\sum_j(\E_{P_j}\phi_{\rm src})^2$ (strictly $<V$ unless
every per-window mean vanishes), so the CI constant $V$ is conservative; the one-step
estimator nevertheless \emph{remains centered on the pooled} $\theta$, and bias arises only
relative to the \emph{transfer estimand} $D_{\mathrm{cov}}^{(K+1)}$, whose very definition
requires (A1) extended to window $K{+}1$. (A1) is essential for the CLT constant.
(iv) For growing $K=K_N$
(Cor.~\ref{cor:H4}(b)), convergence of $|I_\ell|^{-1}\sum_j\Var_{P_j}$ is a \emph{new} variance-stabilization
hypothesis.

\paragraph{What is not covered.} Step (3) of the proof of Theorem~C.12 (\OAname{}), the inid extension of the \emph{ratio-type} bracketing argument for
the shape-constrained least-squares estimator
\citep[Thm.~2.7.5]{vandervaartwellner1996}, is a separate open point and is not
settled here.

%% file: appendix/app_J.tex
\section{Measurement protocol and the proxy/sensitivity study}
\label{app:J}

This appendix specifies the measurement protocol behind the empirical
sections of the paper, namely the probe, oracle, and gate configurations
of Sections~\ref{sec:anatomy}--\ref{sec:methodology}, the bootstrap
reconstruction of Section~\ref{sec:weld}, and the decision-rule and
attribution evaluations of
Sections~\ref{sec:certificate}--\ref{sec:attribution-procedure}, and it
documents the two evidence levels used in
Section~\ref{sec:map}. Eight TabReD streams provide noisy, single-stream
proxies for all threshold ingredients. Fifteen related benchmarks lack the
labeled temporal drift field needed to estimate $\alpha_s$ and $\gamma$;
they therefore enter only a sensitivity analysis at
$(\alpha_s,\gamma)=(1,1)$. The appendix also records the deduplication
correction that invalidated two raw low-dimensional readings
(Section~\ref{app:J-dedup}). These analyses motivate, but do not establish, a population-level
claim about which phase real data sets occupy. Printed values are
transcribed from the measurement record of the reproduction deposit (the
per-stream record of each estimate and its provenance).

\subsection{Measurement protocol}
\label{app:J-protocol}

\paragraph{Data sets.}
Twenty-three real tabular data sets: the eight industrial TabReD streams
(five regression, three classification;
$n_{\mathrm{test}}=4{,}647$--$59{,}951$ test points), and fifteen shift benchmarks with precomputed test embeddings of the
frozen tabular foundation model, namely the four ACS census tasks (the
canonical shift benchmark), fraud (\textsc{baf}, \textsc{ieee\_fraud},
\textsc{credit\_card\_fraud}), network security (\textsc{unsw\_nb15},
\textsc{cicids}), medical (\textsc{brfss\_diabetes},
\textsc{diabetes\_readmission}, \textsc{sepsis},
\textsc{mimic\_iv\_mortality}), credit (\textsc{lending\_club}), and
\textsc{bike\_sharing} ($n_{\mathrm{test}}=6{,}101$--$1{,}162{,}213$).

\paragraph{Ingredients of the threshold.}
A data-set-specific phase assessment requires three quantities and a
sampling protocol that supports their joint interpretation.

\emph{(i) Deduplication first.} Rows identical to four decimal places are
removed and every intrinsic-dimension estimate is recomputed on the unique
subset. The step matters because duplicate rows produce
first-neighbor distances $r_1\approx 0$ that break both estimators below
and bias $d_0$ sharply \emph{downward}
(Section~\ref{app:J-dedup}).

\emph{(ii) Intrinsic dimension $d_0$.} For the TabReD streams, the eight
$d_0$ values printed in Table~\ref{tab:J-map} are TwoNN
\citep{facco2017twonn} estimates on a $12{,}000$-row subsample of the raw
streams. The streams are $99.8$--$100\%$ unique at four decimals and
deduplication moves the estimate by at most $0.08$
(Table~\ref{tab:J-tabred-dedup}, where the check is run with the MLE
estimator), so the raw and deduplicated readings are interchangeable at
the printed precision. For the fifteen related
benchmarks, $d_0$ is the Levina--Bickel MLE \citep{levina2004mle} on a $12{,}000$-row subsample, with TwoNN and
PCA cross-checks on borderline cases. The two estimators do not agree
pointwise on low-uniqueness data (an instability documented in
Section~\ref{app:J-dedup}); the deduplicated MLE values are the corrected
record.

\emph{(iii) Drift smoothness $\alpha_s$.} On the TabReD streams (where
labels exist), $\alpha_s$ is the H\"older exponent of the drift/residual
field, read as half the variogram log-log slope. The slope saturates at $2$
in the limit (finite-sample slopes can slightly exceed it), so the
variogram cannot certify $\alpha_s>1$; we therefore work under the
rough-field bound $\alpha_s\le 1$. Three qualifications are recorded
in the measurement record: the estimate is taken on the $k$-NN-\emph{smoothed} drift
field and is therefore biased \emph{high} (the raw-residual H\"older on
\textsc{weather} is $\sim 0.5$); the point
estimates are unstable: the slope is refit on the lower half of the
variogram bins alone, and on $5$/$8$ streams the two estimates differ by more
than $0.2$ (\textsc{sberbank-housing}, \textsc{delivery-eta},
\textsc{homecredit-default}, \textsc{homesite-insurance},
\textsc{ecom-offers}; differences $0.22$--$0.76$, against at most $0.11$ on
the other three), with four of the sixteen values sitting on the clip bounds
$[0.05,1]$; and the field is proxied from residuals on a \emph{single}
stream, not the idealized multi-window drift response $\Psi$ of the
theory. We therefore report them as stream-level proxies and separately
inspect the rough-field bound $\alpha_s\le 1$. On the fifteen label-free
benchmarks we also measured the readout-field smoothness
$\alpha(\eta_0)$, the variogram exponent of $\eta_0(x)=\sigma(z_0(x))$,
where $z_0$ is the logit of the frozen readout $h_0\circ\phi$ and
$\sigma$ the logistic function; this is a
label-free \emph{proxy} for $\alpha_s$. Because no result here justifies
substituting it for drift smoothness, it is not used in the main table's
threshold calculation.

\emph{(iv) Margin exponent $\gamma$.} On the TabReD streams $\gamma$ is fit
from the same residual field (the same single-stream proxy qualification applies). The
measured values sit near the generic transversal value $\gamma\approx 1$,
ranging from $0.58$ (\textsc{homesite-insurance}) to $3.39$
(\textsc{homecredit-default}). At that value the drift crosses the decision
boundary with a density. On the related benchmarks $\gamma$ is not
measurable; $\gamma=1$ is an explicit sensitivity assumption, not an
estimate.

\paragraph{Probe and oracle configurations.}
All measurements of Sections~\ref{sec:anatomy}--\ref{sec:methodology} run on precomputed predictions and
$192$-dimensional embeddings of the frozen backbone. The regional-offset
probe fits a $k$-means codebook of $1$, $8$, $32$, or $128$ regions and
per-region offsets (mean residual for regression; a one-dimensional
Newton logit offset for classification), regions with fewer than $30$
fitting points shrinking fully to the global offset; the out-of-time arm
fits codebook and offsets on the first half of the stream and evaluates
on the second.
The C4 replication (Section~\ref{sec:channels}) repeats the out-of-time arm at cut
fractions $0.4$, $0.5$, $0.6$, and $0.7$, with $k$-means seeds $0$--$5$
(four initialization restarts per fit) and codebook sizes $8$, $32$, and $128$ regions; the quoted
ranges are the per-cut seed means at the best codebook size, selected on
the same out-of-time evaluation ($8$ regions on \textsc{homesite-insurance}:
$+1.33$, $+1.53$, $+1.93$, $+1.76$; $128$ regions on \textsc{ecom-offers}:
$+1.30$, $+1.44$, $+1.79$, $+1.80$; seed standard errors at most $0.007$ on
\textsc{homesite-insurance} and $0.12$ on \textsc{ecom-offers}).
At the other codebook sizes every cut on both streams remains positive
(\textsc{ecom-offers} with $8$ regions: $+0.56$ to $+0.79$).
Signal ceilings use leave-one-out $k$-NN regression of
residuals on embeddings ($k=20$) in three grades: batch
(time-symmetric), causally growing buffer, and strict out-of-time split.
The temporal-gap stress test on \textsc{homesite-insurance} truncates the
strict split's fitting prefix $G$ rows before the cut, so no retrieved
neighbor lies within $G$ stream positions of the evaluation half;
$G=1{,}000$ (the decoupling distance also used for the variogram) gives
$R^2=0.006$, and the archived sweep reads $0.016$, $0.013$, $0.011$,
$0.006$, $0.008$ at $G=0$, $250$, $500$, $1{,}000$, $2{,}000$.
Noise floors use the temporally decoupled variogram (pairs
$\ge1{,}000$ steps apart). The gate probe couples a streaming $k$-NN
corrector ($k=20$, correction scale $0.7$, block size $64$) to a
test-by-betting e-process on per-round log-loss regret (burn-in
$n_{\mathrm{burn}}\ge20$, level $0.05$). The corrector's remaining
constants (a minimum past buffer of $100$ points; for classification,
clipping of the corrected probability to $[10^{-6},1-10^{-6}]$) are
listed with the reconstruction below and are the values in the deposited
code; the e-process monitor itself is not part of the deposit.
The drift-smoothness proxy $\alpha_s$ is half the log-log slope of the
variogram of the $20$-NN-smoothed residual field against embedding
distance, over $15$ log-spaced bins spanning the $2$nd to $60$th percentile
of $20$-NN distances on a $12{,}000$-row subsample (bins with fewer than
$50$ pairs dropped; the resulting exponent clipped to $[0.05,1]$); the two-range check
refits the slope on the lower half of those bins. The margin exponent
$\gamma$ is the log-log slope of the empirical CDF of the smoothed
residual magnitude between its $5$th and $40$th percentiles.

\paragraph{Resolution selection for $\mathrm{WallB}_\rho$.}
The resolution is not a free parameter: it must match the fiber scale at
which the deployed readout pools. The protocol sets $\rho$ to the median
distance to the $k$-th retrieved neighbor under the deployed retrieval
($k=20$) over the window's unlabeled draws, the operating scale of the
corrector's local averaging, and reports $\mathrm{WallB}_\rho$ together
with the half- and double-resolution values $\mathrm{WallB}_{\rho/2}$ and
$\mathrm{WallB}_{2\rho}$ as a monotone sensitivity band (within-cell
estimation as in Section~\ref{sec:attribution-procedure}, Step~1).
Choosing $\rho$ below the retrieval scale reports capacity the readout
cannot use; above it, reachable structure is misattributed to the wall.
On the eight probe streams the resulting $\rho$ (median Euclidean
distance in the $192$-dimensional embedding to the $20$th retrieved
neighbor over the evaluation half, as recorded in the reproduction
deposit) ranges from $1.95$
(\textsc{homecredit-default}) to $3.97$ (\textsc{sberbank-housing}):
\textsc{weather} $2.28$, \textsc{ecom-offers} $2.20$,
\textsc{delivery-eta} $2.53$, \textsc{maps-routing} $2.53$,
\textsc{cooking-time} $2.79$, \textsc{homesite-insurance} $2.94$.

\paragraph{Worked attribution (Table~\ref{tab:shares}).}
The three-way attribution instantiates Section~\ref{sec:attribution-procedure}
with both scales out of sample: fit on the first half of the stream,
evaluate on the second, $k$-NN residual means at the deployed scale
($k=20$) and at the finest supported scale ($m=5$).
$\mathrm{WallB}_\rho=\max(0,R^2_{m=5}-R^2_{k=20})$; the fine arm's
variance penalty biases the difference downward, so a positive reading is
unlikely to be noise, at the price that a true sub-$\rho$ signal smaller
than that penalty reads as zero (the estimate is a detection, not a
bound). The worked example computes the strict-split analogue of the
protocol resolution ($\rho$ from the fixed fit-half retrieval rather than
the growing deployed buffer) and reports the single operating scale; the
$\rho/2$, $2\rho$ band of the selection protocol applies to the
within-cell estimator, not to this detection form. The $k=20$ column is
checked against the archived strict-split ceilings
behind Table~\ref{tab:channels} (numeric where printed, sign elsewhere)
before the shares are computed.

\paragraph{Screening-rule evaluation (Table~\ref{tab:certificate}).}
The eight-stream evaluation instantiates the two label-free
gates of Section~\ref{sec:certificate} with the rule fixed in advance of the
outcomes. The floor is the calibrated estimator of that section: Bernoulli
$\mathbb E[p(1-p)]$ over the frozen predictions for classification, and for
regression the short-range decoupled semivariance: the mean of
$\tfrac12(r_i-r_j)^2$ over the smallest $1\%$ (and at least $50$) of the
\emph{available} embedding distances among $400{,}000$ uniformly sampled
index pairs at least $1{,}000$ stream steps apart. This is not
the extrapolated nugget of Section~\ref{sec:twins}, which reads the
variogram down to zero distance where \textsc{weather} has no support.
The rule uses the non-extrapolated short-range value deliberately.
On \textsc{weather}, the one stream for which the extrapolated range is
reported, the short-range value is $0.89$ against $0.13$--$0.62$, so it
is the conservative reading for a gate that has to justify adapting. Both floors are divided by the residual
variance of the stream. The local-determinacy statistic is the out-of-sample $R^2$ of the
strictly-past $k$-NN estimate of the residual field under the deployed
retrieval (growing buffer, $k=20$, unweighted neighbor mean, block $64$,
first $100$ points reserved as buffer); every stream is scored in full, so the
evaluated count is $n_{\mathrm{test}}-100$ throughout (the longest stream,
\textsc{maps-routing}, has $n_{\mathrm{test}}=59{,}951$ and evaluates
$59{,}851$ points). The rule fires
only if $1-\text{floor}\ge\tau_{\mathrm{room}}$ and
$R^2_{\mathrm{causal}}>\tau_{\mathrm{sig}}(1-\text{floor})$, at the
operating point $\tau_{\mathrm{room}}=0.10$, $\tau_{\mathrm{sig}}=0$. The
reported verdicts hold for every
$(\tau_{\mathrm{room}},\tau_{\mathrm{sig}})$ in
$[0,0.108)\times[0,0.0616)$, the binding streams being \textsc{weather}
(room $0.108$) and \textsc{sberbank-housing}
($R^2/\text{room}=0.0616$) respectively; both thresholds are swept on a coarse
grid in the archived output, and the exact edges are computed there from
the archived $R^2$ and room columns. Decision-flip fractions, freeze
prices, and post-clip wall widths are reported over an assumed-budget grid
$\beta\in\{0.01,0.05,0.10,0.25\}$, since $\beta$ is a prior rather than a
measurement. The archived output also records, per
classification stream, the variogram floor used for the calibration
contrast and the prediction mean and maximum behind the flip-fraction
statement.

\paragraph{Reconstruction and bootstrap (Table~\ref{tab:bootstrap}).}
The reconstruction reimplements the archived streaming $k$-NN operator
(the corrector re-implemented in the reproduction deposit; past-only
buffer, $k=20$
Euclidean neighbors, \emph{unweighted} residual mean, correction
$p-0.7\,\hat r$ applied in the prediction space and clipped to
$[10^{-6},1-10^{-6}]$ for classification, block $64$, first $100$ points
reserved), applies it to the single precomputed prior of each probe
stream, and scores
\% RMSE reduction (regression) or AUC points (classification) against the
frozen arm. The streaming pass is run once on the true stream order; the
bootstrap resamples evaluation positions in moving blocks of length
$4\lceil n_{\mathrm{test}}^{1/3}\rceil$ ($B=1000$), scoring both arms on the same
resample so the interval is for the paired difference, and the reported
verdicts are unchanged at block multipliers $2$ and $8$. The
metric-resolution column applies the identical resampling to the frozen
arm alone. Deployed gains are \emph{not} recomputed here: the tuned
system is a seven-prior ensemble whose per-sample outputs are not
redistributed, and the reconstruction is offered as an auditable point
estimate of what the local channel alone delivers.

\paragraph{Threshold check.}
The proven sufficient branch for $\sqrt N$ regularity requires
$\gamma>\gamma^\star=d_0/(2\alpha_s)$, equivalently
\[
2\,\alpha_s\,\gamma \;>\; d_0 .
\]
The strict reverse inequality is the lower branch, conditional on gap
(2c) at $\gamma=0$ and open for $0<\gamma<\gamma^\star$
(Section~\ref{sec:dichotomy}); the equality case is unresolved by the
current logarithmic-rate proof. Under
the sensitivity inputs $(\alpha_s,\gamma)=(1,1)$, the check reduces to
$d_0>2$. This is useful for stress testing but cannot turn $d_0$ alone
into a data-set-specific phase measurement. Below threshold, the candidate
upper-rate exponent is $a(1+\gamma)$ with
$a=\alpha_s/(2\alpha_s+d_0)$; its matching scope is stated precisely in
Appendix~\ref{app:H} and Appendix~I of \OAname{}.

\subsection{The proxy and sensitivity display for the twenty-three data sets}
\label{app:J-table}

The full display is Figure~\ref{fig:map} in the main text, with its
inputs in Table~\ref{tab:J-map}
(Section~\ref{sec:map}). Its filled points are stream-level proxies; its
open points are favorable sensitivity scenarios. Both sets lie below the
threshold for their stated inputs, but only the former use observed drift
labels, and neither set carries uncertainty intervals.

\paragraph{How far below.}
We record, for each TabReD stream, the smoothness $\alpha_s$ it
would need for the $\sqrt N$ regime against the measured upper bound:
\textsc{sberbank-housing} needs $\alpha_s>1.42$ (measured
$\sim 0.77$), \textsc{ecom-offers} $>1.90$ ($\sim 0.14$),
\textsc{homecredit-default} $>2.03$ ($\sim 0.50$), and the remaining
five streams need $>5.44$--$9.22$ against measured values of
$0.05$--$0.95$. Even the two lowest-$d_0$ streams, which are exactly
the data sets closest to the boundary, fall short; this is
consistent with \textsc{sberbank-housing} being the leading empirical example of
the earlier probes (largest historical gain, $+4.6\%$;
Table~\ref{tab:bootstrap}, historical column; $+3.80\%$ in the
reconstruction). The
plug-in upper-rate exponent $a(1+\gamma)$ is small throughout:
$0.007$--$0.35$, mostly below $0.15$, with extremes
$0.007$ (\textsc{homesite-insurance}) and $0.347$
(\textsc{sberbank-housing}; from unrounded inputs, the printed inputs
$d_0=3.01$, $\alpha_s=0.77$, $\gamma=1.06$ give $0.349$). These values quantify the implication of the
proxy inputs, not a minimax rate measured from the streams.
Appendix~I of \OAname{} proves the growing-$K$ negative only conditionally at
$\gamma=0$; it does not justify a blanket claim that additional windows
never help these data sets.

\paragraph{If the rough-field bound fails.}
The bound $\alpha_s\le 1$ is an assumption (Section~\ref{sec:setup}), and
the verdicts above depend on it. For $\alpha_s>1$ the threshold
$\gamma^\star=d_0/(2\alpha_s)$ falls, and the required values just listed
are exactly the crossing points under the printed $d_0$ and $\gamma$
inputs: \textsc{sberbank-housing} would enter the sufficient branch at
$\alpha_s>1.42$, \textsc{ecom-offers} at $>1.90$,
\textsc{homecredit-default} at $>2.03$, and the remaining five streams
only at $>5.44$--$9.22$. Because the variogram estimator saturates at
slope $2$, it cannot detect $\alpha_s>1$, so such a crossing would be
invisible to the proxy of item~(iii): the map can place a stream below
threshold under the bound, but it cannot rule out a smoother drift field
that places the two lowest-$d_0$ streams above it.

\subsection{The deduplication correction}
\label{app:J-dedup}

\paragraph{Mechanism.}
Duplicate rows place many points at first-neighbor distance
$r_1\approx 0$ and bias both intrinsic-dimension estimators downward
(the mechanism is described in Section~\ref{sec:dedup}); the artifact is
a sibling of the temporal-twin variogram-nugget artifact of
Section~\ref{sec:twins}, both being duplicate-driven. The consequence for
protocol is stated in Section~\ref{sec:dedup} and enforced here:
intrinsic dimension on FM embeddings must be measured after
deduplication. Numerically, \textsc{unsw\_nb15} moves from $d_0=0.87$
to $3.58$ and \textsc{acs\_employment} from $1.62$ to $3.11$
(Table~\ref{tab:J-dedup}).

\paragraph{The two artifact cases.}
The raw readings placed two data sets in or
near the $\sqrt N$-learnable regime. Both readings were duplicate
inflation.

\textsc{unsw\_nb15} is $58.8\%$ unique ($41\%$ duplicate flows).
The raw MLE gives $d_0=0.87$; after deduplication
$d_0=3.58$, giving the favorable-scenario threshold $d_0/2=1.79$. The raw value is
itself estimator-unstable: a TwoNN pass on the same data gives
$2.93$. On low-uniqueness data the raw numbers are not reproducible
across estimators, which is part of the diagnosis.

\textsc{acs\_employment} is $84.3\%$ unique. On the raw data, TwoNN
gives $0.30$ (unstable, with $4.6\%$ zero-$1$NN pairs), the MLE gives
$1.62$, the PCA
participation ratio is $6.2$ ($11$ PCs for $90\%$ variance), and the
readout field is smooth ($\alpha(\eta_0)\approx 1.12$, a finite-sample
variogram slope
$\sim 2.2$, above the limiting value $2$; $89\%$ of $\eta_0$ in the interior $[0.02,0.98]$). On that
evidence alone the data set could be above threshold: at $d_0\approx 1.6$ and $\gamma\approx 1$
the requirement is $\alpha_s>0.81$, plausibly met if drift smoothness
tracks readout smoothness. The definitive $\alpha_s$ (which needs the
true labels) is not measured, since the effort would be disproportionate for a
boundary estimate. Deduplication closes the question from
the other side under the same sensitivity inputs:
$d_0=1.62\to 3.11$, so $d_0/2=1.56>1$. The apparent explanation of the
raw reading (a low-complexity binary
target concentrating the embedding on a near-$1$-D discriminative
direction) rests on the inflated estimate and does not remain valid after the correction.

Table~\ref{tab:J-dedup} gives the full re-measured block, including
\textsc{acs\_unemployment} ($93.1\%$ unique, $2.24\to 3.49$, raw reading
borderline) and six higher-uniqueness data sets whose values move
little. The lower block extends the identical protocol to the six
high-$d_0$ benchmarks, on which the correction is inert: all six
are $99.97$--$100\%$ unique at four decimals (five carry literally no
duplicate rows), so for them the raw-versus-deduplicated difference is
re-subsampling noise, which bounds the deduplication effect itself by $0.10$
(against the same-run raw re-estimate archived with the verification code; against the raw values of the measurement record the shift is at most
$0.04$).

\begin{table}[t]
\centering
\footnotesize
\caption{The deduplication correction: unique fraction at four-decimal
precision and raw versus deduplicated MLE $d_0$. Upper block: the nine
re-measured related benchmarks; two (\textsc{ieee\_fraud},
\textsc{acs\_income}) are already duplicate-free at this precision.
Lower block: the six
high-$d_0$ benchmarks, essentially duplicate-free, where the shift is a
noise bound rather than a correction. The final column computes
$d_0/2$, the threshold under the sensitivity value $\alpha_s=1$; it is not
a measured drift threshold.}
\label{tab:J-dedup}
\begin{tabular}{lrrrr}
\toprule
data set & unique\% & $d_0$ raw (MLE) & $d_0$ dedup & $d_0/2$ scenario \\
\midrule
\textsc{unsw\_nb15}          & 58.8\%  & 0.87 & 3.58 & 1.79 \\
\textsc{acs\_employment}     & 84.3\%  & 1.62 & 3.11 & 1.56 \\
\textsc{acs\_unemployment}   & 93.1\%  & 2.24 & 3.49 & 1.75 \\
\textsc{cicids}              & 96.0\%  & 2.86 & 2.86 & 1.43 \\
\textsc{acs\_publiccoverage} & 98.2\%  & 4.21 & 4.78 & 2.39 \\
\textsc{ieee\_fraud}         & 100\%   & 4.63 & 4.64 & 2.32 \\
\textsc{bike\_sharing}       & 99.9\%  & 4.88 & 4.90 & 2.45 \\
\textsc{credit\_card\_fraud} & 97.8\%  & 5.09 & 5.32 & 2.66 \\
\textsc{acs\_income}         & 100\%   & 6.47 & 6.49 & 3.25 \\
\midrule
\textsc{diabetes\_readmission} & 100\%   & 8.77  & 8.79  & 4.40 \\
\textsc{brfss\_diabetes}       & 99.97\% & 9.46  & 9.46  & 4.73 \\
\textsc{sepsis}                & 100\%   & 12.75 & 12.75 & 6.38 \\
\textsc{baf}                   & 100\%   & 12.77 & 12.81 & 6.41 \\
\textsc{lending\_club}         & 100\%   & 13.57 & 13.55 & 6.78 \\
\textsc{mimic\_iv\_mortality}  & 100\%   & 14.02 & 14.05 & 7.03 \\
\bottomrule
\end{tabular}
\end{table}

\paragraph{The TabReD-8 streams pass the same check.}
Applying the identical protocol (unique rows at four decimals,
recompute MLE $d_0$ on the unique subset, $12{,}000$-row subsample) to the eight
TabReD streams gives Table~\ref{tab:J-tabred-dedup}: all are
$99.8$--$100\%$ unique, every shift is at most $0.08$ in absolute value, and the
minimum deduplicated $d_0$ is $2.94$. Thus the dimension input for the
TabReD sensitivity check is duplicate-robust; only the low-uniqueness
related data sets (network flows, census microdata) were materially
affected. Note that this check uses
the MLE estimator, whose raw point values differ from the TwoNN
values of Table~\ref{tab:J-map} (for example \textsc{delivery-eta}
$14.82$ MLE against $16.63$ TwoNN). That is an estimator difference, not
a data difference; the quantity the check certifies is the
\emph{deduplication shift}, which is at most $0.08$ in absolute value everywhere.

\begin{table}[t]
\centering
\footnotesize
\caption{Deduplication robustness check on the eight TabReD streams (from the
measurement record). All streams are essentially duplicate-free, so their
intrinsic-dimension inputs are insensitive to this correction. Columns
are rounded independently, so a printed raw/dedup pair can differ from
the printed shift in the last digit.}
\label{tab:J-tabred-dedup}
\begin{tabular}{lrrrr}
\toprule
data set & unique\% & $d_0$ raw & $d_0$ dedup & shift \\
\midrule
\textsc{ecom-offers}        & 99.9\% & 2.95  & 2.94  & $0.00$ \\
\textsc{sberbank-housing}   & 99.9\% & 3.06  & 3.08  & $+0.02$ \\
\textsc{weather}            & 99.8\% & 11.06 & 11.13 & $+0.07$ \\
\textsc{cooking-time}       & 100\%  & 10.64 & 10.57 & $-0.07$ \\
\textsc{homesite-insurance} & 100\%  & 10.40 & 10.34 & $-0.06$ \\
\textsc{homecredit-default} & 100\%  & 12.23 & 12.31 & $+0.08$ \\
\textsc{delivery-eta}       & 100\%  & 14.82 & 14.85 & $+0.03$ \\
\textsc{maps-routing}       & 100\%  & 15.62 & 15.58 & $-0.04$ \\
\bottomrule
\end{tabular}
\end{table}

\paragraph{Corrected conclusion.}
Deduplication removes the apparent low-dimensional examples: the lowest
corrected estimates in the re-measured block are \textsc{cicids} at $2.86$ and
\textsc{acs\_employment} at $3.11$, while the minimum TabReD MLE is $2.94$.
Consequently none of these rows enters the regular branch under the
specific sensitivity values $(\alpha_s,\gamma)=(1,1)$. This does not show
that their true drift fields are below threshold, because neither input is
measured for the related benchmarks. It does rule out reading
\textsc{acs\_employment} or \textsc{unsw\_nb15} as an observed
$\sqrt N$-learnable demonstration data set.

\subsection{Reading the display}
\label{app:J-atlas}

\paragraph{Raw versus deduplicated readings.}
The raw and deduplicated columns of Table~\ref{tab:J-dedup} differ
materially on three data sets only (\textsc{unsw\_nb15},
\textsc{acs\_employment}, \textsc{acs\_unemployment}), all of low
uniqueness; on the twelve others the two readings agree to within
subsample noise. The raw readings placed the first two in or near the
regular branch under the sensitivity inputs and the third at its
border; the deduplicated readings place all three below it. Only the
deduplicated column enters Table~\ref{tab:J-map} and
Figure~\ref{fig:map}.

\paragraph{The dimension input drives the threshold spread.}
Under the common sensitivity value $\alpha_s=1$, the threshold input
$d_0/2$ ranges from $1.43$ (\textsc{cicids}) to $7.03$
(\textsc{mimic\_iv\_mortality}). This spread is driven by intrinsic
dimension. Interpreting it as wall height would additionally require
measuring the drift smoothness, so the display supports comparative stress
testing rather than a ranked empirical ordering of walls.

\paragraph{The ACS family gradient.}
The cleanest cross-data-set pattern is the controlled sweep within the
ACS census family (same data, same representation, same shift,
varying only the prediction target). The deduplicated $d_0$ rises with
target complexity: employment $3.11$, unemployment $3.49$, public
coverage $4.78$, income $6.49$. The ordering
(income $>$ coverage $>$ unemployment $>$ employment) already held in the raw
(pre-deduplication) values and persists after the correction; what the correction
removed is only the reading that the low end of the sweep crosses the
boundary. Under a common $(\alpha_s,\gamma)$ scenario this ordering also
orders the threshold inputs. It is a descriptive pattern, not an
independent test of the threshold of Theorem~\ref{thm:dichotomy}.

\paragraph{Reliability of the inputs.}
The deduplicated $d_0$ values and their cross-data-set ordering are the
strongest part of the record, although estimator sensitivity remains. The
TabReD $\alpha_s$ and $\gamma$ values are noisy proxies from a single
residual field, with two-range slope disagreement on five of eight streams and
upward bias from $k$-NN smoothing. For the related benchmarks, both drift
quantities are unobserved. A phase-validation study would need repeated
labeled windows, uncertainty intervals for all three inputs, and a
pre-specified rule for propagating them through
$\gamma^\star=d_0/(2\alpha_s)$.